\documentclass[11pt]{article}

\usepackage{graphicx} 
\usepackage{enumitem}
\usepackage{amsmath}
\usepackage{fullpage}
\usepackage{amsthm}
\usepackage{amsfonts}
\usepackage{mathrsfs} 
\usepackage{breqn}
\usepackage{bbm}
\usepackage{algorithm}
\usepackage[noend]{algpseudocode}
\usepackage{todonotes}
\usepackage{thmtools} 
\usepackage{hyperref}
\usepackage{subcaption}
\usetikzlibrary{calc,decorations.pathreplacing}

\hypersetup{
    colorlinks=true,
    linkcolor=violet,
    filecolor=magenta,      
    urlcolor=cyan,
    citecolor=blue,
    pdffitwindow=true,
}\usepackage[capitalize, nameinlink]{cleveref}
\theoremstyle{plain}
\newtheorem{theorem}{Theorem}[subsection]
\newtheorem{lemma}[theorem]{Lemma}

\newtheorem{claim}[theorem]{Claim}
\newtheorem{remark}[theorem]{Remark}

\newtheorem{definition}[theorem]{Definition}
\newtheorem{observation}[theorem]{Observation}
\newtheorem{corollary}[theorem]{Corollary}
\newtheorem{fact}[theorem]{Fact}
\usepackage{amsmath,xcolor}

\usepackage{amssymb}
\usepackage{forest}
\usepackage{xcolor}

\newcommand{\E}{\mathop{\bf E\/}}
\newcommand{\R}{\ensuremath{\mathbb R}}

\newcommand{\cQ}{\mathcal{Q}}

\newcommand{\calH}{\mathcal{H}}

\newcommand{\calD}{\mathcal{D}}
\newcommand{\calL}{\mathcal{L}}

\newcommand{\calA}{\mathcal{A}}

\newcommand{\calG}{\mathcal{G}}

\newcommand{\scrE}{\mathscr{E}}
\newcommand{\scrF}{\mathscr{F}}
\newcommand{\scrG}{\mathscr{G}}

\newcommand{\LeafEdges}{\mathsf{leaf\mbox{-}edges}}
\newcommand{\AlgEvent}{\mathsf{ALG}}

\newcommand{\Ext}{\textrm{Ext}}

\newcommand{\Lead}{\textrm{Lead}}

\newcommand{\apex}{\textrm{apex}}

\newcommand{\scov}{\mathsf{covers}}

\newcommand{\events}{\mathsf{events}}
\newcommand{\primeevents}{\mathsf{events}'}
\newcommand{\stars}{\mathsf{stars}}

\newcommand{\subtrees}{\mathsf{subtrees}}
\newcommand{\sm}{\mathsf{small}}

\newcommand{\lrg}{\mathsf{lrg}}

\newcommand{\Restr}{\mathsf{Restr}}

\newcommand{\signif}{\mathsf{signif}}
\newcommand{\restart}{\mathsf{restart}}
\newcommand{\shadows}{\mathsf{shadows}}
\newcommand{\trim}{\mathsf{trim}}

\newcommand{\alge}{\textsc{alg}_{\mid \scrE}}

\renewcommand{\P}[1]{{\Pr}\left[#1\right]}
\newcommand{\PP}[2]{{\Pr}_{#1}\left[#2\right]}

\newcommand{\ola}[1]{\textcolor{blue}{Ola: #1}}

\begin{document}

\title{A Strong Linear Programming Relaxation for Weighted Tree Augmentation}

\author{
  Vincent Cohan-Addad\thanks{Google Research.}
  \and
  Marina Drygala\thanks{EPFL, Switzerland.}
  \and
  Nathan Klein\thanks{Boston University. Supported by the NSF CAREER grant CCF-2442250.}
  \and
  Ola Svensson\thanks{EPFL, Switzerland. Supported by the Swiss State Secretariat for Education, Research and Innovation (SERI) under contract number MB22.00054.}
}
\date{}
\maketitle

\begin{abstract}
The Weighted Tree Augmentation Problem (WTAP) is a fundamental network design problem where the goal is to find a minimum-cost set of additional edges (links) to make an input tree 2-edge-connected. While a 2-approximation is standard and the integrality gap of the classic Cut LP relaxation is known to be at least 1.5, achieving approximation factors significantly below 2 has proven challenging. Recent advances of Traub and Zenklusen using local search culminated in a ratio of $1.5+\epsilon$, establishing the state-of-the-art. 
In this work, we present a randomized approximation algorithm for WTAP with an approximation ratio below 1.49. Our approach is based on designing and rounding a strong linear programming relaxation for WTAP which incorporates variables that represent subsets of edges and the links used to cover them, inspired by lift-and-project methods like Sherali-Adams.
\end{abstract}

\section{Introduction}
Network design problems form a cornerstone of combinatorial optimization and theoretical computer science, addressing fundamental questions about constructing efficient and robust networks. Problems which seek the most cost-effective way to enhance the connectivity of an existing network graph are central within combinatorial optimization. Except for some notable examples such as (weighted) connectivity augmentation~\cite{BGJ20,Nut20,Cecchetto10.1145/3406325.3451086,VZ23}, many such problems 
have proven resistant to approximation factors better than 2, a barrier often achieved through various classic techniques \cite{KT9310.1007/3-540-55719-9_85,KV9410.1145/174652.174654, Goe+94, Jai01}. 
Breaking this barrier for fundamental augmentation problems such as 2-ECSS (the 2-edge-connected spanning subgraph problem) or its directed variant remains a significant research challenge.

 This paper focuses on the Weighted Tree Augmentation Problem (WTAP), a classic and well-studied problem within this class. Given a tree $T = (V, E)$, a set of available \emph{links} $L$ (potential edges not in $E$) with associated costs $w: L \mapsto \R_{\ge 0}$, the goal is to find a minimum-cost subset of links $S \subseteq L$ such that the graph $T' = (V, E \cup S)$ is 2-edge-connected. Equivalently, we seek a minimum-cost set $S$ such that for every edge $e \in E$, its removal does not disconnect the graph, meaning $S$ must cover every cut defined by removing a single tree edge.

WTAP is known to be APX-hard~\cite{KKL04}, inheriting hardness from its unweighted special case, the Tree Augmentation Problem (TAP). TAP has seen a long sequence of improvements \cite{Frederickson1981ApproximationAF,KT9310.1007/3-540-55719-9_85,NAGAMOCHI200383,CKKK08, Eve+09, CN13, KN16, Nut17, CG18a, CG18b, Adj18, FGKS18, GKZ18, KN18}
before the current best factor of 1.393 was achieved ~\cite{Cecchetto10.1145/3406325.3451086}. 
For WTAP, a ratio better than 2 was first achieved by the breakthrough result of Traub and Zenklusen \cite{TZ25} in 2021 via a relative greedy approach. Then, using local search, they improved this to $1.5 + \epsilon$ \cite{traub2022better}. 

Despite this progress, we still have a limited understanding of natural LP relaxations for WTAP, and we do not know any bound better than 2 on the integrality gap of the Cut LP, which is the most natural relaxation for WTAP. There is a lower bound of 1.5 on the integrality gap due to Cheriyan, Karloff, Khandekar, and K\"onemann \cite{CKKK08}. In 2018, 
Fiorini, Groß, K\"onemann, and Sanità introduced the Odd Cut LP~\cite{FGKS18}, which is a stronger relaxation for WTAP and is a key starting point for this work. It is still unknown whether this stronger LP has an integrality gap below 2. 

\subsection{Our Results}
In this paper, we revisit the LP relaxation approach and make progress on the complexity of the WTAP problem. We present a novel approximation algorithm for WTAP with an approximation ratio strictly below 1.5, improving upon the aforementioned recent results and showing that a stronger relaxation can bring the approximation ratio below the integrality gap of the Cut LP. Our algorithm is based on rounding a new, stronger linear programming relaxation. Along the way, we provide a different $(1.5+\epsilon)$-approximation algorithm that is different from the one presented in~\cite{traub2022better}.

\begin{restatable}[
    label={thm:main}
]{theorem}{maintheorem}
\label{thm:main}
There is a randomized $1.49$-approximation algorithm for WTAP. Moreover there is a linear programming relaxation for WTAP with integrality gap at most 1.49. 
\end{restatable}
Our theorem relies on a new approach that uses a significantly strengthened Linear Programming (LP) relaxation and a new rounding strategy.
A so-called \textit{event} $\scrE$ describes a specific configuration of a small set of links covering the edges within a localized structure of the tree. Each event is assigned a probability $y(\scrE)$.  The crucial strengthening of our LP then comes from consistency constraints that relate these probabilities across overlapping local structures (subtrees). In particular, the LP enforces that the marginal distributions over the shared links induced by certain overlapping events $\scrE_1,\scrE_2$ must agree. This forces a global coherence on the fractional solution, capturing interactions that the standard Odd Cut LP misses, akin to initial rounds of a lift-and-project hierarchy (like Sherali-Adams). This is also reminiscent of other work
on network design problems that leverage configuration LPs, see e.g.~\cite{BGRS13} or, even more related, \cite{Adj18} which uses such an LP for WTAP. However our use of consistency between different configurations is novel and, we believe, provides a recipe for making progress on other problems in combinatorial optimization. 

From there we derive two complementary rounding strategies, whose average yields a $(1.5+\epsilon)$-approximation. Roughly speaking, we first show we can reduce the problem to a structured type of instance somewhat reminiscent of algorithms for TAP or WTAP in the bounded cost setting (see \cite{FGKS18} or the splitting procedure of \cite{Adj18}). On such instances, with cross-links $X$ and in-links\footnote{An in-link $\ell=(u,v)$ is any link that is not a cross-link or an up-link, so that $u$ is not an ancestor of $v$, $v$ is not an ancestor of $u$, and the lowest common ancestor of $u$ and $v$ is not the root.} $I$ (ignoring up-links for notational ease, which we get a factor of 1 for), we show:
(1) an algorithm based purely on the Odd Cut LP and an integrality theorem of \cite{FGKS18} with cost at most $c(X) + 2 c(I)$, and (2) an algorithm that uses the strong LP's structure for a cost of at most $2c(X) + c(I)$. We also lose an $\epsilon$ factor which may be chosen arbitrarily small, with the running time depending on the choice of $\epsilon$. By running each algorithm with probability $\frac{1}{2}$ we obtain a $1.5+\epsilon$ approximation.

 The improvement below 1.5 comes from refining the second algorithm. After the initial independent rounding of subtrees, we introduce a ``clean-up phase.'' This phase identifies certain cross-links chosen by the rounding that are effectively ``redundant'' or ``over-covered'' by links sampled from other subtrees. These extra cross-links can be removed, sometimes requiring the addition of cheaper up-links (internal to subtrees) to maintain feasibility, leading to an expected cost strictly less than $2c(X) + c(I)$, and
 thus an overall approximation strictly less than $1.5$ when averaged with the first algorithm.

 \subsection{Organization of the Paper}
The core of our contribution is a novel, strong Linear Programming relaxation (the \textit{Strong LP}) and a randomized rounding algorithm based on it. However, rounding the Strong LP directly is challenging. We therefore introduce an intermediate formulation, the \textit{Structured LP} and the associated Structural Fractional Solution, which facilitates the analysis and the rounding process. Our presentation is organized to highlight our main contributions in a less notation heavy setting before detailing the Strong LP and the relationship between the Strong LP and Structured Fractional Solutions.

After introducing basic notation and definitions in \cref{sec:prelims}, the structure of the paper is as follows:
\begin{enumerate}
    \item In \cref{sec:structured}, we define the Structured LP (which shares several ideas with the more complex Strong LP) and present a randomized rounding algorithm for Structured Fractional Solutions, achieving a $1.5$-approximation relative to the cost of such solutions.
    \item In \cref{sec:cleanup}, we refine this algorithm with a ``clean-up phase,'' improving the approximation ratio below $1.49$ (relative to the Structured Fractional Solution).
    \item In \cref{sec:strongLP}, we formally introduce the Strong LP, which serves as our primary relaxation for WTAP.
    \item Finally, in \cref{sec:red_always_small}, we establish the connection between the Strong LP and Structured Fractional Solutions. We prove that it is sufficient to round Structured Fractional Solutions up to a loss in approximation ratio of $1+\epsilon$ for any $\epsilon > 0$, demonstrating that our algorithm achieves an overall approximation ratio of $1.49$ for WTAP.
\end{enumerate}

\section{Preliminaries}
\label{sec:prelims}
\subsection{Notation}

As input, we are given a WTAP instance with tree $T$, root $r$, and link set $L$ with cost function $c: L \to \R_{\ge 0}$ so that $c(\ell)$ is the weight of a link $\ell$. 

\paragraph{Ancestry Relations and Subtrees.} We use descendant and ancestor with respect to vertices of the tree in the standard way. A \textit{child} of a vertex is a descendant sharing an edge with it in $T$. If $u$ is the child of $v$, we call $v$ the \textit{parent} of $u$. We use $T_v$ to denote the subtree rooted at $v$. If $v$ has a parent $u$, we use $T_v^+$ to denote $T_v$ together with the edge $\{u,v\}$.

\hypertarget{treeedgesandcuts}{\paragraph{Tree Edges and Cuts.} We will sometimes use $uv$ to denote the edge $\{u,v\} \in E$, where we will always assume that $u$ is the parent of $v$. The \textit{parent edge} of a vertex $v$ with parent $u$ is the edge $uv \in E$. We let $e_v$ be the parent edge of vertex $v$. For convenience, we define a dummy edge $e_r$ we call the parent edge of the root that is not covered by any links. For a set of vertices $S \subseteq V$, we use $\delta_E(S)$ to denote the set of tree edges crossing the cut $(S, \overline{S})$. Similarly, we use $\delta_L(S)$ to denote the set of links crossing the cut $(S,\overline{S})$. For a vertex $v$, we will use $E(v)$ as a shorthand for $\delta_E(\{v\})$ to denote the set of edges incident to $v$.}

\paragraph{Link Types.} A link $\ell=\{u,v\}$ is associated with the unique $u-v$ path in the tree $T$. We say $\ell$ \textit{covers} the set of edges in this unique $u-v$ path. The \textit{apex} of $\ell$, denoted $\apex(\ell)$, is the vertex of $T$ on this path closest to the root, or equivalently, the lowest common ancestor of $u$ and $v$ in $T$.
\begin{itemize}
\item $\ell$ is called an \textit{up-link} if $u$ or $v$ is the apex of $\ell$, or equivalently if $u$ is an ancestor of $v$ or vice versa. We will use $L_{UP}$ to denote the set of up-links. 
\item $\ell$ is called a \textit{cross-link} if the apex of $\ell$ is $r$.
\item Otherwise, $\ell$ is called an \textit{in-link}. 
\end{itemize}

\paragraph{Leading Edges of Links.}
 We define an edge $e \in T$ to be a \emph{leading edge} of $\ell$ if it is the first edge on the path from $\apex(\ell)$ to endpoint $u$ or $v$ (whichever is not the apex). Let $\Lead(\ell)$ denote the set of leading edges of $\ell$. We note that $|\Lead(\ell)| = 1$ if $\ell$ is an up-link, and $|\Lead(\ell)| = 2$ otherwise.

\paragraph{Shadows and Splitting.} A \textit{shadow} of a link $\ell$ is a link which covers a strict subset of the edges that $\ell$ covers. \hypertarget{shadowcomplete}{Similar to other recent works} on WTAP such as \cite{FGKS18}, given an input instance, we make it \textit{shadow complete} by adding all shadows of every link.
with each shadow having the price of the original link. If the resulting graph has multiple links
with the same endpoints, we keep the cheapest one. \hypertarget{splitting}{A common operation is to \textit{split} a cross-link $\ell=\{u,v\}$ or in-link into two of its shadows, both of which are up-links. This involves deleting $\ell$ and replacing it with two up-links $\ell_u =\{u,\apex(\ell)\}$ and $\ell_v=\{v,\apex(\ell)\}$. Note that the union of the edges covered by $\ell_u$ and $\ell_v$ is the set of edges covered by $\ell$.} Often, this operation is done with respect to a link and a fractional vector $x$ so that $0 \le x(\ell) \le 1$. In this case, one assigns $x(\ell_u) = x(\ell_v) = x(\ell)$ and then sets $x(\ell) = 0$. 

\paragraph{Link Coverage.} We use $L_e$ to denote the set of links covering edge $e$. For a vertex $v$, we will also use $L_v$ to denote $L_{e_v}$. For a set of edges $F \subseteq E$, we use $L(F) = \bigcup_{e \in F} L_e$ to denote the set of links that cover any edge in $F$. 
For a link $\ell \in L$, we will use $P_\ell$ to denote the path in $T$ containing all edges $e$ satisfying $\ell \in L_e$.

\hypertarget{smallcov}{\paragraph{Small Link Sets.}} For $F \subseteq E$, we say a collection of links $L$ is \textit{small on $F$} if for every $e \in F$, $|L \cap L_e| \le \rho$ for a fixed constant $\rho \in \mathbb{Z}_{\ge 0}$. We use $\scov(F) \subseteq 2^{L(F)}$ to denote the collection of link sets that are small on $F$ and cover $F$, i.e., for which $1 \le |L \cap L_e| \le \rho$ for all $e \in F$.

\subsection{Two Linear Programs for WTAP}
\label{sec:prelimLPs}

\paragraph{Notation for LPs.} For a vector $x: L \to \R_{\ge 0}$,  we let $x(F) = \sum_{\ell \in F} x(\ell)$ for $F \subseteq L$. We also use $c(F) = \sum_{\ell \in F} c(\ell)$. Finally, we will use $c(x(F))$ (and sometimes just $c(x)$ when $F$ is clear from context) to denote $\sum_{\ell \in F} c(\ell)x(\ell)$. 

\paragraph{The Cut LP.} The classic linear programming relaxation for WTAP is the following, often known as the \hypertarget{cutlp}{Cut LP}:
\begin{align*}
    \min &\sum_{\ell \in L} c(\ell) x(\ell)\\
    \text{s.t.} & \quad x(L_e) \ge 1 && \forall e \in E \\
    & \quad x(\ell) \ge 0 && \forall \ell \in L 
\end{align*}
It is well known that the Cut LP has integral extreme points when $L$ contains only up-links. Therefore, a standard result is that the Cut LP has an integrality gap of at most 2. To prove this, one first \hyperlink{splitting}{splits} all cross-links and in-links into up-links, increasing the cost of $x$ by a factor of at most 2. Then, using the integrality of the Cut LP on the resulting instance, one can find an integral vertex of cost at most $2c(x)$. It is also worth noting here that the best known upper bound on the integrality gap of the cut LP is 2 for WTAP, while the best lower bound is 1.5 \cite{CKKK08}. 

\paragraph{The Odd Cut LP.} \hypertarget{oddcutlp}{The Cut LP can be strengthened in the following way:}
\begin{align*}
    \min & \quad \sum_{\ell \in L} c(\ell) x(\ell)\\
    \text{s.t.} &\quad  x(\delta_L(S)) + \sum_{e \in \delta_E(S)} x(L_e) \ge |\delta_E(S)| + 1 & \forall S \subseteq V, |\delta_E(S)| \text{ odd} \\
    & \quad x(\ell) \geq 0 & \forall e \in L
\end{align*}

This LP, although exponentially large, has a polynomial time separation oracle and therefore can be solved in polynomial time. This LP was introduced by Fiorini, Groß, Könemann, and Sanità \cite{FGKS18}. Moreover, they proved the following:
\begin{theorem}[\cite{FGKS18}]\label{thm:FGKS}
For trees containing only cross-links and up-links, the Odd Cut LP has integral vertices.
\end{theorem} 
Now, using the Odd Cut LP, we only need to split the in-links into up-links to reach an instance with an integrality gap of 1 (instead of both the cross-links and in-links). Also note that the LP does not depend on the choice of the root. So, after solving the LP, one can choose the root before splitting and applying their theorem to maximize the LP value of the set of links that do not need to be split.

The Odd Cut LP and \cref{thm:FGKS} will make up one of the key algorithms in this work. We will also use that the LP is feasible for any choice of root.


\section{Rounding Structured Fractional Solutions}\label{sec:structured}

This section presents our main rounding algorithm.  The core novel contribution of our paper is the combination of this algorithm,  its extension detailed in \cref{sec:cleanup} and the Strong LP relaxation introduced in \cref{sec:strongLP}. 

We first introduce the Structured LP and  Structured Fractional Solutions. The Structured LP captures several key ideas of the more complex Strong LP, making it an instructive starting point for exposition. Our main rounding algorithm will operate on these Structured Fractional Solutions.

We justify this approach by showing that any solution to the Strong LP can be efficiently converted into a Structured Fractional Solution with nearly the same cost. This reduction, proven formally in \cref{sec:red_always_small}, is formalized by the following theorem. It allows us to focus our primary efforts on the subsequent rounding phase.

\begin{restatable}[
        label={thm:reduction_to_structured}
    ]{theorem}{reductiontostructured}
    \label{thm:reduction_to_structure}
    Let $\calA$ be a (randomized) polynomial-time algorithm that, given a Structured Fractional Solution $(x, y, z)$ to a WTAP instance, returns an integral solution of (expected) cost at most $\alpha\cdot c(z(L))$. Then there exists a polynomial-time $(1 + o_\epsilon(1))\alpha$-approximation algorithm for WTAP (with respect to the Strong LP), where $\epsilon > 0$ is a small fixed parameter and we use $o_\epsilon(1)$ to denote a quantity that tends to $0$ as $\epsilon \to 0$.
\end{restatable}

The proof of \cref{thm:reduction_to_structure}, while technically involved and requiring detailed calculations, is conceptually less novel than our new LP and the rounding procedure. We have therefore deferred it to \cref{sec:red_always_small}.  We now proceed to define the Structured LP and Structured Fractional Solutions. We then present the details of our main rounding algorithm.

\subsection{Structured LP and Structured Fractional Solutions}\label{subsec:structured_LP_struc_frac}

Our approximation algorithm relies on rounding an intermediate structure derived from the Strong LP (introduced in \cref{sec:strongLP}). This intermediate structure, which we call a Structured Fractional Solution, captures the essential local properties of the Strong LP solution while being more amenable to rounding.

The Strong LP, and consequently the Structured LP defined here, draws inspiration from lift-and-project methods such as the Sherali-Adams hierarchy. In these methods, a base LP is strengthened by introducing new variables representing local configurations and enforcing consistency constraints between overlapping configurations.

Similarly, our approach enhances the basic Cut LP by introducing variables that capture how small subsets of edges (which we call "stars") are covered by localized sets of links. The crucial strengthening comes from enforcing consistency across these local views: if two stars overlap, the distributions over the links covering the shared edges must agree. This forces a global coherence that the basic Cut LP misses.

As we define below, the Structured Fractional Solution is a tuple $(x, y, z)$. The components $(x, y)$ capture this refined local consistency and are used in our primary new rounding algorithm (\cref{subsec:structured-LP-rounding}). The vector $z$ is a feasible solution to the standard Odd Cut LP, used in a complementary rounding algorithm based on Fiorini et al.~\cite{FGKS18} (\cref{sec:general-Splitting-Alg}). We emphasize that the Structured LP defining $(x, y)$ is \textbf{not} itself a relaxation of WTAP, as it imposes structural restrictions (like edges being covered by a small, bounded number of links). Instead, we rely on \cref{thm:reduction_to_structure}, which guarantees that a solution to the Strong LP can be transformed into a Structured Fractional Solution with negligible cost increase.

\subsubsection{Structured LP}\label{subsubsec:structuredLP}

The Structured LP is defined relative to a partition of the tree's structure into correlated and uncorrelated components. It is also parameterized by the parameter $\rho \geq 1$. Recall that $\rho$ defines the notion of a small link set: a set $S \subseteq L$ is small on $F \subseteq E$ if $1 \leq |S \cap L_e|\leq \rho$ for $e\in F$. 

\paragraph{Correlated Nodes and Edges.} The Structured LP  is defined with respect to a list of nodes $V_{cor} \subseteq V$ we will call the correlated nodes. For each $v \in V$, let $E^*(v) \subseteq E(v)$ consist of $e_v$ and the edges going from $v$ to its correlated children, i.e., those in $V_{cor}$. We will call edges in $E^*(v) \smallsetminus \{e_v\}$ for some $v \in V$ \textit{correlated} and the remaining edges \textit{uncorrelated}. We will call links that have a correlated leading edge \textit{correlated} (this set is denoted $L_C$) and the remaining links \textit{uncorrelated} (this set is denoted $L_U$).

\paragraph{Stars.} 
Let $\stars$ be the collection of edge sets $F \subseteq E$ so that for some $v \in V$, either $F=\{e_v\}$ or $F$ consists of $E^*(v)$ together with the addition of at most two edges from $E(v)$. 

\paragraph{Events.} An event is defined as a tuple $(E',L')$ where $E' \subseteq E$ and $L' \subseteq L(E)$. Intuitively, an event $(E',L')$ is in indicator that the link set used to cover $E'$ is exactly $L'$. We will sometimes abuse notation and use $(e,L')$ to denote $(\{e\},L')$. Given an event $\scrE$, let $E(\scrE)$ be the edge set in the tuple and $L(\scrE)$ the link set. Let $\events$ be the set of events such that $E(\scrE) \in \stars$. Let $\events(F)$ be the subset of events $(F',L') \in \events$ such that $F'=F$.

\begin{remark}
    The Structured Fractional Solutions produced in the reduction \cref{thm:reduction_to_structured}  is such that the size of $F$ for each $F\in \stars$ and the parameter $\rho$ are all upper bounded by a constant (which depends on $\epsilon$).  The total number of events in $\events$ (and thus variables of the Structured LP) is thus polynomially bounded in the size of the WTAP instance for any fixed $\epsilon$. Strictly speaking, this is not necessary in this section as our goal is to devise a polynomial time algorithm given a Structural Fractional Solution (i.e., its running time is allowed to depend on the size of the solution), but it is helpful for intuition. 
\end{remark}

\paragraph{Conditional Event Sets.} \hypertarget{conditionalEvents}{Given a pair of edge sets $E_1,E_2 \in \stars$ with $E_1 \subseteq E_2$ and an event $\scrE_1 \in \events(E_1)$, let $\events(E_2)_{\mid \scrE_1}$ be the subset of events $\scrE_2 \in \events(E_2)$ for which $L(\scrE_2) \cap L(E_1) = L(\scrE_1)$. In other words, it is the set of events for $E_2$ that agree with $\scrE_1$ on $E_1$.}\\

We can now build a linear program using variables $ x(\ell)$ for the links and $y(\scrE)$ for the events in $\events$. Intuitively, $x(\ell)$ will indicate whether $\ell$ is in the solution. 
Furthermore we will have non-negativity: $x(\ell) \ge 0$ and $y(\scrE) \ge 0$. In addition, the Structured LP enforces the linear constraints~\cref{eq:constr:small_prob_1:marginals,eq:constr:small_prob_1:coverage,eq:constr:small_prob_1:consistency} that we now introduce. 

\paragraph{Marginal Preserving Constraints.}
For each edge $e \in E$ and  link $\ell \in L_e$ , we will have the constraint:
\begin{equation}\label{eq:constr:small_prob_1:marginals}
    \quad x(\ell) = \sum_{\scrE \in \events(e)}y(\scrE)\mathbb{I}[\ell \in L(\scrE)]
\end{equation}

In other words, if we look an edge $e$ and a link $\ell$ covering $e$, the distribution over link sets covering $e$ should use link $\ell$ with probability $x(\ell)$. 

\paragraph{Coverage Constraints.}  
For each $e \in T$
\begin{equation}\label{eq:constr:small_prob_1:coverage}
    \sum_{\scrE \in \events(e)}y(\scrE) = 1
\end{equation}
This constraint says every edge $e$ gives us a distribution over link sets that cover $e$ and are small on $e$. Note this is not necessarily true of an optimal solution and is, as aforementioned, why this Structured LP \textbf{is not a relaxation of the problem.} The definition of the Strong LP in \cref{sec:strongLP} handles this by introducing more complex events for general subtrees that allow tree edges to also be covered by many links. 

\paragraph{Consistency Constraints.} 
These constraints enforce that the distributions over two overlapping stars must agree on their intersection, forcing local consistency. For every pair of edge sets $E_1,E_2 \in \stars$ with $E_1 \subseteq E_2$, 
add the following constraint for all $\scrE_1 \in \events(E_1)$:
\begin{equation}\label{eq:constr:small_prob_1:consistency}
\sum_{\scrE_2 \in \events(E_2)_{\mid \scrE_1}} y(\scrE_2) = y(\scrE_1)
\end{equation}
This set of constraints is essential because it ensures consistency between different distributions over stars. If $E_1$ is a single edge, for example, it ensures that any two stars $E_2,E_3$ which share edge $e$ have the same distribution on $L_e$. 

\subsubsection{Structured Fractional Solutions}


The tuple $(x,y,z)$ with $x,z \in [0,1]^L$ and $y \in [0,1]^{\events}$ is a \textbf{Structured Fractional Solution} for a set of correlated nodes $V_{cor}$ if:
\begin{itemize}
    \item $(x,y)$ is a feasible solution to the Structured LP, i.e., the marginal preserving constraints \eqref{eq:constr:small_prob_1:marginals}, coverage constraints \eqref{eq:constr:small_prob_1:coverage}, and consistency constraints \eqref{eq:constr:small_prob_1:consistency} are satisfied by $(x,y)$. 
    \item $z$ is a feasible solution to the Odd Cut LP, i.e., 
    \begin{equation}\label{eq:constr:small_prob_1:odd_cut}
        z(\delta_L(S)) + \sum_{e \in \delta_E(S)} z(L_e) \ge |\delta_E(S)| + 1 \hspace*{3mm} \forall S \subseteq V, |\delta_E(S)| \text{ odd.}
    \end{equation}
    \item For each $v \in V$ and every edge $e$ from $v$ to an uncorrelated child of $v$, the supports of $x(L_v)$ and $x(L_{e})$ are disjoint and the supports of $z(L_v)$ and $z(L_e)$ are disjoint.
    \item $c(x(L)) \leq c(z(L))$ and 
    \begin{equation}\label{eq:uncorr_cost}
    c(x(L_U \setminus L_{UP})) \leq c(z(L_U))
    \end{equation}
\end{itemize}

We say that $x, y, z$ is a Structured Fractional Solution if there exists subsets a node set $V_{cor} \subseteq V$ with the above properties such that $(x, y, z)$ is a Structured Fractional Solution for $V_{cor}$. Note that $x$ and $z$ will be obtained from a common, strong relaxation for WTAP.


\subsection{Rounding Overview}\label{sec:overview}

In this section we prove the following theorem using the rounding results of the following two subsections:
\begin{theorem}
    Let $x,y, z$ be a Structured Fractional Solution. Then there is a randomized polynomial time algorithm that produces a WTAP solution of expected cost at most $1.5 \cdot \sum_{\ell \in L} c(\ell)z(\ell)$.
    \label{thm:1.5structured}
\end{theorem}
This is a corollary of the following two theorems. In \cref{subsec:structured-LP-rounding}, we show:
\begin{restatable}[label={thm:structured_LP_round}]{theorem}{roundStructured}\label{thm:round_structured}
Given a WTAP instance with a Structured Fractional Solution $(x,y,z)$, there is a randomized polynomial time algorithm that outputs an integral solution with expected cost at most
\[
    c(x(L_C \cup L_{UP})) + 2 \cdot c(x(L_U \setminus L_{UP})).
\]
\end{restatable}

In \cref{sec:general-Splitting-Alg}, we show:
\begin{restatable}[label={thm:round_odd_cut}]{theorem}{roundOddCut}

Given a WTAP instance with a Structured Fractional Solution $(x,y,z)$, there is a polynomial time algorithm that outputs an integral solution with cost at most
\[
    2 \cdot c(z(L_C)) + c(z(L_U)).
\]
\end{restatable}

Thus, combining these gives an approximation guarantee of $1.5 \cdot c(z(L))$ in expectation by running each algorithm with probability $\frac{1}{2}$. Indeed, we then have the expected cost, 
\begin{align*}
    \frac{1}{2} \cdot \Big(
        c(x(L_C \cup L_{UP})) 
        + 2 \cdot c(x(L_U \setminus L_{UP}))
    \Big) \\
    + \frac{1}{2} \cdot \Big(
        2 \cdot c(z(L_C)) + c(z(L_U))
    \Big)
\end{align*}
which can be rewritten as
\begin{align*}
    &\frac{1}{2} \cdot \Big(c(x(L_C \cup L_{UP})) + c(x(L_U\setminus L_{UP}))\\ 
    &+ z(L_C) + z(L_U) \Big) + \frac{1}{2} \Big( c(x(L_U \setminus L_{UP})) + c(z(L_C)) \Big)\,.
\end{align*}
To conclude, we use that $(x,y,z)$ is a Structured Fractional Solution. In particular, that it satisfies $c(x(L)) \leq c(z(L))$ and $c(x(L_U \setminus L_{UP})) \leq c(z(L_U))$ (see \cref{eq:uncorr_cost}). We thus have that each term above is upper bounded by $c(z(L))$ and the total expected cost of our algorithm is at most $1.5 \cdot c(z(L))$. This proves \cref{thm:1.5structured} and, 
by \cref{thm:reduction_to_structure}, this leads to a randomized $1.5+\epsilon$ approximation algorithm. In  \cref{sec:cleanup}, we will improve our approximation guarantee to below 1.49 by improving upon \cref{thm:round_structured}.


\subsection{The Structured LP Rounding Algorithm}\label{subsec:structured-LP-rounding}

We will now discuss the rounding algorithm based on a solution $(x,y)$ to the Structured LP. The intuition for this algorithm is fairly simple. We begin by sampling from the distribution over events at each child of the root independently. These are all uncorrelated children, so we pay a factor of 2 on these links as we have a probability of $x(\ell)$ to sample each link on each of its (up to two) leading edges. Given a link set $L'$ coming into an edge $e$, we then sample link sets for its child edges by choosing from a distribution conditioned on the fact that the solution uses exactly the links $L'$ in the set $L_e$. We do so in such a way so that correlated links and up-links are not sampled independently  and therefore pay a factor of 1, and uncorrelated links pay a factor of 2. In the rest of this subsection we formalize this algorithm and show that the constraints of the Structured LP allow it to support the necessary conditioning operations. 

\paragraph{Conditional Sampling.} Given $E_1,E_2 \in \stars$ with $E_1 \subseteq E_2$ and $\scrE_1 \in \events(E_1)$, to sample from $E_2$ \textit{conditioned} on $\scrE_1$ with $y(\scrE_1) > 0$ means the following: pick exactly one event from \hyperlink{conditionalEvents}{$\events(E_2)_{\mid \scrE_1}$} so that $$\P{\scrE_2 \mid \scrE_1} = \frac{y(\scrE_2)}{y(\scrE_1)}$$
for all $\scrE_2 \in \events(E_2)_{\mid \scrE_1}$. This is a probability distribution due the following simple fact:
\begin{fact}\label{lem:conditioning_well_defined}
    Let $E_1,E_2 \in \stars$ with $E_1 \subseteq E_2$ and $\scrE_1 \in \events(E_1)$ with $y(\scrE_1) > 0$. Then, for all $\scrE_2 \in \events(E_2)_{\mid \scrE_1}$, we have
    $$\sum_{\scrE_2 \in \events(E_2)_{\mid \scrE_1}} \frac{y(\scrE_2)}{y(\scrE_{1})} = 1$$
    So, conditional sampling is well defined.
\end{fact}
\begin{proof}
By consistency, we have
$$\sum_{\scrE_2 \in \events(E_2)_{\mid \scrE_1}} y(\scrE_2) = y(\scrE_1)$$
Therefore,
\[
\sum_{\scrE_2 \in \events(E_2)_{\mid \scrE_1}} \frac{y(\scrE_2)}{y(\scrE_1)} = \frac{y(\scrE_1)}{y(\scrE_1)} = 1. \qedhere 
\]
\end{proof}
As a consequence, the law of total probability behaves as you would expect it to, as the following lemma formalizes:
\begin{lemma}\label{lem:extension_sampling}
Let $E_1,E_2 \in \stars$ with $E_1 \subseteq E_2$. Let $\mu_1$ be a distribution over $\scrE_1 \in \events(E_1)$ so that $\PP{\mu_1}{\scrE_1} = y(\scrE_1)$ for all $\scrE_1 \in \events(E_1)$. Let $\mu_2$ be the distribution over $\events(E_2)$ resulting from sampling $\scrE_1$ from $\mu_1$ and then sampling an event $\scrE_2$ from $E_2$ conditioned on $\scrE_1$. Then:
$$\PP{\scrE_2 \sim \mu_2}{\scrE_2} = y(\scrE_2)$$
for all $\scrE_2 \in \events(E_2)$. 
\end{lemma}
\begin{proof}
	Let $\scrE_2 \in \events(E_2)$. Then the probability we sample $\scrE_2$ from $\mu_2$ is:
	\begin{align*}
		&= \sum_{\scrE_1 \in \events(E_1)} y(\scrE_1)\P{\scrE_2 \mid \scrE_1} \\
		&= \sum_{\scrE_1 \in \events(E_1)} y(\scrE_2)\mathbb{I}[\scrE_2 \in \events(E_2)_{\mid \scrE_1}] && \text{By \cref{lem:conditioning_well_defined}} \\
		&= y(\scrE_2) \sum_{\scrE_1 \in \events(E_1)} \mathbb{I}[L(\scrE_1) = L(\scrE_2) \cap L(E_1)] \\
		&= y(\scrE_2)
	\end{align*}
\end{proof}
Where the last inequality follows because there is exactly one such event.
This allows us to define  \cref{alg:structured_rounding}, which we will run using \textsc{Structured-Rounding}$(r,\emptyset)$. It will return a WTAP solution.

\begin{algorithm}\label{alg:structured_rounding}
\caption{\textsc{Structured-Rounding}$(v,L')$}
\label{alg:structured_rounding}
\begin{algorithmic}[1]
    \State $\scrE = (e_v,L')$, $O = \emptyset$ \Comment{$O$ will be a multi-set so that links can be added twice.}
    \If{$E^*(v) \not= \{e_v\}$}
        \State Sample an event $\scrE_{cor}$ from $\events(E^*(v))$ conditioned on $\scrE$. 
        \State Set $\scrE = \scrE_{cor}$.
        \State Add $L(\scrE_{cor}) \smallsetminus L_{e_v}$ to $O$.
    \EndIf
    \For{each $e \in E(v) \smallsetminus E^*(v)$}
        \State Sample an event $\scrE_e$ from $\events(E^*(v) \cup \{e\})$ conditioned on $\scrE$. Add $L(\scrE_e) \cap L_e$ to $O$.
    \EndFor
    \Return $O \cup \bigcup_{u \in \text{children}(v)}$ \textsc{Structured-Rounding}$(u,O \cap L_{\{u,v\}})$
\end{algorithmic}
\end{algorithm}
Note conditional sampling on the event $(e_r,\emptyset)$ for the root is equivalent to simply sampling unconditionally. We also note that the fact that $O$ is a multi-set will only be used in \cref{sec:cleanup} when we add a post-processing procedure to the algorithm. Using that the input $L'$ to the algorithm is always $L(\scrE) \cap L_{v}$ for some event $\scrE$ with $y(\scrE) > 0$, the algorithm can always sample an event in lines (3) and (6) by \cref{lem:conditioning_well_defined}, so it is well-defined. 
Furthermore, since it samples an event for each edge, and each link set corresponding to an event covers its corresponding edges, it is clearly a WTAP solution. 

\begin{lemma}\label{lem:simple_structured_rounding_main_technical}
    For each $v \in V$, the following hold:
    \begin{enumerate}[label=(\roman*)]
        \item For each $L' \subseteq L_{v}$, the probability the algorithm calls \textsc{Structured-Rounding}$(v,L')$ is $y((e_v,L'))$.
        \item If $E^*(v) \not= \{e_v\}$,  $\scrE_{cor}$ in $\events(E^*(v))$ is sampled in line (3) with probability $y(\scrE_{cor})$. 
        \item For each $e \in E(v) \smallsetminus E^*(v)$, each event $\scrE_e \in \events(E^*(v) \cup \{e\})$  is sampled in line (7) with probability $y(\scrE_e)$. 
    \end{enumerate}
\end{lemma}
\begin{proof}
    Note that in every run of the algorithm, \textsc{Structured-Rounding}$(v,L')$ is called for every $v \in V$ with some set $L' \subseteq L_{v}$ since \textsc{Structured-Rounding}$(v,L')$ calls itself on each child of $v$ and we begin at the root. We will prove (i) by induction starting from the root of the tree. The base case holds because $y((e_r,\emptyset)) = 1$ and we call \textsc{Structured-Rounding}$(r,\emptyset)$ with probability 1.
    
    So, assume (i) holds for a vertex $v$, and we will prove it for all of its children. We will also show (i) implies (ii) and (iii) for $v$ itself, which is sufficient to show the whole lemma. 
    So, fix a child $c$ and some $K \subseteq L_{e_c}$. We need to show that the algorithm calls \textsc{Structured-Rounding}$(c,K)$ in step (7) with probability $y((e_c,K))$. By the inductive hypothesis, \textsc{Structured-Rounding}$(v,L')$ is called for each $L' \subseteq L_{v}$ with probability $y((e_v,L'))$. In other words, there is a distribution $\mu$ over $\events(e_v)$ where $\PP{\mu}{\scrE} = y(\scrE)$ for all $\scrE \in \events(e_v)$. 
    \begin{itemize}
    \item If $c$ is a correlated child, its call is determined by $\scrE_{cor}$. Our distribution over $\scrE_{cor}$ is the result of first sampling $\scrE$ from $\mu$ and then sampling $\scrE_{cor}$ from $\events(E^*(v))$ conditioned on $\scrE$. Since $\{e_v\} \subseteq (E^*(v))$, we can apply \cref{lem:extension_sampling}, which gives that the probability we obtain a particular event $\scrE_{cor} \in \events(E^*(v))$ is exactly $y(\scrE_{cor})$, proving (ii). Given this, we can compute the probability we call \textsc{Structured-Rounding}$(c,K)$. The events $\scrE_{cor}$ with link set $K$ on $L_{e_c}$ are exactly the events in $\events(E^*(v))_{\mid (e_c,K)}$. So, the probability of calling \textsc{Structured-Rounding}$(c,K)$ is given by:
    $$\sum_{\scrE_{cor} \in \events(E^*(v))_{\mid (e_c,K)}} y(\scrE) = y((e_c,K))$$
    as desired. 
	\item Thus, when we reach line (6), the distribution over $\scrE$ is over $\events(E^*(v))$ with $\P{\scrE} = y(\scrE)$ for all $\scrE \in \events(E^*(v))$. (If $E^*(v) = \{e_v\}$, then this is simply the inductive hypothesis.) Now we sample $\scrE_{e_c}$ for child $c$ conditioned on $\scrE$. This induces a distribution over $\events(E^*(v) \cup \{e_c\})$, and we can apply \cref{lem:extension_sampling} similar to above to conclude each event is sampled with probability $y(\scrE_{e_c})$, proving (iii). In this case, the probability of calling \textsc{Structured-Rounding}$(c,K)$ is the probability we sample an event $\scrE_{e_c}$ where $L(\scrE_{e_c}) \cap L_{e_c} = K$, which is:
	$$\sum_{\scrE_{cor} \in \events(E^*(v) \cup \{e_c\})_{\mid (e_c,K)}} y(\scrE) = y((e_c,K))$$
    \end{itemize}
    And so by induction the claim holds.
\end{proof}
We can use this to show the following:
\begin{lemma}
    Every correlated link $\ell$ and every up-link $\ell$ is selected in a solution to\\ $\textsc{Structured-Rounding}$ with probability $x(\ell)$. Every uncorrelated link $\ell$ which is not an up-link is selected with probability $2x(\ell)$.
\end{lemma}
\begin{proof}
    Let $\ell$ be a link with apex $v$. If it has one leading edge $e=\{u,v\}$ (equivalently, if it is an up-link), then by \cref{lem:simple_structured_rounding_main_technical} (i) on vertex $u$, we have the link in our solution with probability:
    $$\sum_{\scrE \in \events(e)} y(\scrE)\mathbb{I}[\ell \in L(\scrE)] = x(\ell)$$
    using the marginal preserving constraints. 
    Otherwise, $\ell$ has two leading edges $e=\{u,v\}$ and $f=\{w,v\}$ and it is an uncorrelated link, then the probability it is in the solution is:
    $$\sum_{\scrE \in \events(e)} y(\scrE)\mathbb{I}[\ell \in L(\scrE)]+\sum_{\scrE \in \events(f)} y(\scrE)\mathbb{I}[\ell \in L(\scrE)] = 2x(\ell),$$
    where we get equality because we will allow ourselves to add multiple copies of a link: one in each independent call. (This will become relevant in \cref{sec:cleanup} where we improve the algorithm.) If $\ell$ is correlated, at least one of $e$ and $f$ is in $E^*(v)$. But then, the call to \textsc{Structured-Rounding} for $u$ and the call to \textsc{Structured-Rounding} for $w$ either both use the link $\ell$, or neither of them use $\ell$. If both $u,w$ are correlated, this is clear as the link sets in their calls are determined by $\scrE_{cor}$. If exactly one is correlated, WLOG $u$, then the call to $w$ is conditioned on the event $\scrE_{cor}$. So, by the definition of conditional sampling, any event sampled for $w$ will have $\ell$ if and only if $\ell \in L(\scrE_{cor})$. Therefore, the link is used with probability $x(\ell)$.
\end{proof}
We get our desired theorem as an immediate corollary:
\roundStructured*

\subsection{The Odd Cut LP Rounding Algorithm}
\label{sec:general-Splitting-Alg}

Here we give a simple algorithm for Structured Fractional Solutions $(x,y,z)$ that gets a factor of 2 on the correlated links and a factor of 1 on the uncorrelated links. We will only use that $z$ is a feasible solution to the Odd Cut LP and that for each $v \in V$ and each edge $e$ from $v$ to an uncorrelated child of $v$, the supports of $z(L_{v})$ and $z(L_{e})$ are disjoint.
This algorithm is essentially an application of \cref{thm:FGKS} (the result of \cite{FGKS18}) which demonstrates that the Odd Cut LP has an integrality gap of 1 on instances with only cross-links and up-links. 

We will exploit the following basic property of Structured LP solutions:
\begin{fact}\label{fact:leading-uncorr-edges}
    Consider the path of edges covered by any link $\ell$ with $z_\ell > 0$ in a Structured LP solution. The only uncorrelated edges that may lie on path are leading edges of $\ell$.
\end{fact}
\begin{proof}
    Suppose there is an uncorrelated edge $e$ covered by $\ell$ which is not a leading edge. Let $f$ be the parent edge of $e$ covered by $\ell$. Then, where $v$ is the vertex shared by $e$ and $f$, $e$ goes to an uncorrelated child of $v$. So, since $z$ is a Structured LP solution, the supports of $z(L_e)$ and $z(L_f) = z(L_v)$ are disjoint, which contradicts the fact that $z_\ell > 0$. 
\end{proof}



 
\begin{figure}[htb!]
\centering
\begin{forest}
for tree={
  circle, draw, minimum size=0.5cm, inner sep=1pt, outer sep=0pt,
  l sep=30pt, s sep=30pt,
}
[{}, name=a, fill=white
  [{}, name=q, fill=white, edge label={node[midway,left]{}},
    [{}, name=u, fill=white, edge label={node[midway,left]{}}, edge={blue, thick}
      [{}, name=v9, fill=white, edge label={node[midway,left]{}},edge={blue, thick}]
      [{}, name=v, fill=white, edge label={node[midway,left]{}},edge={blue, thick}]
    ]
  ]
  [{}, name=q1, fill=white, edge label={node[midway,right]{}},
    [{}, name=u1, fill=white, edge label={node[midway,right]{}}, edge={blue, thick}
      [{}, name=v1, fill=white, edge label={node[midway,right]{}}]
      [{}, name=w1, fill=white, edge label={node[midway,right]{}}]
    ]
  ]
  [{}, name=q2, fill=white, edge label={node[midway,right]{}}
  	  [{}, name=va1, fill=white, edge label={node[midway,right]{}}]
      [{}, name=wa1, fill=white, edge label={node[midway,right]{}}]
  ]
  ]
]
\draw[dashed, thick, red, bend left=20] (v) to (q1);
\draw[dashed, thick, red, bend left=10] (q) to (q1);
\draw[dashed, thick, red, bend left=10] (u1) to (q2);
\draw[dashed, thick, red, bend left=10] (q2) to (a);
\draw[dashed, thick, red, bend left] (v) to (u1);
\draw[dashed, thick, red, bend left] (v1) to (u1);
\draw[dashed, thick, red, bend right] (v1) to (w1);
\draw[dashed, thick, red, bend left] (u1) to (w1);
\draw[dashed, thick, red, bend left] (va1) to (q2);
\draw[dashed, thick, red, bend right] (wa1) to (q2);
\draw[dashed, thick, red, bend right] (va1) to (wa1);
\draw[dashed, thick, red, bend right] (va1) to (wa1);
\draw[dashed, thick, red, bend left] (v9) to (u);
\draw[dashed, thick, red, bend left] (v9) to (q);
\node[draw=red, thick, rounded corners, inner sep=2pt, fit=(q2)(va1)(wa1)] {};
\node[draw=red, thick, rounded corners, inner sep=2pt, fit=(u1)(v1)(w1)] {};
\end{forest}

\caption{Blue edges are correlated, and the remaining edges are not. Here we show a possible split Structured LP solution. The links of $\tilde{z}$ are in red and all have value $\frac{1}{2}$. There will be three sets $V_i$ for which $E(V_i)$ is non-empty. Two of the sets are shown boxed. The third set contains all vertices not in the two displayed partitions, plus the top vertex of each of the displayed boxed sets. Observe that the sets partition the edges and links, while there can be some overlap in the vertex sets.}\label{fig:disjoint_instances}
\end{figure}
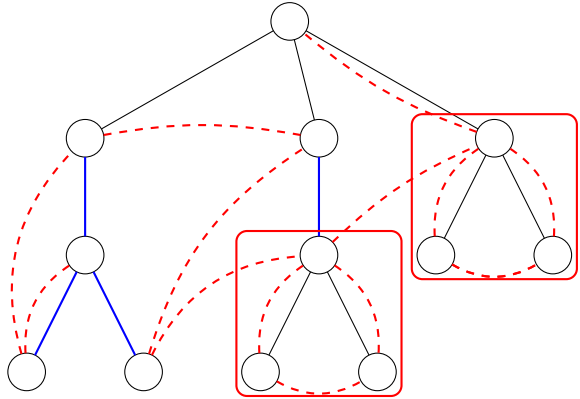

\roundOddCut*
\begin{proof}
    Split all correlated links into up-links (do nothing for correlated links that are already up-links). Let $\tilde{z}$ be the resulting LP solution. Then, $c(\tilde{z}) = 2c(z(L_C)) + c(z(L_U))$. Intuitively, one can perhaps see that now the instance can be split into a collection of edge and link-disjoint instances which consist of only cross-links and up-links as in \cref{fig:disjoint_instances}. At this point we can apply \cref{thm:FGKS} of \cite{FGKS18} to complete the proof. 

    Now we will show the statement formally. For convenience, delete all links $\ell$ with $\tilde{z}(\ell) = 0$. We will show our theorem by creating vertex sets $V_1,\dots,V_n$ of $V$ (one for each vertex) that have the following properties. Let $E(V_i)$ be the set of edges with both endpoints in $V_i$ and $L(V_i)$ the set of links with both endpoints in $V_i$.
    \begin{itemize}
        \item $E(V_1),\dots,E(V_n)$ is a partition of the edge set. 
        \item $L(V_1),\dots,L(V_n)$ is a partition of the set of links and every link in $L(V_i)$ only covers edges in $E(V_i)$.
        \item In each instance $i$ induced by $V_i$, every link in $L(V_i)$ is either a cross-link or an up-link. 
    \end{itemize}
    The theorem follows immediately from these properties and \cref{thm:FGKS} of \cite{FGKS18}, since then $\tilde{z}_{\mid L(V_i)}$ (meaning, $\tilde{z}$ restricted to the links in $L(V_i)$) can be rounded to a solution covering every edge in $E(V_i)$ of cost at most $c(\tilde{z}_{\mid L(V_i)})$. In particular, one can find a vertex solution of the LP in each subproblem to obtain an integral solution of cost at most $c(\tilde{z}_{\mid L(V_i)})$ and return the union at a cost of at most $c(\tilde{z})$. So, in the remainder we will show how to construct this partition and prove it has these three properties.

    To construct $V_1,\dots,V_n$, do the following. For each vertex $i$, add it, all its uncorrelated children $v_1,\dots,v_k$, and all vertices in all maximal paths of correlated edges hanging off of $v_1,\dots,v_k$ to $V_i$. In this way there is a set $V_i$ for each vertex $i$, some of which may only contain vertex $i$. See \cref{fig:disjoint_instances} for an example.
    
    We first show that $E(V_1),\dots,E(V_n)$ is a partition of the edge set (with some $E(V_i)$ possibly equal to the empty set). Let $e=v_iv_j \in E$, where $v_i$ is the parent of $v_j$. If $e$ is uncorrelated, then $e \in E(V_i)$, and this is clearly the unique set containing $e$. If $e$ is correlated, then look at the path from $e$ to the root. Let $e'=v_iv_j$ be the uncorrelated edge that appears first on this path, i.e., is closest to $e$. Such an edge exists because the edges incident to the root are uncorrelated. So, $e \in E(V_i)$, and $e$ cannot be in any other set by definition of $V_1,\dots,V_n$. Thus this is a partition of the edge set.

    Second, we show that $L(V_1),\dots,L(V_n)$ is a partition of the link set of $\tilde{z}$ and every link in $L(V_i)$ only covers edges in $E(V_i)$. Consider a link $\ell=v_iv_j$. If it is uncorrelated, then by \cref{fact:leading-uncorr-edges}, it covers only its leading edge(s) and the maximal paths of correlated edges hanging off of its leading edges. However these edges are all in $E(V_i)$ by definition of $V_i$, so this is the unique set containing $\ell$. Second, consider a correlated up-link $\ell$. Then, by \cref{fact:leading-uncorr-edges}, $\ell$ only covers correlated links. However any path of correlated links is not broken into different sets $V_1,\dots,V_n$ so this link only covers edges in exactly one partition.

    Finally, we notice that for every instance induced by $V_i$, every link is a cross-link or an up-link. The only in-links that could remain are uncorrelated. However, the only uncorrelated links in any partition $E(V_i)$ are those with apex $v_i$, which are thus cross-links. This completes the proof. 
\end{proof}

\section{Cleanup Phase}\label{sec:cleanup}

Here we show how to go below a 1.5 approximation given a Structured LP solution by ensuring that every uncorrelated in-link $\ell$ is used with probability slightly less than $2x(\ell)$. Intuitively, to do so, we would like to positively correlate the distributions over $L_e$ and $L_f$ for every in-link $\ell$ with leading edges $e,f$ so that if $\ell$ is taken on one side, it is more likely to be taken on the other as well (while maintaining the marginals on each side). Unfortunately, we do not know how to do this.\footnote{We are aware of a strategy that gives positive correlation on the order of 1 over the maximum size of the link sets sampled by the algorithm. Unfortunately this is not enough to counteract the error of $\epsilon$ in the cost of the Structured LP solution.} Finding such a strategy would be very interesting. 

Instead, we introduce a \textit{clean-up phase} which runs after the completion of the above algorithm. In this step, we show that each uncorrelated link is made (nearly) redundant with some probability by those sampled from other independent subtrees. 

\subsection{Preliminaries}

For the analysis, we will focus our attention on a vertex $v$ with uncorrelated children $v_1,\dots,v_k$ (for $k \ge 1$). Let the edges $\{v,v_1\},\dots,\{v,v_k\}$ be denoted $e_1,\dots,e_k$.

\paragraph{Subtree Notation.} For each uncorrelated child $v_i$, let $Z_i$ be the edges in the connected component containing $v_i$ in the forest defined by the correlated edges. Let $Z_i^+ = Z_i \cup \{e_{i}\}$. 

\paragraph{Conditioning on $\scrE$.} Given our vertex $v$, let $\scrE$ be the value of $\scrE$ at line (6) in the algorithm in the call \textsc{Structured-Rounding}$(v,L')$ (for some $L'$). 
For every link $\ell$ with apex at or lower than $v$, let $\alge(\ell)$ be the conditional probability that $\ell$ is sampled by \cref{alg:structured_rounding} given the value of $\scrE$ at line (6). 

\paragraph{Subtrees $A^\scrE_i$.}
For each $v_i$, some of the edges in $Z^+_i$ are now likely to be covered by sampling events for \textit{other} uncorrelated children $v_j$ conditioned on $\scrE$. Let $A^\scrE_i \subseteq Z_i^+$ be the collection of edges $e \in Z_i^+$ for which the probability that $e$ is covered by a link in $\bigcup_{j \not= i} L(\scrE_{e_i})$ sampled in line (7) of the algorithm conditioned on $\scrE$ is at least $\gamma > 0$, a constant we will fix later. Also add to $A_i$ all edges in $Z_i^+$ which are covered by a link in $\scrE$.

We will use the following fact, which holds because if $e \in Z_i^+$ is covered by a link $\ell \in \bigcup_{j \not= i} L(\scrE_{e_i})$, $\ell$ covers all ancestors of $e$ in $Z_i^+$ as well.
\begin{fact}\label{fact:Ai-subtrees}
    $A_i^\scrE$ is a sub-tree of $Z_i^+$.
\end{fact}



\subsection{Cleanup}

For each $e_i$, the set of links in $L_{e_i}$ which may be sampled by other uncorrelated subtrees is $L_{e_i} \smallsetminus L_{UP}$. Call this set $C_i$ as they are the ones which may \textit{cover} edges in $Z_i^+$ from other uncorrelated nodes.

The following is a key technical lemma of this section. It demonstrates that given an event $\scrE$, edges $e \in Z_i^+$ that are \textit{not} in $A_i^\scrE$ are also not likely to be covered by $L(\scrE_i) \smallsetminus L_{UP}$. An important corollary of this is that for such edges, $\alge(L_e \smallsetminus C_i) \ge 1-2\gamma$. In other words, the edges $e \in Z_i^+ \smallsetminus A^\scrE_i$ are likely to be covered by links which have all their leading edges in the tree $Z^+_i$. (Equivalently, this is all links with all their leading edges in $Z_i$ together with the up-links with leading edge $e_i$). 
\begin{lemma}\label{lem:coverage_probability}
    Fix a vertex $v$ and an event $\scrE$ for line (6). Let $e \in Z_i^+ \smallsetminus A_i^\scrE$ for some $1 \le i\le k$ (where recall $v_1,\dots,v_k$ are the uncorrelated children of $v$ with edges $e_1,\dots,e_k$ to $v$). Then, the probability $e$ is covered by a link in $C_i$ from an event sampled from $\events(E^*(v) \cup \{e_i\})_{\mid \scrE}$ is at most $2\gamma$ as long as $\gamma \le \frac{3}{4}$.
\end{lemma}
\begin{proof}
    We will suppose otherwise and get a contradiction. So, assume the probability $e$ is covered by $L(\scrE') \cap C_i = L(\scrE')\smallsetminus L_{UP}$ for $\scrE' \sim \events(E^*(v) \cup \{e_i\})_{\mid \scrE}$ is greater than $2\gamma$. We will show this would imply that the probability $e$ is covered by links sampled from other uncorrelated children $v_i$ of $v$ is at least $\gamma$, which would contradict $e \not\in A^\scrE_i$. 

    For each $j \in [k]$ with $j \not= i$, let $L_j = L_e \cap L_{e_i} \cap L_{e_j}$ be the links between $v_i$ and $v_j$ that are in $L_e$. 
    Notice that now $\{L_j\}_{j \in [k], j\not= i}$ form a partition of the links $\ell \in L_e \cap C_i$ with $x(\ell) > 0$, as every link in $L_e \cap C_i$ which does not go to an uncorrelated child has already been conditioned on and none have been selected (as $e \not\in A^\scrE_i$). Let $p_j$ be the probability that at least one link in $L_j$ is sampled in $\events(E^*(v) \cup \{e_i\})_{\mid \scrE}$. Then, since we know we cover $e$ with probability at least $2\gamma$ with a link from this set, it must be that: 
    $$\sum_{j \in [k], j \not= i}^k p_j \ge 2\gamma$$
    By \cref{lem:uncorr_consistency} (which we will prove shortly), it is also the case that the probability at least one link in $L_j$ is sampled from $\events(E^*(v) \cup \{e_j\})_{\mid \scrE}$ is $p_j$. This is a natural consequence of the consistency constraints: both distributions should have the same behavior on the link set $L_e \cap L_{e_i} \cap L_{e_j}$. 

    Therefore, since we sample from the distribution $\events(E^*(v) \cup \{e_j\})_{\mid \scrE})$ of each uncorrelated child $v_j$ \textit{independently} conditioned on $\scrE$, the probability we cover $e$ with an edge due to sampling from a correlated node $v_j$ for $i \not= j$ is:
    \[
     1-\prod_{j \in [k], j \not= i} (1-p_i) \ge 1-\exp\left(-\sum_{j \in [k], j \not= i} p_i\right) \ge \frac{1}{2}\sum_{j \in [k], j \not= i} p_i \ge \frac{1}{2} \cdot 2\gamma = \gamma 
    \]
    where we use the inequality $e^{-x} \le 1-\frac{1}{2}x$ for $x \le 1.59$ (if $x \ge 1.59$, the inequality holds for any $\gamma \le \frac{3}{4}$) which is a contradiction.
\end{proof}

We get the following useful fact as a corollary:
\begin{corollary}\label{cor:large-coverage-outside-Ai}
    For every edge $e \in Z_i^+ \smallsetminus A_i^\scrE$, we have 
    $\alge(L_e \smallsetminus C_i) \ge 1-2\gamma$.
\end{corollary}
\begin{proof}
         Since $e$ is not covered by $\scrE$ (as $e \not\in A^\scrE_i$), it must be covered in the algorithm by an event $\scrE'$ sampled from $\events(E^*(v) \cup \{e_i\})_{\mid \scrE}$ or by a link with a leading edge in $Z_i$. By \cref{lem:coverage_probability}, it is covered by an edge in $C_i \cap L(\scrE')$ with probability at most $2\gamma$. Therefore, it must be covered by a link in $L_e \cap L_{UP}$ or a link with a leading edge in $Z_i$ with probability at least $1-2\gamma$, which gives the corollary. 
\end{proof}

\begin{lemma}\label{lem:uncorr_consistency}
    Let $v_i,v_j$ be two uncorrelated children of $v$. Let $\mu_i$ be the distribution over link sets sampled from $\events(E^*(v) \cup \{e_i\})_{\mid \scrE}$ and $\mu_j$ the distribution over link sets sampled from $\events(E^*(v) \cup \{e_j\})_{\mid \scrE}$ in step (7) of \cref{alg:structured_rounding}. Then the distribution over link sets in $L_{e_i} \cap L_{e_j}$ is identical in $\mu_i$ and $\mu_j$. 
\end{lemma}
\begin{proof}
    $E^*(v) \cup \{e_i,e_j\}$ is in $\stars$ since it is $E^*(v)$ together with at most two uncorrelated edges. By \cref{lem:extension_sampling}, the distribution $\mu$ over $\events(E^*(v) \cup \{e_i,e_j\})$ can be obtained in two different ways:
    \begin{enumerate}[label=(\roman*)]
        \item First sample $\scrE \sim \events(E^*(v))$. Then, sample $\scrE' \sim \events(E^*(v) \cup \{e_i\})_{\mid \scrE}$. Finally, sample $\scrE'' \sim \events(E^*(v) \cup \{e_i,e_j\})_{\mid \scrE'}$.
        \item First sample $\scrE \sim \events(E^*(v))$. Then, sample $\scrE' \sim \events(E^*(v) \cup \{e_j\})_{\mid \scrE}$. Finally, sample $\scrE'' \sim \events(E^*(v) \cup \{e_i,e_j\})_{\mid \scrE'}$.
    \end{enumerate}
    In other words, it does not matter which order we condition on $i$ and $j$ in; both distributions lead to $\mu$. However, using the definition of conditional sampling, the links in $L_{e_i} \cap L_{e_j}$ are fixed after sampling either from $\events(E^*(v) \cup \{e_i\})_{\mid \scrE}$ in (i) or from $\events(E^*(v) \cup \{e_i\})_{\mid \scrE}$ in (ii). Therefore the distributions must be identical on this link set conditioned on $\scrE$. 
\end{proof}

This means that for every $\scrE$ chosen in line (6), if we fix a link $\ell \in L_{e_i}$, we can look at the path of edges it covers in $Z^+_{i}$, and for each one, either there is a probability of $\gamma$ that it is covered by links sampled for other uncorrelated children or it is likely to be covered by links with all their leading edges in $Z_i^+$. What we want to do now is for each $i$, \textbf{remove} some copies of links for $v_i$  which are not in $L_{UP}$ and have $P_\ell \cap A^\scrE_i$ covered by links sampled in events for other uncorrelated nodes $v_j$ for $j \not= i$. (Recall that $P_\ell$ is the set of edges link $\ell$ covers.) For each link removed in such a way, we will need to cover the edges $e \in P_\ell \smallsetminus A^\scrE_i$. We will do so with links that have all their leading edges in $Z^+_i$.

We need to be slightly careful, as we don't want to have link $\ell_1$ deleted because of the presence of $\ell_2$ and then also delete $\ell_2$ because of $\ell_1$. So, to deal with this in a simple way, we will make each $i \in [k]$ \textit{protected} with probability $\frac{1}{2}$, and we will never remove protected links. We will formalize this now. Let $V^{prot}$ be the set of protected uncorrelated vertices. 

\paragraph{Protected Links and Copies.} Since we may sample each uncorrelated link twice (once for each of its leading edges), as discussed above, we think of each uncorrelated link in $L_{e_i} \cap L_{e_j}$ as having two copies instead, one for $v_i$ and one for $v_j$. We say the copy $\ell$ for $v_i$ is protected if index $i$ is protected.  Otherwise the copy of $\ell$ for $v_i$ is unprotected. Let $L^{prot}$ be the union of all protected links sampled by the algorithm. We will sometimes just call $L^{prot}$ the set of protected links. 

\begin{figure}[htb!]
\centering
\begin{forest}
for tree={
  circle, draw, minimum size=0.5cm, inner sep=1pt, outer sep=0pt,
  l sep=30pt, s sep=30pt,
}
[{$v$}, name=r, fill=white, edge={blue, thick},
    [{$v_i$}, name=a, fill=white, edge={blue, thick},
        [{}, name=q, fill=white, edge label={node[midway,left]{$Q_1$}}
            [{}, name=u, fill=white, edge label={node[midway,left]{$Q_1$}}]
            [{}, name=u2, fill=white, edge label={node[midway,right]{$Q_1$}}]
        ]
        [{}, name=q1, fill=white, edge label={node[midway,right]{}}, edge={blue, thick}
            [{}, name=u1, fill=white, edge label={node[midway,right]{}}, edge={blue, thick},
                [{}, name=v1, fill=white, edge label={node[midway,left]{$Q_2$}}]
                [{}, name=w1, fill=white, edge label={node[midway,right]{$Q_3$}}]
            ]
        ]  
        [{}, name=q2, fill=white, edge label={node[midway,right]{}}, edge={blue, thick}
            [{}, name=va1, fill=white, edge label={node[midway,right]{$Q_4$}}]
        ]
    ]
]
\end{forest}

\caption{Shown is an example of a tree $Z_i^+$, where in bold blue is the subtree $A_i^\scrE$. $\cQ_i$ has four members, $Q_1,\dots,Q_4$, labeled here. Note that $Q_2$ and $Q_3$ are treated as different subtrees.}\label{fig:Qtrees}
\end{figure}
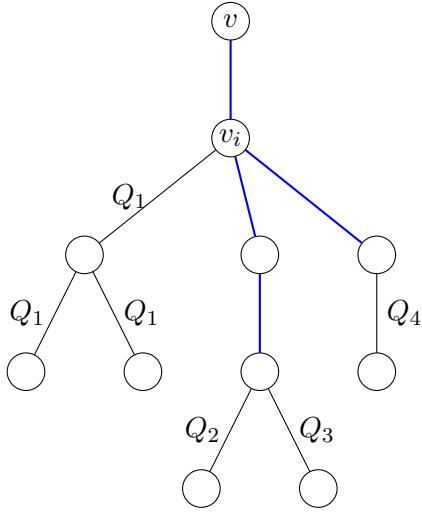

\paragraph{Subtrees $\cQ_i$, Active Subtrees, and $ADD(Q)$.} For each event $\scrE$ and each $v_i$, by \cref{fact:Ai-subtrees} $A^\scrE_i$ is a subtree of $Z_i^+$. So, $Z_i \smallsetminus A^\scrE_i$ is a collection of edge disjoint subtrees. In particular, we will start each subtree at an edge, not a vertex, so that each subtree has a unique highest edge in $Z_i \smallsetminus A_i^\scrE$ (see \cref{fig:Qtrees} for an example). Call this collection $\cQ_i$ (while it depends on $\scrE$, we will drop having this as a superscript for notational ease). For a subtree $Q \in \cQ_i$, we let $ADD(Q)$ be the cheapest up-link solution covering $Q$. Intuitively, $ADD(Q)$ is not too expensive as by \cref{cor:large-coverage-outside-Ai} its edges are all very likely to be covered by links with all their leading edges in $Z_i^+$. We say a subtree $Q \in \cQ_i$ with root $r$ is \hypertarget{active}{\textbf{active}} if $e_r$, the parent edge of $r$, is covered by links in $L^{prot}$, and index $i$ is unprotected. (We use the term active to communicate that in this event, we may need to take action to cover edges in $Q$.)

\begin{fact}\label{fact:subtrees-covered}
    For every $i$ and every $\scrE$, each subtree $Q \in \cQ_i$ is active with probability at least $\frac{1}{4}\gamma$.
\end{fact}
\begin{proof}
    Let $r$ be the root of $Q$. Then, $r \in A_i^\scrE$. So, $e_r$, the parent edge of $r$, is covered by a copy of a link $\ell$ for $v_j$ with $\ell \in \bigcup_{j \not= i} L(\scrE_{e_i})$ with probability at least $\gamma$. With probability $\frac{1}{4}$, $i$ is unprotected and $j$ is a protected, in which case $Q$ is active.
\end{proof}

\paragraph{Dominated Links.} We say the copy of an uncorrelated link $\ell$ for $v_i$ which is not an up-link is \textit{dominated} if:
\begin{enumerate}
    \item If $P_\ell \cap Z^+_i \subseteq A_i^\scrE$, we will require that $i$ is unprotected and $A^\scrE_{i} \cap P_\ell$ is covered by links in $L^{prot}$, i.e., for all $e \in A^\scrE_{i} \cap P_\ell$, there exists a link $\ell' \in L^{prot}$ so that $\ell' \in L_e$. 
    \item If $P_\ell \cap Z^+_i \not\subseteq A_i^\scrE$, $\ell$ covers at least one edge in exactly one subtree $Q \in \cQ_i$ (there is only one subtree as $P_{\ell} \cap Z_i$ is a path). For such a copy of a link $\ell$ for $v_i$, let $Q_\ell \in \cQ_i$ be the subtree it covers an edge in. For $\ell$ to be dominated, we required that $Q_\ell$ is \hypertarget{active}{active}.
\end{enumerate}
Two remarks are needed here. First notice that the copy of $\ell$ for $v_i$ may be dominated and not sampled, and in fact the event that it is sampled and the event that it is dominated are independent conditioned on $\scrE$. Second, notice that for links $\ell$ with  $P_\ell \cap (Z_i \smallsetminus A_i^\scrE) \not= \emptyset$, the stated conditions are equivalent. We use the above setup only to make arguments later in this section more natural. 


\begin{lemma}
    For any fixed $i$ and uncorrelated link $\ell \in L_{e_i}$ which is not an up-link, the copy of link $\ell$ for $v_i$ is dominated with probability at least $\frac{1}{4}\gamma$.
\end{lemma}
\begin{proof}
    Fix any event $\scrE$ in line (6) of \cref{alg:structured_rounding}. If $P_\ell \cap Z_i^+ \not\subseteq A_i^\scrE$, it is dominated exactly when $Q_\ell \in \cQ_i$ is active. By \cref{fact:subtrees-covered}, this occurs with probability at least $\frac{1}{4}\gamma$.

    Otherwise, $P_\ell \cap Z_i^+ \subseteq A_i^\scrE$. Let $e$ be the lowest edge in $P_\ell \cap A_i^\scrE$. Then, $\ell$ is dominated when $e$ is covered by a sampled, protected link and $i$ is an unprotected index. $e$ is covered by a sampled link in $L(\scrE_{e_j})$ for some $j \not= i$ with probability at least $\gamma$ by definition of $A_i^\scrE$. In this case, $i$ is unprotected and $j$ is protected with probability $\frac{1}{4}$. 

    Since this holds for any choice of $\scrE$, $\ell$ is dominated with probability at least $\frac{1}{4}\gamma$. 
\end{proof}

%

Recall our plan is to remove some sampled copies of uncorrelated links $\ell \not\in L_{UP}$ which are dominated. In particular, our goal is to remove each link with constant probability conditioned on it being sampled. The notion of protected links  resolved the issue of removing $\ell_1$ for $\ell_2$ and vice versa, however, we still cannot remove links freely, as the definition of domination does \textit{not} imply that $P_\ell$ is covered by protected links, just that $P_\ell \cap A_{e_i}^\scrE$ is covered by protected links. Consider a link $\ell$ which we would like to remove with some probability. Since it is not an up-link, it goes between two uncorrelated nodes $v_i,v_j$. By the definition of Structured Fractional Solutions, we have $P_\ell \subseteq Z_i^+ \cup Z_j^+$ by \cref{fact:leading-uncorr-edges}. To remove the copy of a link $\ell$ for $v_i$, we only need to cover $P_\ell \cap Z_i^+$ since we sampled this copy of $\ell$ for $v_i$. So, if we remove this copy of $\ell$,  we must cover the edges in $P_\ell \cap (Z_i^+ \setminus A_i^\scrE)$. This will exactly be the role of $\cQ_i$. 

We are now ready to define the clean-up procedure formally.  
Before invoking it, we will run \cref{alg:structured_rounding}. Let $O$ be the set of links produced. During its run, for each node $v$, record as $\scrE_v$ the value of $\scrE$ in line (6) when \textsc{Structured-Rounding}$(v,L')$ is called for some $L'$. For each uncorrelated node $v_i$ also record as $O_{v_i}$ the new links in $O$ sampled in line (7) of \cref{alg:structured_rounding}, i.e., in $\scrE_{e_i}$. In this way $O_{v_i}$ contains the copies of links sampled for $v_i$.  Finally, after running  \cref{alg:structured_rounding}, mark each uncorrelated node as protected or unprotected with probability $\frac{1}{2}$. Call the set of sampled, protected links $L^{prot} \subseteq O$ and the set of protected vertices $V^{prot}$.

\begin{algorithm}
\caption{\textsc{Cleanup}$(O,V^{prot},L^{prot},\{\scrE_v\}_{v \in V}, \{O_v\}_{v \in V})$}\label{alg:cleanup}
\begin{algorithmic}[1]
    \For{each uncorrelated vertex $v_i \not\in V^{prot}$}
        \State Let $\scrE = \scrE_{v_i}$.  Compute $A^{\scrE}_i$ and $\cQ_i$. 
        \State Remove all dominated links $\ell$ in $O$ for which $P_\ell \cap Z^+_i \subseteq A_i^\scrE$. 
        \State For each $Q \in \cQ_i$, if $Q$ is \hyperlink{active}{active} and if $c(ADD(Q)) < c(O_{v_i} \cap L(Q))$, delete all links in $O_{v_i} \cap L(Q)$ and add $ADD(Q)$ to the solution.
    \EndFor
\end{algorithmic}
\end{algorithm}

Note that \cref{alg:cleanup} can be made to run in polynomial time as the \hyperlink{cutlp}{Cut LP} relaxation for WTAP is integral on instances consisting of only up-links. (Recall our instance is shadow-complete so there exist up-links covering every edge.) Furthermore, we can estimate the probability an edge in $Z_i$ is covered up to negligible error, allowing us to compute $A_i^\scrE$. We will ignore this error in the analysis.  

We will also consider the following algorithm we call \textsc{Cleanup-Analysis}. There are two key differences:
\begin{enumerate}
    \item We will modify the definition of active in the following way. For each $i$ and every subtree $Q \in \cQ_i$, recall that by \cref{fact:subtrees-covered}, $Q$ is active with probability at least $\frac{1}{4}\rho$ (conditioned on $\scrE$). In this modified algorithm, subsample the active event for each subtree $Q \in \cQ_i$ and each $i$ so that it is active with probability exactly $\frac{1}{4}\rho$. If $p_Q$ is the probability $Q$ is active, this can be done by creating a Bernoulli with success probability $\frac{\gamma}{4p_{Q}}$ and calling $Q$ active only if it is active and this Bernoulli is equal to 1.    
    \item Instead of only removing link sets $O_{v_i} \cap L(Q)$ only if they are sufficiently cheap in step (4) of \cref{alg:cleanup}, we will instead remove every dominated link $\ell$ after sub-sampling the active events for subtrees $Q$ according to the above. In this way, every link with $P_\ell \cap Z_i \not\subseteq A_i^\scrE$ is dominated with probability exactly $\frac{1}{4}\gamma$. 
    To cover the edges that are now uncovered, instead of adding $ADD(Q)$ for each subtree $Q \in \cQ_i$ with a dominated link, we will do the following. By \cref{cor:large-coverage-outside-Ai}, for every $e \in Z_i^+ \smallsetminus A_i^\scrE$, we have $\alge(L_{e} \smallsetminus C_i) \ge 1-2\gamma$. 
    Consider the feasible up-link only LP solution $\tilde{x}$ to the sub-problem of covering $Z^+_i \smallsetminus A^\scrE_i$ which results from starting with $\frac{1}{1-2\gamma} \cdot \alge(\ell)$ on all links $\ell$ which have all their leading edges in $Z^+_i$, and then splitting every link which is not an up-link into two up-links (maintaining the original $\alge$ value on each). 
    Now, since the support of $\tilde{x}$ is only on up-links, the vertices of the resulting LP are integral, so we can decompose $\tilde{x}$ into a convex combination of solutions covering $Z_i^+ \smallsetminus A_i^\scrE$ so that the probability link $\ell'$ is used is $\tilde{x}(\ell')$. Sample a solution $S$ from this convex combination. Then, for each $Q \in \cQ_i$ with a dominated link, add $S \cap L(Q)$ to the solution. 
\end{enumerate}

\begin{fact}
    The solution returned by \cref{alg:cleanup} is at most the cost of the solution returned by \textsc{Cleanup-Analysis} (with probability 1). 
\end{fact}
\begin{proof}
    Each up-link solution $S \cap L(Q)$ added in \textsc{Cleanup-Analysis} is no cheaper than $ADD(Q)$, so every time a replacement is done in both algorithms, \cref{alg:cleanup} only does better.  \cref{alg:cleanup} may perform replacements that \textsc{Cleanup-Analysis} does not do because of sub-sampling, but these replacements only decrease the cost of the solution compared to \cref{alg:cleanup}. Finally, \textsc{Cleanup-Analysis} may make some replacements which \cref{alg:cleanup} does not, as they increase the cost of the solution, but of course these only make the solution from \textsc{Cleanup-Analysis} more expensive.
\end{proof}

So we can now analyze the performance of \textsc{Cleanup-Analysis}.
\begin{fact}
    The probability an uncorrelated link $\ell \not\in L_{UP}$ is included in the solution after running \textsc{Cleanup-Analysis} is at most $(2-\gamma/2)x(\ell)$.
\end{fact}
\begin{proof}
    Let $e_i,e_j$ be the two leading edges of $\ell$. Then, the copy of $\ell$ for $v_i$ is sampled by the algorithm with probability $x(\ell)$ by \cref{lem:simple_structured_rounding_main_technical}. When it is sampled, it is removed in the clean-up phase with probability $\gamma/4$ (regardless of the choice of $\scrE$). Since it is an uncorrelated link and not an up-link, it is never added in the clean-up phase. This is because we only use links that have all their leading edges in $Z_i^+$ to cover subtrees $Q \in \cQ_i$: every such link is either a correlated link (if it has a leading edge in $Z_i$) or an up-link (if its sole leading edge is $e_i$). So, it is used with probability $(1-\gamma/4)x(\ell)$. The same holds for $v_j$, and adding these together we obtain the fact. 
\end{proof}

\begin{fact}
    Let $\ell$ be a correlated link or an up-link. Then, the probability $\ell$ is added to the solution due to \textsc{Cleanup-Analysis} is at most $\gamma^2\frac{x(\ell)}{1-2\gamma}$.
\end{fact}
\begin{proof}
Split every correlated link into two up-links. For each event $\scrE$, each up-link copy of $\ell$ may appear to fix at most one subtree $Q \in \cQ_i$, where $v_i$ is the root of the connected component of correlated edges containing the leading edge of $\ell$. So, we will focus on the probability that we need to add an up-link solution covering $Q$, and in that solution we include this up-link copy of $\ell$. The final probability of using $\ell$ is then twice this probability. $Q$ is active with probability $\gamma/4$. Now the probability that the event $\scrE_{e_i}$ uses a link $\ell\in C_i \cap L(Q)$ is at most $2\gamma$ by \cref{lem:coverage_probability}. Hence the total probability that we select an edge to be added to fix $Q$ is at most $\gamma/4 \cdot 2\gamma \leq \gamma^2/2$. In that case, the uplink of $\ell$ corresponding to $Q$ is used with probability at most $\alge(\ell)/(1-2\gamma)$. Since $\ell$ is a correlated link, the probability $\ell$ is used by the algorithm is $x(\ell)$. Moreover, the probability of adding a link is at most $\gamma^2/2$ regardless of the choice of $\scrE$. So the cost of $\ell$ for the associated $Q$ is at most $(\gamma/4) \cdot 2\gamma \cdot \frac{x(\ell)}{1-2\gamma} = \frac{\gamma^2}{2(1-2\gamma)}x(\ell)$.

For correlated links $\ell$ that are up-links to begin with, we do not pay a factor of 2, but the fact still holds. For uncorrelated links $\ell$ that are up-links, the fact holds as well as they can only be used to fix one subtree $Q \in \cQ_i$; here, $v_i$ is the vertex so that $e_{v_i}$ is the leading edge of $\ell$. 
\end{proof}

We can combine the above two facts to obtain our theorem analyzing the clean-up procedure.
\begin{theorem}
    After running the clean-up procedure, the expected cost of our solution is at most 
    $$(1+\gamma^2\frac{1}{1-2\gamma})c(x(L_C \cup L_{UP})) + (2-\gamma/2)c(x(L_U \smallsetminus L_{UP}))\,.$$
    \label{thm:cleanupguarantee}
\end{theorem}

\subsection{Proof of 1.49 Approximation Algorithm}
We combine \cref{thm:cleanupguarantee} and \cref{thm:round_odd_cut} to obtain an improved approximation guarantee by carefully selecting the randomization probability and the parameter $\gamma$.

\begin{theorem}
    Let $x,y, z$ be a Structured Fractional Solution. Then there is a randomized polynomial time algorithm that produces a WTAP solution of expected cost at most $\frac{789}{530} \cdot c(z(L)) \approx 1.48868 \cdot c(z(L))$.
    \label{thm:1.49structured}
\end{theorem}

\cref{thm:1.49structured} improves upon the $1.5$ approximation shown in \cref{thm:1.5structured}. By \cref{thm:reduction_to_structure}, this leads to a randomized $(1.489+\epsilon)$ approximation algorithm for WTAP (for any $\epsilon>0$), proving the main result of the paper \cref{thm:main} (by selecting $\epsilon \leq 1/1000$).

\begin{proof}[Proof of \cref{thm:1.49structured}]
We design a randomized algorithm that combines the algorithm from \cref{thm:round_odd_cut} (Alg1) and the algorithm derived from \cref{alg:structured_rounding} followed by \cref{alg:cleanup}, analyzed in \cref{thm:cleanupguarantee} (Alg2).

We set the parameters to specific rational values close to the optimum\footnote{We remark that the optimal settings are $\gamma = \frac{4 - \sqrt{7}}{9} \approx 0.1504$ and $ p = \frac{35-2\sqrt{7}}{63} \approx 0.47156$. This gives an approximation guarantee of   $\frac{110+4\sqrt{7}}{81}$, which is almost the same so we have opted for cleaner calculations.}:
\[
\gamma = \frac{3}{20} = 0.15, \quad p = \frac{25}{53}.
\]
Noting that $\gamma \le \frac{3}{4}$ as required by \cref{lem:coverage_probability}. We run Alg1 with probability $p$ and Alg2 with probability $1-p$.

First, we calculate the term $K(\gamma) = \frac{\gamma^2}{1-2\gamma}$.
\begin{align*}
1-2\gamma & = 1 - \frac{6}{20} = \frac{14}{20} = \frac{7}{10} \\
\gamma^2 & = \frac{9}{400} \\
K(\gamma) & = \frac{9/400}{7/10} = \frac{9}{400} \cdot \frac{10}{7} = \frac{9}{280}.
\end{align*}

The expected cost $\mathbb{E}[C]$ is bounded by:
\begin{align*}
\mathbb{E}[C] &\le p \left( 2 c(z(L_C)) + c(z(L_U))\right) \\
&\quad + (1-p) \left((1+K(\gamma))c(x(L_C \cup L_{UP})) + (2-\gamma/2)c(x(L_U \smallsetminus L_{UP}))\right).
\end{align*}

We use the properties of the Structured Fractional Solution: $c(x(L)) \le c(z(L))$ and $c(x(L_U \smallsetminus L_{UP})) \le c(z(L_U))$. We aim to bound $\mathbb{E}[C]$ by finding non-negative multipliers $\alpha, \beta \ge 0$ such that:
\begin{align*}
C_1 := (1-p)(1+K(\gamma)) &\le \alpha \\
C_2 := (1-p)(2-\gamma/2) &\le \alpha + \beta
\end{align*}
If these hold, the expected cost is bounded by:
\begin{align*}
\mathbb{E}[C] &\le 2p c(z(L_C)) + p c(z(L_U)) + \alpha c(x(L)) + \beta c(x(L_U \smallsetminus L_{UP})) \\
&\le (\alpha+2p) c(z(L_C)) + (\alpha+\beta+p) c(z(L_U)).
\end{align*}
The approximation ratio is $\max(\alpha+2p, \alpha+\beta+p)$. We choose $\beta=p$ to balance the coefficients, making the ratio $\alpha+2p$. We must verify that we can choose $\alpha$ such that $C_1 \le \alpha$ and $C_2 \le \alpha+p$.

Let's calculate the coefficients using the chosen values.
$1-p = 1 - \frac{25}{53} = \frac{28}{53}$.
$1+K(\gamma) = 1 + \frac{9}{280} = \frac{289}{280}$.
$2-\gamma/2 = 2 - \frac{3}{40} = \frac{77}{40}$.

\begin{align*}
C_1 &= \frac{28}{53} \cdot \frac{289}{280} = \frac{1}{53} \cdot \frac{289}{10} = \frac{289}{530}. \\
C_2 &= \frac{28}{53} \cdot \frac{77}{40} = \frac{7}{53} \cdot \frac{77}{10} = \frac{539}{530}.
\end{align*}

Now we verify the condition on $\alpha$ by calculating $C_2-p$:
\begin{align*}
C_2-p &= \frac{539}{530} - \frac{25}{53} = \frac{539}{530} - \frac{250}{530} = \frac{289}{530}.
\end{align*}
Since $C_1 = C_2-p$, we set $\alpha = C_1$. This confirms that the choice $\beta=p$ is valid.

Finally, we calculate the resulting approximation ratio $\alpha+2p$:
\begin{align*}
\alpha+2p &= \frac{289}{530} + 2 \left(\frac{25}{53}\right) = \frac{289}{530} + \frac{500}{530} = \frac{789}{530}. \qedhere
\end{align*}
\end{proof}

\section{Strong Linear Programming Relaxation}
\label{sec:strongLP}
In this section, we introduce our strong linear programming relaxation for WTAP, which we refer to as the Strong LP.
This relaxation is then used in \cref{sec:red_always_small} to reduce the problem to that of devising an approximation algorithm with respect to a structural fractional solution, i.e., to prove \cref{thm:reduction_to_structured}.
In the previous sections, we have seen that such structural fractional solutions can be rounded and thus combining our reduction together with that rounding implies our improved approximation guarantee for WTAP.

One of the main reasons that the constraints of the structured LP do not form a valid relaxation is that they assume that every tree edge is covered by only a small number of links.
To handle edges that are covered by a larger number of links, we generalize the notion of stars to more flexible subtrees. In the next subsection, we first set the parameters that will be used to define these subtrees and the corresponding events.
Then, we present the variables of the linear programming relaxation followed by its constraints.

\subsection{Parameters, Subtrees, and Events}
The Strong LP we present depends on two parameters:
\begin{itemize}
    \item $\beta$, which bounds the maximum number of leaves in the considered subtrees.
    \item $\rho$, which defines the notion of a \emph{small} link set: a set $S \subseteq L$ is small if $1 \leq |S| \leq \rho$.
\end{itemize}

\paragraph{Subtrees.}

A \emph{subtree} of the input tree $T$ is a tree obtained by deleting edges and vertices from $T$. We are particularly interested in those with at most $\beta+3$ leaves.
Let $\subtrees$ denote the family of such subtrees. Since a subtree is uniquely determined by its leaf set, we have $|\subtrees| = O(n^{\beta+3})$.
Given a subtree $R \in \subtrees$, we use $R$ to denote its edge set and $V(R)$ for its vertex set.


\paragraph{Events.}
Fix $R \in \subtrees$. Let $\LeafEdges(R)$ denote the edges incident to leaves of $R$. For a subset $R_\sm \subseteq \LeafEdges(R)$, let $L(R_\sm) = \bigcup_{e \in R_\sm} L_e$ be the set of links that cover at least one edge in $R_\sm$.  Further let $L_\sm \subseteq L(R_\sm)$ be a subset of these links.  An \emph{event} $\scrF = (R, R_\sm, L_\sm)$ encodes a local configuration with the following semantics (see Figure~\ref{fig:event_example} for an example):
\begin{quote}
    Each edge $e \in R_\sm$ is covered by a small number (at most $\rho$) of links in $L_\sm \cap L_e$, and each $e \in R \setminus R_\sm$ is covered by more than $\rho$ links.
\end{quote}
For such an event $\scrF$, we let $R(\scrF)$, $R_\sm(\scrF)$ and $L_\sm(\scrF)$ denote $R, R_\sm$ and $L_\sm$, respectively.

We define $\events$ as the collection of all such events. Since there are $O(n^{\beta+3})$ subtrees $R$, at most $2^{\beta+3}$ subsets $R_\sm \subseteq \LeafEdges(R)$, and at most $|L|^\rho$ link combinations for each edge in $R_\sm$, we get:
\begin{align}
    |\events| = O\left(n^{\beta+3} \cdot 2^{\beta+3} \cdot |L|^{\rho \cdot (\beta+3)}\right),
    \label{eq:num_events}
\end{align}
which is polynomial for fixed constants $\beta$ and $\rho$.
To avoid confusion, we mostly  use  $\scrF$ to denote an event in $\events$ and $\scrE$ to denote the events/stars in the structural fractional solution. (Exceptions are made in Section~\ref{sec:red_always_small} where we produce several intermediate solutions, and so it is convenient to use both $\scrE$ and $\scrF$ for events in those intermediate solutions.)

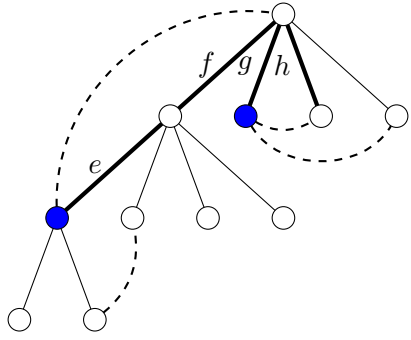
\begin{figure}[t]
    \centering
    \begin{forest}
        for tree={
        circle, draw, minimum size=0.3cm, inner sep=0pt, outer sep=0pt,
        l sep=30pt, s sep=20pt,
        }
        [{}, name=r, fill=white,  
        [{}, fill = white, name=a,  edge label={node[midway,left]{$f$}}, edge=ultra thick
                    [{}, fill=blue, name=f, edge label={node[midway,left]{$e$}}, edge=ultra thick
                            [{}, fill=white, name=v]
                            [{}, fill=white, name=w]
                    ]
                    [{}, fill=white, name=g,]
                    [{}, fill=white, name=h]
                    [{}, fill=white, name=i]
            ]
            [{}, fill=blue, name=k, edge label={node[midway,left]{$g$}}, edge=ultra thick]
            [{}, fill=white , name=l, edge label={node[midway,left]{$h$}}, edge=ultra thick]
            [{}, fill=white, name=m]
        ]
        \draw[dashed, thick, bend right] (k) to (l);
        \draw[dashed, thick, bend right=60] (k) to (m);
        \draw[dashed, thick, bend right] (w) to (g);
        \draw[dashed, thick, bend left=50] (f) to (r);
    \end{forest}

    \caption{Example of an event $\scrF = (R, R_\sm, L_\sm)$ for a subtree $R$ defined by edges $\{e,f,g,h\}$ (depicted by thick lines), with its leaves in $V(R_\sm)$ in blue.
        The set
        $L_\sm$ is depicted by dashed lines and note that the edges in
        $R_\sm = \{e,f\}$ are each covered by at most $2 \leq \rho$ links in $L_\sm$. We remark that the tree edges in $R \setminus R_\sm = \{f, h\}$ are covered by more than $\rho$ links, which are not depicted in the figure.}
    \label{fig:event_example}
\end{figure}

\subsection{Variables of Strong LP}
The Strong LP uses the following variables:
\begin{itemize}
    \item For each link $\ell \in L$, we have a variable $x(\ell)$ that is meant to indicate whether $\ell$ is in the solution.
    \item For each event $\scrF \in \events$ we define a variable $y(\scrF)$ with the intuition that $y(\scrF) = 1$ when the event holds for the optimal solution and $y(\scrF) = 0$ otherwise.
    \item A link $\ell$ is relevant to an event $\scrF = (R, R_\sm, L_\sm)$ if it covers an edge of $R$. In other words, the relevant links are $L(R) = \cup_{e\in R} L_e$. For each  link $\ell$ relevant to $\scrF$ we also define a variable $y_\ell(\scrF)$.
\end{itemize}
We remark that by~\eqref{eq:num_events}, the total number of variables is polynomial for fixed constants $\beta$ and $\rho$.

\subsection{Constraints of Strong LP}\label{subsec:strongLP_constr}
The Strong LP is defined by minimizing $\sum_{\ell \in L} w(\ell) \cdot x(\ell)$, subject to the constraints~\eqref{constr:odd_cut_consraint}--\eqref{constr:extension_consistency_with_ell}
introduced below.

For intuition, it is helpful to consider an optimal solution $L^*$ and the corresponding intended linear programming solution, which sets $x(\ell) = 1$ for $\ell \in L^*$ and $x(\ell) = 0$ for $\ell \in L \setminus L^*$. We also set $y(\scrF) = 1$ for every event $\scrF$ consistent with $L^*$, and $y(\scrF) = 0$ otherwise. Similarly, we set $y_\ell(\scrF) = y(\scrF)$ for every $\ell \in L^*$, and $y_\ell(\scrF) = 0$ for every $\ell \in L \setminus L^*$.

We now introduce the constraints of the linear program along with an explanation of why they are valid, namely why they are satisfied by the reference solution described above.

\subsubsection{Odd Cut LP Constraints} We have
\begin{equation}
    x(\delta_L(S)) + \sum_{e \in \delta_E(S)} x(L_e) \ge |\delta_E(S)| + 1
    \label{constr:odd_cut_consraint}
\end{equation}
for every  $S \subseteq V$ with $|\delta_E(S)|$ odd, and $x(\ell) \geq 0$  $\forall \ell \in L$.
These are the Odd Cut LP constraints that we discussed in Section~\ref{sec:prelimLPs}.
They were introduced in~\cite{FGKS18} that are known to be valid, i.e., they are satisfied by any $x$ that indicates a feasible solution.

\subsubsection{Coverage Constraints} For each edge $e$ there should be exactly one event $(R, R_\sm, L_\sm)$ with $R = \{e\}$ that is satisfied by the optimal solution.  Specifically, For each tree-edge $e\in T$, we enforce
\begin{equation}\label{constr:cov_consistency}
    \sum_{\scrF \in \events(e)}y(\scrF) = 1\,,
\end{equation}
where $\events(e)$ denotes the subset of events $(R, R_\sm, L_\sm)$ where $R = \{e\}$ is the tree consisting of the single edge $e$. So in words, the above constraint says that we either cover a tree-edge $e$ with at most $\rho$ links $L_\sm$ in which case the event where $(R, R_\sm, L_\sm)$ with $R_\sm = \{e\}$  is true (i.e., the corresponding variable is set to $1$), or we cover $e$ by more than $\rho$ links in which case the event $(R, \emptyset, \emptyset)$ is true.

\subsubsection{Marginal Preserving Constraints} For each tree-edge $e \in E$ and link $\ell \in L_e$, we enforce
\begin{equation}\label{constr:marginal_pres_const_child}
    x(\ell) = \sum_{\scrF \in \events(e)}y_\ell(\scrF)
\end{equation}

If $\ell \in L^*$, then exactly one event $\scrF$ is true and for that event we have $y_\ell(\scrF) = y(\scrF) = 1$ and so the constraint is satisfied since $x(\ell) = 1$. Otherwise, if $\ell \not \in L^*$ then $x(\ell) = 0$ and $y_\ell(\scrF) = 0$ for every event and so the constraint is again valid.
We refer to this as ``marginal preserving constraints" as these constraints enforce consistency between marginal usage of $\ell$ and its contribution to events involving $e$.

\subsubsection{Non-Negativity, Link Consistency}
For each $$\scrF  = (R, R_\sm, L_\sm) \in \events$$ we have
\begin{equation}\label{constr:non_neg}
    y(\scrF), y_\ell(\scrF) \geq 0
\end{equation}
and
\begin{equation}\label{constr:inclision_on_small}
    y_\ell(\scrF)=y(\scrF)
\end{equation}
for each $\ell \in L_\sm$.
In addition, for each $e \in R_\sm$ and $\ell \in L_e \setminus L_\sm$ we have
\begin{equation}\label{eq:non_inclusion_links_on_huge_edges}
    y_\ell(\scrF)=0,
\end{equation}
and for each $e \in R \setminus R_\sm$ we have
\begin{equation}\label{constr:huge_edges_const_child}
    \sum_{\ell \in L_e}y_\ell(\scrF) \geq \rho \cdot y(\scrF)
\end{equation}

These ensure that $y_\ell(\scrF)$ aligns with how $\ell$ participates in the event $\scrF$, and are valid constraints for the intended solution.

\subsubsection{Consistency Constraints Among Events}
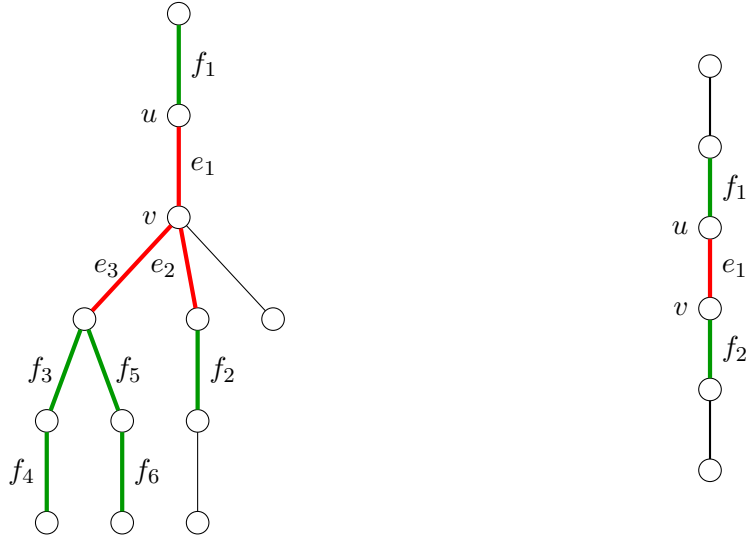
\begin{figure}[t]
    \centering
    \begin{minipage}{0.45\textwidth}
        \centering
        \begin{forest}
            for tree={
            circle, draw, fill=white, minimum size=0.3cm, inner sep=0pt, outer sep=0pt,
            l sep=30pt, 
            s sep=20pt, 
            grow=south, 
            }
            [{}, name=root 
            [{}, name=u, label={left:$u$},
            edge label={node[midway,right]{$f_1$}}, edge={ultra thick, draw=green!60!black}
                [{}, name=v, label={left:$v$},
                    edge label={node[midway,right]{$e_1$}}, edge={ultra thick, draw=red}
                        [{}, name=n1,
                            edge label={node[midway,left]{$e_3$}}, edge={ultra thick, draw=red}
                                [{}, name=n1a,
                                    edge label={node[midway,left]{$f_3$}}, edge={ultra thick, draw=green!60!black}
                                        [{}, name=n1a1,
                                            edge label={node[midway,left]{$f_4$}}, edge={ultra thick, draw=green!60!black}
                                        ]
                                ]
                                [{}, name=n1b, edge label={node[midway,right]{$f_5$}}, edge={ultra thick, draw=green!60!black}
                                        [{}, name=n1b1, edge label={node[midway,right]{$f_6$}}, edge={ultra thick, draw=green!60!black}
                                        ]
                                ]
                        ]
                        %
                        [{}, name=n2,
                            edge label={node[midway,left]{$e_2$}}, edge={ultra thick, draw=red}
                                [{}, name=n2a,
                                    edge label={node[midway,right]{$f_2$}}, edge={ultra thick, draw=green!60!black}
                                        [{}, name=n2b]
                                ]
                        ]
                        %
                        [{}, name=n3 
                        ]
                ]
            ]
            ]
        \end{forest}
    \end{minipage}
    \begin{minipage}{0.45\textwidth}
        \centering
        \begin{forest}
            for tree={
            circle, draw, fill=white, minimum size=0.3cm, inner sep=0pt, outer sep=0pt,
            l sep=20pt, 
            s sep=20pt,
            grow=south, 
            }
            [{},
            [{},  edge=thick 
            [{}, label=left:$u$, 
            edge label={node[midway,right]{$f_1$}}, edge={ultra thick, draw=green!60!black}
                [{}, label=left:$v$, 
                    edge label={node[midway,right]{$e_1$}}, edge={ultra thick, draw=red}
                        [{}, 
                            edge label={node[midway,right]{$f_2$}}, edge={ultra thick, draw=green!60!black}
                                [{},  edge=thick] 
                        ]
                ]
            ]
            ]
            ]
        \end{forest}
    \end{minipage}

    \caption{Examples of extensions. Left: An event $\scrF$ defined on subtree $R$ (edges $e_1, e_2, e_3$ in red) is extended to a larger subtree $Q$ which also includes edges in green ($f_1,f_2, \ldots, f_6$). Here, the unique extension of $\scrF$ that is consistent is found by growing outwards from the center $c=v$. Right: An example where the event $\scrF$ is defined on a subtree $R$ with a single edge $e_1$ and $Q$ consists additionally of $f_1, f_2$. Here we have to be careful to only extend $R$ in one direction to make sure that there is a unique extension consistent with $\scrF$ and the intended solution. As we prioritize the root of $R$ if it is an internal vertex of $Q$, we get that $c=u$ in this case, and we extend $R$ upwards.}
    \label{fig:extension_example}
\end{figure}

Our final constraints are designed to ensure consistency among events defined on different subtrees. These constraints enforce that the variables behave like a hierarchical system of marginal probabilities: the weight $y(\scrF)$ of an event on a subtree $R$ must be exactly distributed among compatible events on any larger subtree $Q \supseteq R$.

To formalize this compatibility, we define the notion of an extension. Intuitively, an extension $\scrF'$ describes how the configuration $\scrF$ propagates outwards within the larger context $Q$. Figure~\ref{fig:extension_example} illustrates this concept with two examples.

A crucial challenge here is ensuring this propagation is uniquely defined and directed. Specifically, to ensure that constraints~\eqref{constr:extension_consistency},\eqref{constr:extension_consistency_with_ell} are valid, we wish to have that there is a unique extension that is consistent with $\scrF$ and the intended solution. Consider the case where $\scrF$ is the event consistent with the reference solution $L^*$ for the single edge $R=\{e=(u,v)\}$, and the larger subtree $Q$ extends both upwards from $u$ and downwards from $v$ (see e.g.\, the right example of Figure~\ref{fig:extension_example}). If the set of extensions can grow in both directions, one can see that there would be several consistent extensions, which would violate the uniqueness we require. Therefore, we cannot allow the extension to grow simultaneously in both directions.  We must therefore enforce a canonical direction of growth whenever $R$ contains a single edge. When $R$ contains more than a single edge, we could grow the extension from any internal (non-leaf) vertex but to avoid ambiguity we define the extension center uniquely also in this case below.

\paragraph{The extension center.} To enforce this unique direction of growth, we define an \emph{extension center} $c \in V(R)$. The extension process will always be defined as originating from this center $c$. The choice of $c$ depends on the structure of $R$ and how $R$ is situated within $Q$. Recall that we assume the input tree $T$ is rooted, and for any subtree $R$, we denote its root (the vertex closest to the root of $T$) by $u$.

We define $c$ by analyzing the structure of $R$ and prioritizing $u$ as the center, unless $u$ is unsuitable because it lies on the boundary of the structure.

\begin{enumerate}
    \item If $u$ is an internal vertex of $R$ (i.e., not a leaf of $R$), we set $c=u$.

    \item Otherwise, if $u$ is a leaf of $R$, the selection of $c$ depends further on the structure of $R$ and $Q$. Let $v$ be the unique vertex adjacent to $u$ in $R$ and further distinguish:
          \begin{itemize}
              \item {Case  $|R| \geq 2$:}  We set $c= v$ which is an internal node in $R$ by the assumption $|R| \geq 2$.
              \item {Case $|R| = 1$:} In this case $R=\{(u,v)\}$ and the choice of $c$ is critical in determining the direction of the extension. We prioritize the root $u$, provided it allows for expansion within the context of $Q$. Specifically, if $u$ is not a leaf in $Q$, we set $c=u$; otherwise, we set $c=v$.
          \end{itemize}
\end{enumerate}
This careful case analysis ensures that $c$ is uniquely and canonically determined for any pair $(\scrF, Q)$.

\paragraph{Definition of extensions.} Once the center $c$ is fixed, the extension process is defined by growing outwards from $c$ along paths in $Q$. The growth continues as long as the edges encountered are "huge" ($>\rho$ links) and stops precisely when the first "small" edge ($\leq \rho$ links) is encountered, or when we reach the boundary of $Q$.

Formally, for an event $\scrF = (R, R_\sm, L_\sm) \in \events$ and a subtree $Q \in \subtrees$ with $R \subseteq Q$, the \emph{set of extensions} $\Ext(\scrF, Q) \subseteq \events$ contains all events $\scrF' = (R', R'_\sm, L'_\sm)$ such that $R' \subseteq Q$ and the following conditions hold (where $c$ is the extension center defined above):

\begin{itemize}
    \item {Consistency with $\scrF$:} The extension must agree with the base event on $R$.
          \begin{itemize}
              \item $R \cap R'_\sm = R_\sm$.
              \item For every $e \in R_\sm$, we have $L_\sm \cap L_e = L'_\sm \cap L_e$.
          \end{itemize}
    \item {Extension Structure (relative to $c$):} For every path $P$ in $Q$ starting at $c$ and ending at a leaf of $Q$, one of the following holds:
          \begin{itemize}
              \item (Huge path) The entire path is contained in $R'$, and no edge of $P$ lies in $R'_\sm$; that is, $P \subseteq R'$ and $P \cap R_\sm'=\emptyset$.
              \item (Path ending in small edge) The subpath $P' = P \cap R'$ begins at $c$ and contains exactly one edge from $R'_\sm$, which is the final edge of $P'$.
          \end{itemize}
\end{itemize}
We say that such an event $\scrF' \in \Ext(\scrF, Q)$ is an \emph{extension} of $\scrF$ on $Q$.

\paragraph{The constraints.} The consistency constraints require that the weight of $\scrF$ is exactly distributed among its extensions in $Q$. For each event $\scrF = (R, R_\sm, L_\sm) \in \events$ and $Q \in \subtrees$ satisfying $R \subseteq Q$ we require
\begin{equation}\label{constr:extension_consistency}
    y(\scrF) = \sum_{\scrF' \in \Ext(\scrF, Q)} y(\scrF'),
\end{equation}
and for every $\ell \in L(R)$
\begin{equation}\label{constr:extension_consistency_with_ell}
    y_\ell(\scrF) = \sum_{\scrF' \in \Ext(\scrF, Q)} y_\ell(\scrF')\,.
\end{equation}

\paragraph{Validity of the constraints.} We now explain why the extension consistency constraints are valid for the reference solution derived from an optimal solution $L^*$.

Consider first an event $\scrF$ that is not satisfied by $L^*$ ($y(\scrF) = 0$). Due to the "Consistency with $\scrF$" condition, if $\scrF$ conflicts with $L^*$, any extension $\scrF'$ (which must agree with $\scrF$ on $R$) will also conflict with $L^*$. Thus, the constraints hold.

Now suppose instead that $\scrF$ {is} satisfied by $L^*$, i.e., $y(\scrF) = 1$. Then there exists a \emph{unique} extension $\scrF'$ in $\Ext(\scrF, Q)$ consistent with $L^*$. This uniqueness relies critically on the definition of the extension center $c$ and the associated directionality it enforces.

We construct this unique extension by considering every path $P$ in $Q$ (from $c$ to a leaf in $Q$). Starting from $c$, we greedily traverse $P$, adding the edges to the extended subtree $R'$, until one of the following occurs:

\begin{enumerate}
    \item We encounter a tree edge $e$ that is covered by at most $\rho$ links in $L^*$ (a small edge). In this case, we stop the traversal along $P$, add $e$ to $R'_\sm$, and include the corresponding links $L_e \cap L^*$ in $L'_\sm$.
    \item We reach the end of the path $P$ in $Q$ without encountering such an edge (all edges along the traversal are huge in $L^*$).
\end{enumerate}

We remark that this extension process might stop when we hit a leaf-edge of the original subtree $R$. Indeed, if the path $P$, starting from $c$, encounters an edge $e \in R_\sm$, then by the consistency of $\scrF$ with $L^*$, $e$ must be small in $L^*$. Thus, the first condition is met, and the extension along $P$ terminates immediately at $e$. This ensures that the extension respects the existing boundaries defined by $\scrF$.

Because the growth process is uniquely defined starting from the canonical center $c$ and proceeding along the paths according to the structure of $L^*$, the resulting event $\scrF'$ is the unique consistent extension of $\scrF$. This confirms the validity of the constraints. 

\section{Construction of Low Cost Structured Fractional Solution}\label{sec:red_always_small}

This section is dedicated to proving the following theorem, which reduces general instances of WTAP to the problem of rounding a Structured Fractional Solution:
\reductiontostructured*

Before providing the formal reduction, we first fix the parameters of the strong linear programming relaxation and the Structured Fractional Solution.

\paragraph{Parameters.} Fix an arbitrarily small parameter $\epsilon > 0$. Based on $\epsilon$, we define the following parameters:
\[
    \rho = \lceil 2/\epsilon^2 \rceil, \qquad \delta = \epsilon, \qquad \beta = \rho, \qquad \rho' = \rho \cdot (\beta+1).
\]
Here, $\rho$ bounds the size of \emph{small} subsets of links $S \subseteq L$, meaning $1 \leq |S| \leq \rho$ for the Strong LP, and similarly $\rho'$ bounds the subsets of small subsets of links for the Structured LP; and $\beta$ bounds the number of leaves in the subtrees $R \in \subtrees$ considered in the Strong LP. While the LPs sizes grow with both $\rho$ and $\beta$, they remain polynomial in the input size for any fixed $\epsilon$. The parameter $\delta$ is used to define the notion of correlated edges, as introduced in \cref{sec:overview}.

A crucial step in our reduction is to show that we can effectively ignore edges $e \in E$ for which the Strong LP assigns a non-negligible probability to the event that $e$ is covered by many links, by computing a solution of negligible cost that covers all such edges.

\paragraph{$\epsilon$-almost-always-small Edges}
Consider a WTAP instance together with a fractional solution 
$(x, y, \{y_\ell\}_{\ell \in L})$ to the Strong LP 
(from \Cref{sec:strongLP}).  
Suppose that for an edge $e \in E$, the event 
$\scrF = (\{e\}, \emptyset, \emptyset)$ satisfies 
$y(\scrF) > \epsilon$.  
Intuitively, this means that it is relatively inexpensive 
(with respect to~$x$) to cover~$e$.  
Using standard techniques, we will see that we can essentially ignore such edges 
and restrict our attention to those edges $e$ 
for which $y(\scrF) \le \epsilon$.

For technical convenience, we will impose a slightly stronger condition 
on the edges that we retain:
\begin{equation}\label{eq:esp_always_small}
    x(L_e) \le \epsilon \cdot \rho.
\end{equation}
Any edge satisfying \Cref{eq:esp_always_small} is said to be 
\emph{$\epsilon$-almost-always-small}.  
An instance (together with its fractional solution) is called 
\emph{$\epsilon$-almost-always-small} if every edge satisfies this condition.
Finally we remark, by \Cref{constr:huge_edges_const_child}, 
if \Cref{eq:esp_always_small} holds for an edge~$e$, 
then indeed $y(\scrF) \le \epsilon$.

\paragraph{Notation $e_r$ and Pairwise Correlated Edges}
For ease of exposition--- and to avoid treating the root as a special case --- we introduce a dummy edge $e_r$, referred to as the \emph{parent edge of the root}, which is not covered by any links. We assume $y(e_r, \emptyset) = 1$.

Recall from \Cref{subsec:structured_LP_struc_frac} that we defined correlated 
nodes and edges, which are used in constructing the Structured LP and 
Structured Fractional Solutions, and that for each $v \in V$ we use $E(v)$ to denote the star containing all edges incident to $v$. To formally define the set~$V_{\mathrm{cor}}$ that our reduction uses, 
we make use of the notion of \emph{pairwise correlation}. 
Given a fractional link vector $x : L \to \mathbb{R}_{\ge 0}$, 
we say that two edges $e, e' \in E$ are \emph{pairwise correlated} if
\[
    x(L_e \cap L_{e'}) \ge \delta.
\]

\paragraph{Reduction.}
Given an instance $(T, L, w)$ of WTAP, our reduction proceeds as follows:
\begin{enumerate}[label=(\arabic*)]
    \item Solve the Strong LP from ~\cref{sec:strongLP} to obtain $(x, y, \{y_\ell\}_{\ell \in L})$, with parameters $(\rho, \beta)$.
    \item Contract all edges that are not $\epsilon$-almost-always-small with respect to $(x, y, \{y_\ell\}_{\ell \in L})$ to obtain tree $T'=(V', E')$ and find a negligible cost solution $L_\lrg$ covering all such edges (details in \cref{subsec:red_almost_always_small}).
    \item On the residual instance defined by \(T'\), solve the Strong LP with the 
    additional constraints given by \Cref{eq:esp_always_small} for every 
    \(e \in E'\), obtaining a new fractional solution 
    \((x', y', \{y'_\ell\}_{\ell \in L})\). 
    This enforces that each edge remains 
    \(\epsilon\)-almost-always-small.
    \label{step:solve_eps_always_sm}
    \item 
    We directly define \(E^*(v)\) for each \(v \in V'\), as the subset of edges in \(E(v)\) consisting of \(e_v\) 
    together with all edges that are pairwise correlated with \(e_v\) (with respect to $x'$).
    This also immediately defines $V_{cor}$ as the set of vertices $v$ such that $e_v$ is contained in $E^*(u)$ where $u$ is the parent of $v$.
    \item Next, remove from \(L\) all links that simultaneously cover both 
    \(e_v\) and an edge in \(E(v) \setminus E^*(v)\). 
    Then, solve for an optimal Structured Fractional Solution with parameter 
    \(\rho'\) bounding the size of small link set, yielding \((x'', y'', z'')\). 
    Recall that a Structured Fractional Solution satisfies    \Cref{eq:constr:small_prob_1:marginals,eq:uncorr_cost}.
    \label{step:obtain_structured_frac_solution}
    \item Apply $\calA$ to \((x'', y'', z'')\), and output the 
    resulting solution, augmented with the links in \(L_\lrg\).
\end{enumerate}

We remark here that we will abuse notation by often using $E^*(e_v)$ and $E(e_v)$ to denote the sets $E^*(v)$ and $E(v)$.

\paragraph{Main Lemmas.}
The remainder of this section is dedicated to proving the following two lemmas.

\begin{restatable}[
    label={lem:reduction_to_eps_small}
]{lemma}{reductiontoepssmall}
The solution $(x', y', \{y_\ell'\}_{\ell \in L})$ obtained in Step \ref{step:solve_eps_always_sm} of our reduction has cost at most that of the input solution $(x, y, \{y_\ell\}_{\ell \in L})$ to the Strong LP.
Furthermore the cost of $L_\lrg$ is at most $\epsilon \cdot c(x(L)).$
\end{restatable}

\begin{restatable}[
    label={lem:reduction_esp_small_always_small}
]{lemma}{reductionepstoalwayssmall}
The Structured Fractional Solution $(x'', y'', z'')$ obtained in Step \ref{step:obtain_structured_frac_solution} satisfies $$c(z''(L)) \leq (1+o_\epsilon(1))c(x'(L)),$$ where $(x', y', \{y'_\ell\}_{\ell \in L})$ is the $\epsilon$-almost-always-small solution  obtained in Step \ref{step:solve_eps_always_sm}. 
\end{restatable}

The proof of \Cref{thm:reduction_to_structured}  then directly follows from \Cref{lem:reduction_to_eps_small,lem:reduction_esp_small_always_small}. To see this we remark that one can simply solve the Strong LP, determine $T'$ and $L_\lrg$ and then solve for an optimal Structured Fractional Solution directly with respect to $T'$. Completing the argument we run $\calA$ on the obtained solution.


\paragraph{Notation $y(e, S)$, $\scov^+(e)$ and $\events(F)$.}
Events of the form $\scrE = (R, R_\sm, L_\sm)$ where $R = \{e\}$ consists of a single edge will play a special role in our analysis. In this case, the subset $R_\sm$ is redundant: $R_\sm = \emptyset$ if and only if $L_\sm = \emptyset$. Thus, we simplify the notation by writing such an event as $(e, L_\sm)$, and often rename $L_\sm$ to $S$. Accordingly, we denote the associated LP variables by $y(e, \bot)$ for the case where $L_\sm = \emptyset$, and by $y(e, S)$ for each $S \in \scov(e)$ when $L_\sm = S$.

Recall from \cref{sec:prelims} that for a set of edges $F \subseteq E$, we let $\scov(F) \subseteq 2^{L(F)}$ denote the collection of link sets that are \emph{small} on $F$ and jointly cover $F$ --- that is, sets $S \subseteq L(F)$ such that $1 \le |S \cap L_e| \le \rho$ for every edge $e \in F$. For a single edge $e \in T$, we also define the extended set $\scov^+(e) := \scov(e) \cup \{\bot\}$, where $\bot$ represents the case in which $e$ is covered by more than $\rho$ links.

Recall from \Cref{subsec:strongLP_constr} that for the Strong LP and each edge $e \in T$ we use $\events(e)$ to denote the subset of events $\scrF$ such that $R(\scrF)=\{e\}$. We extend this in the natural way so that for each $F \subseteq E$ $\events(F)$ is the subset of events $\scrF$ such that $R(\scrF) = F$.

\subsection{Reduction to Almost Always Small}\label{subsec:red_almost_always_small}
In this section, we prove \Cref{lem:reduction_to_eps_small}.  
The key idea is simple: if an edge $e$ has a relatively large event weight
$y(\scrF) \ge \epsilon$ for $\scrF = (\{e\}, \emptyset, \emptyset)$, then $e$
is already heavily covered in the fractional solution.  
Indeed, by our choice of $\rho$, such an edge must be fractionally covered by
at least $\rho \cdot y(\scrF) \ge 2 / \epsilon$ links.  
Hence, even after scaling the fractional solution $x$ by a factor of
$\epsilon / 2$, these edges remain fractionally covered.

Since the integrality gap of the \hyperlink{cutlp}{Cut LP} is at most~2,
we can round this scaled-down solution to obtain a set of large links,
denoted $L_\lrg$, at an additional cost of at most an
$\epsilon$-fraction of the original LP value. 

Next, let $T' = (V', E')$ denote the tree obtained by contracting the edges that are not
$\epsilon$-almost-always-small, and let $L'$ be the set of links
obtained from $L$ by contracting the same set of edges.  
We define the new weight function $w' : L' \to \R_{\ge 0}$ by setting
$w'(\ell') = w(\ell)$, where $\ell \in L$ is the preimage of $\ell'$ under
contraction.  For simplicity, we continue to refer to this contracted link
set as $L$, with weights $w$.

The remainder of the proof shows that for every $e \in E'$,
$x(L_e) \le \epsilon \cdot \rho$ can
be added to the Strong LP without increasing the cost of the Strong LP for the contracted
WTAP instance $(T', L, w)$. In words, this constraint means that edges which were
already nearly ``small'' in the original LP remain so in the reduced
instance.

To complete the proof of the lemma, it then suffices to solve the Strong LP for the original
instance, identify $T'$, and then solve the contracted instance with these
additional small-edge constraints.

Our construction proceeds by contracting one edge at a time.  
Suppose that $(x, y, \{y_\ell\}_{\ell \in L})$ is feasible for the Strong LP
for the instance $(T, L, w)$.  We then build a feasible solution
$(x', y', \{y'_\ell\}_{\ell \in L})$ for the contracted instance $(T', L, w)$
such that:
\begin{itemize}
    \item $T'$ is obtained from $T$ by contracting an edge
          $e=uv \in E$ that is not $\epsilon$-almost-always-small, and
    \item every edge in $T'$ that was $\epsilon$-always-small in the
          original LP solution remains $\epsilon$-always-small in the new
          solution.
\end{itemize}

We will use $\events,$ $\subtrees$ and $\Ext$ to denote the corresponding sets for
the original instance, and $\events',$ $\subtrees'$ and $\Ext'$ for the contracted
instance.  The key point is that the consistency constraints
(\Cref{constr:extension_consistency,constr:extension_consistency_with_ell})
allow the LP to assign value to events involving overlapping or nested
subtrees in $\subtrees$. 

This property is crucial for defining the LP solution in the contracted instance: it ensures that
when we contract an edge $e$, the LP can consistently assign value to
events in $\events'$ by effectively ``ignoring'' what happens at $e$.
In other words, the constraints preserve the structure of the fractional
solution across contractions, allowing us to define a feasible LP on the
smaller tree $T'$.

\newcommand{\expand}{\mathsf{expand}}
\paragraph{Defining our New LP solution.}
For ease of exposition, we use $\scrE, \scrE'$ etc. to denote events in the original LP and $\scrF, \scrF'$ etc. to denote events in the contracted LP.  
Throughout this discussion, we fix the edge $e=uv$ being contracted.

By slight abuse of notation we say that $e \in R(\scrF)$ whenever the super-node $s$ created by
contracting $e$ belongs to $R(\scrF)$, and that $e \notin R(\scrF)$ otherwise.

For any tree $R \in \subtrees'$, let $\expand(R)$ denote the corresponding
tree in the original instance obtained by expanding $s$ back into the edge
$e$.

If $e \notin R(\scrF)$, then $\scrF$ naturally corresponds to an element of
$\events$, and we simply set 
\[
  y'(\scrF) = y(\scrF), \qquad 
  y'_\ell(\scrF) = y_\ell(\scrF)
\] for all relevant links~$\ell$.

Now assume $e \in R(\scrF)$.  
Let $R_u$ and $R_v$ be the two subtrees of $\expand(R(\scrF))$ such that:
\begin{itemize}
    \item $R_u \cap R_v = \{e\}$,
    \item $R_u \cup R_v = \expand(R(\scrF))$, and
    \item $u$ is a leaf of $R_u$ and $v$ is a leaf of $R_v$.
\end{itemize}
For an illustration consult \Cref{fig:contraction_events}.
For each $S \in \scov^+(e)$ and $w \in \{u,v\}$,  
define the event $\scrE^w_S$ on $R_w$ that agrees with $\scrF$ on 
$R_w \setminus \{e\}$ and uses $S$ on $e$.  
Formally:

\begin{itemize}
    \item The edge set of $\scrE^w_S$ is $R_w$,  
          i.e., $R(\scrE^w_S) = R_w$.

    \item All small edges of $\scrF$ contained in $R_w$ remain small in 
          $\scrE^w_S$; moreover, if $S \in \scov(e)$, then $e$ is also 
          declared small by $\scrE^w_S$.

    \item The link set $L_\sm(\scrE^w_S)$ consists of all links in 
          $L_\sm(\scrF)$ relevant to $R_w$, together with the 
          links of $S$ whenever $S \in \scov(e)$.

    \item This definition of $\scrE^w_S$ is consistent with the intended 
          coverage: if $S \in \scov(e)$, then for each edge 
          $e' \in R_\sm(\scrE^w_S) \setminus \{e\}$, no link of $S$ outside 
          $L_\sm(\scrF)$ covers $e'$, and no link of $L_\sm(\scrF)$ outside $S$ 
          covers $e$.
\end{itemize}

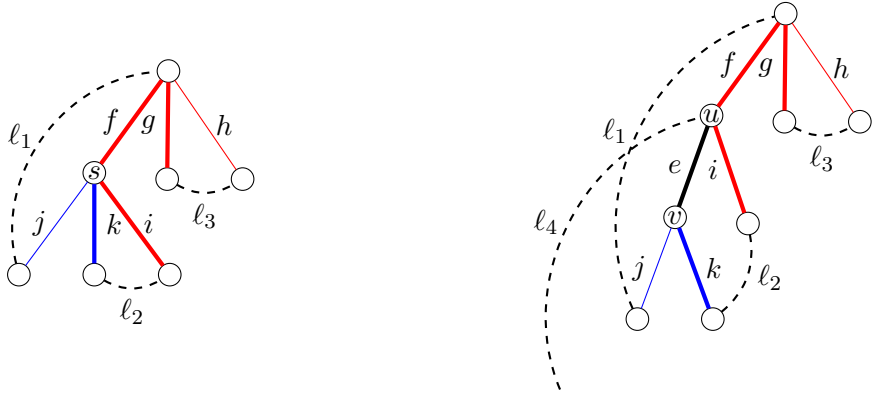
\begin{figure}[t]
    \centering
    \begin{minipage}{0.45\textwidth}
        \centering
        \begin{forest}
        for tree={
        circle, draw, minimum size=0.3cm, inner sep=0pt, outer sep=0pt,
        l sep=30pt, s sep=20pt,
        }
        [{}, name=r, fill=white,  
        [{$s$}, fill = white, name=u,  edge label={node[midway,left]{$f$}}, edge={ultra thick, draw=red}
                    [{}, fill=white, name=v, edge={draw=blue}, edge label={node[midway,left]{$j$}}]
                    [{}, fill=white, name=w, edge label={node[midway,right]{$k$}}, edge={ultra thick, draw=blue}]
                    [{}, fill=white, name=g, edge label={node[midway,right]{$i$}}, edge={ultra thick, draw=red}]
            ]
            [{}, fill=white, name=k, edge label={node[midway,left]{$g$}}, edge={ultra thick, draw=red}]
            [{}, fill=white , name=l, edge label={node[midway,right]{$h$}}, edge={draw=red}]
        ]
        \draw[dashed, thick, bend right] (k) to node[midway, below] {$\ell_3$} (l);
        \draw[dashed, thick, bend right] (w) to node[midway, below] {$\ell_2$} (g);
        \draw[dashed, thick, bend left=50] (v) to node[midway, left] {$\ell_1$} (r);
    \end{forest}
    \end{minipage}
    \begin{minipage}{0.45\textwidth}
        \centering
        \begin{forest}
        for tree={
        circle, draw, minimum size=0.3cm, inner sep=0pt, outer sep=0pt,
        l sep=30pt, s sep=20pt,
        }
        [{}, name=r, fill=white,  
        [{$u$}, fill = white, name=u,  edge label={node[midway,left]{$f$}}, edge={ultra thick, draw=red}
                    [{$v$}, fill=white, name=f, edge label={node[midway,left]{$e$}}, edge=ultra thick
                            [{}, fill=white, name=v, edge label={node[midway,left]{$j$}}, edge={draw=blue}]
                            [{}, fill=white, name=w, edge label={node[midway,right]{$k$}}, edge={ultra thick, draw=blue}]
                    ]
                    [{}, fill=white, name=g, edge label={node[midway,left]{$i$}}, edge={ultra thick, draw=red}]
            ]
            [{}, fill=white, name=k, edge label={node[midway,left]{$g$}}, edge={ultra thick, draw=red}]
            [{}, fill=white , name=l, edge label={node[midway,right]{$h$}}, edge={draw=red}]
        ]
        \draw[dashed, thick, bend right] (k) to node[midway, below] {$\ell_3$} (l);
        \draw[dashed, thick, bend right] (w) to node[midway, right] {$\ell_2$} (g);
        \draw[dashed, thick, bend left=50] (v) to node[midway, left] {$\ell_1$} (r);
        \draw[dashed, thick, bend left=50] (-3, -5) to node[midway, left] {$\ell_4$} (u);
    \end{forest}
        
    \end{minipage}

    \caption{Left: An event $\scrF$ defined on subtree $R(\scrF)$ (edges $f, g, h, i, j, k$) with $R_\sm(\scrF)$ being the thick edges $(f, g, i, k)$ and $L_\sm(\scrF)$ the dashed links ($\ell_1, \ell_2, \ell_3)$. Right: Expanding the supernode $s$ we obtain the tree $\expand(R(\scrF))$, where $R_u$ consists of $e$ plus the blue edges and $R_v$ consists of $e$ plus the red edges. Suitable choices of $S$ include $\{\ell_1, \ell_2\}$ and $\{\ell_1, \ell_2, \ell_4\}$, so that $\scrE_S^w$ will be well-defined for $w \in \{u, v\}$. }
    \label{fig:contraction_events}
\end{figure}

If the resulting event $\scrE^w_S$ is not well-defined,  
we set $y(\scrE^w_S) = 0$ and $y_\ell(\scrE^w_S) = 0$ 
for all relevant links~$\ell$. This occurs when there is a conflict between the edges in $S$ and $L_\sm(\scrF)$. 

Define:
\begin{equation}\label{eq:new_LP_vars}
    y'(\scrF) = \sum_{S \in \scov^+(e): y(e, S) > 0} y(\scrE_S^u) \cdot \frac{y(\scrE_S^v)}{y(e, S)}
\end{equation}
and for each link $\ell$ relevant to $\scrF$:
\begin{equation}
    y'_\ell(\scrF) =
    \begin{cases}
        \sum_{S \in \scov^+(e):y(\scrE_S^u)>0} y_\ell(\scrE^u_S) \cdot \frac{y(\scrE_S^v)}{y(e, S)} & \ell \in L(R_u)\\
        \sum_{S \in \scov^+(e):y(\scrE_S^v)>0} y_\ell(\scrE^v_S) \cdot \frac{y(\scrE_S^u)}{y(e, S)} & \text{ otherwise }.
    \end{cases}
%
\end{equation}

Finally, we set $x'$ equal to $x$.

\paragraph{Proof of Feasibility}
Since $x'=x$, the Odd Cut LP constraints (\Cref{constr:odd_cut_consraint}) are satisfied.
Since for each $e \in E$ and $\scrF \in \events'(e)$ we have $y'(\scrF)=y(\scrF)$ the
Coverage Constraints (\Cref{constr:cov_consistency}) are satisfied. Using the same reasoning the Marginal Preserving Constraints (\Cref{constr:marginal_pres_const_child}) are satisfied.
It is also clear that the Non-Negativity Constraints (\Cref{constr:non_neg}) are satisfied.  We focus our attention on proving the remaining constraints are satisfied. Namely we need to prove that \Cref{constr:inclision_on_small,eq:non_inclusion_links_on_huge_edges,constr:huge_edges_const_child,constr:extension_consistency,constr:extension_consistency_with_ell} hold.

To prove the following for each $\scrF \in \events'$ with $e \in R(\scrF)$ we will use $\scrE_S^w(\scrF)$ to define the events $\scrE_S^w$  for $w \in \{u, v\}$, used to define $y'(\scrF)$ and $y'_\ell(\scrF)$ for each link $\ell$ relevant to $\scrF$.

\begin{fact} For every event $\scrF \in \primeevents$ and link $\ell \in L(R(\scrF))$ \Cref{constr:inclision_on_small,eq:non_inclusion_links_on_huge_edges,constr:huge_edges_const_child} hold.
\end{fact}
\begin{proof}
    It suffices to prove the constraints hold for the case when $e \in \scrF$, since otherwise $y(\scrF)=y'(\scrF)$ and $y_\ell(\scrF)=y'_\ell(\scrF)$ for each relevant link $\ell$. 
    First we consider \Cref{constr:inclision_on_small}.
    Suppose that $\ell \in L_\sm(\scrF)$. Then, if $\ell \in L(R_u)$ we know that for each $S \in \scov^+(e)$ such that $y_\ell(\scrE_S^u) > 0$ we have that $y_\ell(\scrE^u_S)=y(\scrE^u_S)$ by feasibility of the original solution, and the fact that $\ell \in L_\sm(\scrE_S^u)$. Therefore 

    $$y_\ell'(\scrF) = \sum_{S \in \scov^+(e):y(\scrE_S^u)>0} y_\ell(\scrE^u_S) \cdot \frac{y(\scrE_S^v)}{y(e, S)}  =  \sum_{S \in \scov^+(e):y(\scrE_S^u)>0} y(\scrE^u_S) \cdot \frac{y(\scrE_S^v)}{y(e, S)}= y'(\scrF).$$ An analogous argument can be made for when $\ell \in L(R_v) \setminus L(R_u)$.
    Similar arguments show \Cref{eq:non_inclusion_links_on_huge_edges,constr:huge_edges_const_child} hold and are omitted.

\end{proof}

\begin{fact} For every event $\scrF = (R, R_\sm, L_\sm) \in \events$ and every 
$Q \in \subtrees$ satisfying $R \subseteq Q$, 
~\Cref{constr:extension_consistency} is satisfied, and for each link every $\ell \in L(R(\scrF))$
the constraint~\Cref{constr:extension_consistency_with_ell} is satisfied.
\end{fact}

\begin{proof}
    We prove the claim for \Cref{constr:extension_consistency}. The argument for \Cref{constr:extension_consistency_with_ell}  is analogous and is omitted.
    If $e \notin \expand(Q)$, the constraint holds by feasibility of the input 
    solution. In the remainder, assume $e \in  \expand(Q)$. In a similar fashion we let $Q_u$ and $Q_v$ be the subtrees of $Q$ such that:
    \begin{itemize}
    \item $Q_u \cap Q_v = \{e\}$,
    \item $Q_u \cup Q_v = Q$, and
    \item $u$ is a leaf of $Q_u$ and $v$ is a leaf of $Q_v$.
\end{itemize}

\paragraph{Case 1: $e \notin R(\scrF)$.}
    Then $R(\scrF) \subseteq Q_w$ for some $w \in \{u,v\}$; without loss of generality, let $w = u$.  

Partition $\Ext(\scrF, Q_u)$ into $A$ and $B$, 
where $A$ consists of those events $\scrE$ with $e \in R(\scrE)$.  
Similarly, partition $\Ext'(\scrF, Q)$ into $A'$ and $B'$, 
where $A'$ contains those events $\scrF'$ with $e \in R(\scrF')$

    Then 
    \begin{align*}
        y'(\scrF) 
        & = \sum_{\scrE \in A \cup B} y(\scrE) \\
        & = \sum_{\scrE \in A} \left( y(\scrE) \sum_{\scrE' \in \Ext((e, L_\sm(\scrE) \cap L_e), Q_v)} \frac{y(\scrE')}{y(e, L_\sm(\scrE) \cap L_e)}\right)+ \sum_{\scrE \in B} y(\scrE)\\
        & = \sum_{\scrF' \in A'} \sum_{S \in \scov^+(e)} y(\scrE^u_S(\scrF')) \cdot \frac{y(\scrE^v_S(\scrF'))}{y(e, S)} + \sum_{\scrF' \in B'} y'(\scrF')\\
        & = \sum_{\scrF' \in A'} y'(\scrF') + \sum_{\scrF' \in B'} y'(\scrF')\\
        & = \sum_{\scrF' \in \Ext'(\scrF, Q)} y'(\scrF').
    \end{align*}
    The first inequality follows by feasibility of the original solution. Formally we have
\[
  y'(\scrF) = y(\scrF)
            = \sum_{\scrE \in \Ext(\scrF, Q_u)} y(\scrE).
\]
    The second equality multiplies each term in the first sum by one; 
this is valid because the input solution satisfies 
\Cref{constr:extension_consistency} for the event 
$(e, L_\sm(\scrE) \cap L_e)$ and tree $Q_v$. The third equality groups terms based on their correspondence to elements in $A'$ 
and the final equality applies the definition of $y'(\scrF')$ 
for each resulting event~$\scrF'$.

\paragraph{Case 2: $e \in R(\scrF)$.}
We have
    \begin{align*}
        y'(\scrF) & = \sum_{S \in \scov^+(e): y(e, S) > 0} y(\scrE_S^u(\scrF)) \cdot \frac{y(\scrE_S^v(\scrF))}{y(e, S)}\\
        & = \sum_{S \in \scov^+(e): y(e, S) > 0} \left(\sum_{\scrE' \in \Ext(\scrE_S^u(\scrF), Q_u)} y(\scrE') \cdot \frac{\sum_{\scrE' \in \Ext(\scrE_S^v(\scrF), Q_v)}y(\scrE')}{y(e, S)} \right)\\
        & =  \sum_{\scrF' \in \Ext(\scrF, Q)} \sum_{S \in \scov^+(e): y(e, S) > 0} y(\scrE_S^u(\scrF')) \cdot \frac{y(\scrE_S^v(\scrF'))}{y(e, S)}\\
        & = \sum_{\scrF' \in \Ext(\scrF, Q)} y'(\scrF)
    \end{align*}
Our reasoning follows a similar template to the previous case.
    The first equality follows from the definition of $y'(\scrF)$ in
\Cref{eq:new_LP_vars}.  
The second applies the consistency constraints for the input solution,
namely~\Cref{constr:extension_consistency} for each 
$\scrE^w_S(\scrF)$, $Q_w$, where 
$w \in \{u,v\}$ and $S \in \scov^+(e)$.  
The third equality groups terms based on their correspondence to elements in $\Ext(\scrF, Q)$. The fourth inequality follows by applying the definition of $y'(\scrF')$ for each $\scrF' \in  \Ext(\scrF, Q)$.
\end{proof}

\subsection{Pre-processing of Almost Always Small Solutions}\label{subsec:preprocessing_removing_correlations}

To ensure that we can convert our $\epsilon$-almost-always-small solution 
into the desired Structured Fractional Solution, we first preprocess it so 
that any given edge shares links with only a constant number of its children.  
Namely, $e_v$ will share links only with its children in $E^*(v)$, where 
$E^*(v)$ is the subset of edges in $E(v)$ consisting of $e_v$ together with 
all edges correlated with $e_v$ (with respect to $x$).
We recall that for each $v \in V$ the subset of edges in \(E(v)\) consists of \(e_v\) 
together with all edges that are pairwise correlated with \(e_v\) (with respect to $x$).
The main result of this section is the following. 

\begin{lemma}\label{lem:removing_overlap_uncorrelated}
Given any solution $(x, y, \{y_\ell\}_{\ell \in L})$ to the Strong LP for a 
WTAP instance, we can compute another solution 
$(x', y', \{y'_\ell\}_{\ell \in L})$ to the same LP such that:
\begin{enumerate}
    \item for each $v \in V$, 
    $x'(L_{e_v}) \cap x'(L_e) = \emptyset$ for all 
    $e \in E(v) \setminus E^*(v)$, and
    \item 
    \[
       c(x'(L)) = (1 + o_\epsilon(1)) \cdot c(x(L)).
    \]
\end{enumerate}
\end{lemma}

We call a link $\ell \in L$ \textbf{dangerous} if 
$\ell \in L_{e_v} \cap L_e$ for some 
$e \in E(v) \setminus E^*(v)$, and \textbf{safe} otherwise.  
The argument proceeds by iteratively replacing a dangerous link with possibly 
a shorter dangerous link and a collection of safe links.
In this section, we describe the proof of the lemma that is intuitive. For a fully formal argument with all calculations, we refer the reader to \Cref{apx:formal_preprocessing}.
\paragraph{Description.} 
While there exists a dangerous link in the support of $x'$, the algorithm 
selects the lowest uncorrelated edge $f = uv$ whose parent edge $e_u$ is covered 
by a link $\ell$ in the support of $x'$ that also covers $f$.

We then split $\ell$ at $u$ to create two shadows, $\ell_1$ and $\ell_2$, 
of $\ell$ (covering $e_u$ and $f$, respectively), where $w(\ell_1) = w(\ell)$ 
and $w(\ell_2)$ is defined as the cost of the cheapest up-link solution $H$ 
that:
\begin{itemize}
    \item covers each edge of $P_{\ell_2}$ exactly once, and
    \item covers no edge outside $P_{\ell_2}$.
\end{itemize}
Consult the left image of \Cref{fig:link_splitting} for a visualization of how $\ell_1$ and $\ell_2$ are obtained.

For any event~$\scrF$ with $y'_\ell(\scrF) > 0$, 
we shift the corresponding LP weight from~$\scrF$ 
to a new event~$\scrF'$ obtained by replacing~$\ell$ 
with the relevant links among~$\ell_1$ and~$\ell_2$.  
In particular:
\begin{itemize}
    \item If both~$\ell_1$ and~$\ell_2$ are relevant for~$\scrF$, 
          we replace~$\ell$ by both.
    \item Otherwise, we replace~$\ell$ only by the relevant shadow.
\end{itemize}

Formally, if $\ell \in L_\sm(\scrF)$, 
we construct $\scrF'$ by removing $\ell$ from $L_\sm(\scrF)$ 
and adding the links among $\ell_1$ or $\ell_2$ that cover an edge in $R_\sm(\scrF)$.  
Intuitively, $\scrF'$ represents the same configuration as~$\scrF$, 
except that the link~$\ell$ has been replaced by its relevant shadow(s).

To update the LP variables consistently, 
we create a corresponding copy~$\scrF'_{\text{new}}$ of each modified event~$\scrF'$.  
This copy serves purely as a bookkeeping device: 
it accumulates the LP weight that is being shifted from~$\scrF$ to~$\scrF'$.

Suppose that $y'(\scrF)=\eta$ and 
$y'_{\ell'}(\scrF)=\eta_{\ell'}$ for each link~$\ell'$ relevant to~$\scrF$.  
We perform the following updates:
\begin{itemize}
    \item Increase $y'(\scrF'_{\text{new}})$ by~$\eta$.
    \item For every link $\ell'$ relevant to~$\scrF$ other than~$\ell$, 
          increase $y'_{\ell'}(\scrF'_{\text{new}})$ by~$\eta_{\ell'}$.
    \item For the link~$\ell$, increase 
          $y'_{\ell'}(\scrF'_{\text{new}})$ by~$\eta_{\ell}$ 
          for each relevant shadow $\ell' \in \{\ell_1, \ell_2\}$.
\end{itemize}

Finally, we set all variables associated with the original event~$\scrF$ 
to zero, effectively replacing it by~$\scrF'_{\text{new}}$.  
We also update the link variables by transferring the total weight of~$\ell$ 
to its shadows:
\[
    x'(\ell_1) \leftarrow x'(\ell_1) + x'(\ell), 
    \qquad
    x'(\ell_2) \leftarrow x'(\ell_2) + x'(\ell),
    \qquad
    x'(\ell) \leftarrow 0.
\]

The final LP solution is obtained by summing the LP values 
assigned to both the original and copied events, 
ensuring that the total mass of the distribution and the objective value 
remain unchanged.  
In summary, this transformation redistributes LP weight from events 
involving~$\ell$ to those involving its shadows~$\ell_1$ and~$\ell_2$, 
while preserving feasibility and cost.  
An equivalent formulation—avoiding explicit event copies—is provided in 
\Cref{apx:formal_preprocessing}.

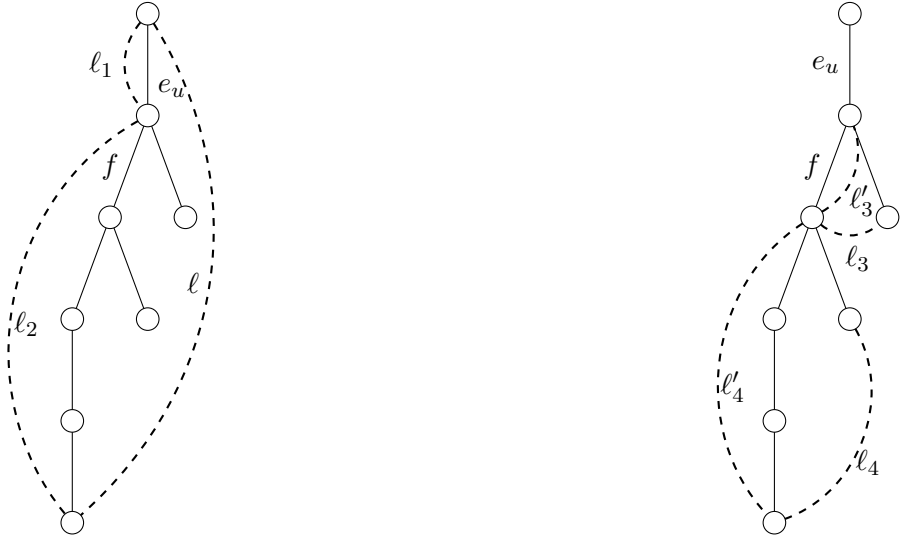
\begin{figure}[t]
\centering
\begin{minipage}{0.45\textwidth}
\centering
\begin{forest}
for tree={
    circle, draw, minimum size=0.3cm, inner sep=0pt, outer sep=0pt,
    l sep=30pt, s sep=20pt
}
[{}, name=root, fill=white
    [{}, name=a, fill=white, edge label={node[midway,right, yshift=-9pt]{$e_u$}}, edge={draw=black}
        [{}, name=r, fill=white, edge label={node[midway,left]{$f$}}, edge={draw=black}
            [{}, fill=white, name=u, edge={draw=black}
                [{}, fill=white, name=v, edge={draw=black}
                    [{}, fill=white, name=b, edge={draw=black}]
                ]
            ]
            [{}, fill=white, name=k, edge={draw=black}]
        ]
        [{}, fill=white, name=c, edge={draw=black}] 
    ]
]
\draw[dashed, thick, bend left=40] (root) to node[midway,left] {$\ell$} (b);
\draw[dashed, thick, bend right=40] (root) to node[midway,left] {$\ell_1$} (a);
\draw[dashed, thick, bend right=50] (a) to node[midway,below, xshift=5pt] {$\ell_2$} (b);
\end{forest}
\end{minipage}
\hfill
\begin{minipage}{0.45\textwidth}
\centering
\begin{forest}
for tree={
    circle, draw, minimum size=0.3cm, inner sep=0pt, outer sep=0pt,
    l sep=30pt, s sep=20pt
}
[{}, name=root, fill=white
    [{}, name=a, fill=white, edge label={node[midway,left]{$e_u$}}, edge={draw=black}
        [{}, name=r, fill=white, edge label={node[midway,left]{$f$}}, edge={draw=black}
            [{}, fill=white, name=u, edge={draw=black}
                [{}, fill=white, name=v, edge={draw=black}
                    [{}, fill=white, name=b, edge={draw=black}]
                ]
            ]
            [{}, fill=white, name=k, edge={draw=black}]
        ]
        [{}, fill=white, name=c, edge={draw=black}] 
    ]
]
\draw[dashed, thick, bend right=40] (r) to node[midway,below, xshift=3pt] {$\ell_3$} (c);
\draw[dashed, thick, bend left=50] (k) to node[midway,below, xshift=2pt] {$\ell_4$} (b);
\draw[dashed, thick, bend right=40] (r) to node[midway,below, xshift=3pt] {$\ell_3'$} (a);
\draw[dashed, thick, bend right=50] (r) to node[midway,below, xshift=5pt] {$\ell_4'$} (b);
\end{forest}
\end{minipage}

\caption{Left: Illustration of the link-splitting step.
The original link~$\ell$ covers both~$e_u$ and~$f$, and the algorithm replaces it with two shadows,~$\ell_1$ and~$\ell_2$, that together preserve the same coverage.
Right: Illustration of the set~$M(f)$.
The link~$\ell$ is not contained in~$M(f)$ since it covers both~$e_u$ and~$f$.
The link~$\ell_3$ is a correlated link covering only edges that are descendants of~$f$, while~$\ell_4$ is an uncorrelated link with~$f$ as a leading edge; both are included in~$M(f)$.
The cost of~$\ell_2$ is bounded by~$w(\ell_3) + w(\ell_4)$, corresponding to the cost of the up-link solution composed of their shadows~$\ell_3'$ and~$\ell_4'$.}
\label{fig:link_splitting}
\end{figure}

\paragraph{Bounding the Cost of $\ell_2$.}
We now argue that the cost of $\ell_2$ is not too large for our purposes.  
Recall that $\ell_2$ is obtained by splitting $\ell$ at $u$ so that it covers only
edges on $P_{\ell_2}$, that is, the path of edges below $f$.

We construct the cheapest possible up-link solution $H$ covering $P_{\ell_2}$
using only shadows of the following types of links:
\begin{itemize}
    \item \textbf{Correlated links} that cover only edges that are descendants of~$f$.
    \item \textbf{Uncorrelated links} for which~$f$ is one of the leading edges.
\end{itemize}
Intuitively, these are exactly the links that can be \emph{charged to}~$f$ when it is
selected by the algorithm.  
Let $M(f)$ denote this set of eligible links. Consult the right image of \Cref{fig:link_splitting} for a visualization of $M(f)$.

Observe that:
\begin{itemize}
    \item For any two distinct edges $f, f'$, 
          the sets $M(f) \cap L_C$ and $M(f') \cap L_C$ are disjoint, 
          since a correlated link can only be charged to the unique highest
          uncorrelated ancestor of its uncorrelated covered edges.
    \item Each uncorrelated link appears in at most two such sets,
          corresponding to its two leading edges.
\end{itemize}

Hence, the collection $\{M(f)\}_{f \in E}$ forms a bounded-overlap decomposition
of the link set, similar to how the Odd Cut LP partitions coverage
contributions across edges.

Moreover, since $f$ and $e_v$ are uncorrelated, for any edge $e \in P_{\ell_2}$
we have
\[
    x(M(f) \cap L_e) \ge 1 - \delta.
\]
That is, the links in $M(f)$ already provide $(1 - \delta)$-fractional
coverage of every edge on the path $P_{\ell_2}$.

By the integrality of the Cut LP restricted to up-links,
a feasible integral up-link solution $H$ covering $P_{\ell_2}$ 
can therefore be obtained at cost at most
\[
    \frac{2 \, c(x(M(f)))}{1 - \delta}.
\]
This follows from the 2-approximation property of the Odd Cut LP 
when restricted to up-links.

Since $x(L_{e_u} \cap L_f) < \delta$, 
the total increase in cost associated with splitting $\ell$
is at most a $\delta$-fraction of this amount.
Hence, the overall penalty in total cost is bounded by
\[
    \frac{2\delta}{1 - \delta} \, c(x(L)).
\]

Finally, because $\ell$ is replaced by the pair $(\ell_1, \ell_2)$,
whose combined coverage exactly matches that of $\ell$,
the modified LP solution remains feasible.
A complete formal argument of feasibility and the corresponding
pseudo-code are provided in \Cref{apx:formal_preprocessing}.
.

\subsection{Reduction from Almost Always Small to Structured Fractional Solution}
We now focus on designing an algorithm, which we call the
\emph{Marking Algorithm}, that converts an
$\epsilon$-almost-always-small solution into a
\emph{Structured Fractional Solution}, incurring only a
negligible loss.

At a high level, an $\epsilon$-almost-always-small solution
already exhibits the desired structure for each edge on all
but an $\epsilon$ fraction of the underlying distribution.
The main challenge is therefore to remove this
probability mass while approximately preserving the original
marginal distributions.

The Marking Algorithm addresses this challenge by
\emph{marking} correlated events from the LP solution, through a sampling procedure and
making repairs whenever a large cover set is present.
The main result of the section is formalized in the following lemma.

\begin{lemma}\label{lem:esp_alm_sm_struc_frac}
Let $(x, y, \{y_\ell\}_{\ell \in L})$ be an $\epsilon$-almost-always-small solution to the Strong LP for a WTAP instance. 
Suppose that for every $v \in V$ and every edge $e \in E(v) \setminus E^*(v)$, it holds that 
\[
    x(L_{e_v} \cap L_e) = 0.
\]
Then there exists an algorithm that computes a Structured Fractional Solution $(x', y', z')$ for the same instance, with total cost
\[
    c(z'(L)) \;\le\; (1 + o_\epsilon(1)) \cdot c(x(L)).
\]
\end{lemma}

\Cref{lem:reduction_esp_small_always_small} then follows by applying \Cref{lem:esp_alm_sm_struc_frac} to the solution guaranteed by \Cref{lem:removing_overlap_uncorrelated}.
For the remainder of this section, we fix the input solution $(x, y, \{y_\ell\}_{\ell \in L})$, and all subsequent results are stated with respect to this solution.

\paragraph{Outline}

Intuitively, since every edge is
$\epsilon$-almost-always small with respect to the input
solution, we can modify the solution so that it is
\emph{always} small---from which a
\emph{Structured Fractional Solution} follows---while
incurring only a small loss.
The formal proof, however, is technical.
To guide the reader, we outline the remainder of this
section below.

\begin{itemize}
    \item \Cref{subsubsec:warm_up_always_small} provides a
    warm-up algorithm illustrating how the Marking Algorithm
    simplifies when no probability mass is assigned to large
    sets.

    \item \Cref{sec:desc_subtree_alg} introduces the Marking
    Algorithm, which samples events from the input
    $\epsilon$-almost-always-small solution
    $(x, y, \{y_\ell\}_{\ell \in L})$.

    \item Each event $\scrE$ in the Structured Fractional
    Solution is assigned a value equal to the probability
    that it is \emph{marked} by the algorithm.

    \item In \Cref{subsubsec:feas_struc_frac}, we prove that
    the resulting solution is a feasible Structured
    Fractional Solution.

    \item In \Cref{sec:resampling_distributions}, we show
    that the algorithm is approximately
    marginal-preserving, which ensures that the increase in
    cost is negligible.

    \item Finally,
    \Cref{subsubsec:resampling_dist,subsubsec:up_link_soln_huge_edges}
    establish key lemmas used in the analysis of
    \Cref{sec:resampling_distributions}.
\end{itemize}

For the remaining for each $v \in V$ we will use $E^*(v)$ and $E^*(e_v)$ as well as $E(v)$ and $E(e_v)$ interchangeably.

\subsubsection{Warm-Up: The Always-Small Case}\label{subsubsec:warm_up_always_small}
Before presenting the full Structured-Sampling Algorithm, it is helpful to 
consider a simplified setting in which no probability mass is assigned to 
events that cover an edge with many links.  Formally, we assume that 
$y(\scrF) = 0$ for every event $\scrF$ containing an edge $e$ with 
$\scrF = (\{e\}, \emptyset, \emptyset)$.  In this case we call the input 
solution $(x, y, \{y_\ell\}_\ell)$ \emph{always-small}.  

For an always-small solution every event $\scrF = (R, R_{\sm}, L_{\sm})$ 
satisfies $R = R_{\sm}$, so the Strong-LP notation becomes unnecessarily 
heavy.  In this warm-up we therefore adopt the notation of the Structured LP: 
we write an event as $\scrE = (E', L')$, where $E' = R$ and $L' = L_{\sm}$.  
We also use $E(\scrE)$ and $L(\scrE)$ to denote these components.  

The full Structured-Sampling Algorithm in \Cref{sec:desc_subtree_alg} 
returns to the Strong-LP notation $\scrF = (R, R_{\sm}, L_{\sm})$, since 
in the general case $R$ and $R_{\sm}$ need not coincide.

For an always-small solution the algorithm performs no modifications: the 
Structured-Sampling Algorithm simply returns $(x, y)$ 
unchanged (where here we restrict the support to the events of the Structured LP).  The general algorithm of \Cref{sec:desc_subtree_alg} includes 
additional steps to correct the rare occurrences of large cover sets, but 
reduces to this behaviour when the input solution is always-small.  

Before presenting the warm-up, we introduce some notation.
\paragraph{Notation (Ancestry Relationship between Edges $E$).} 
Up to this point we have referred only to the ancestry relation between 
vertices of the tree. We now extend this notion to edges. For each vertex 
$v\in V$, let $E(v) \setminus \{e_v\}$ denote the set of child edges of $e_v$, i.e., the edges 
whose parent endpoint is $v$. Moreover, for two edges $e = uv$ and $e' = u'v'$ (with $u$ and $u'$ the respective parent endpoints), we say that 
$e$ is a descendant of $e'$ whenever $u$ is a descendant of $u'$ in the 
vertex ancestry relation.

\paragraph{Warm-Up Algorithm.}
The warm-up algorithm proceeds as follows.
We maintain a queue $X$ containing pairs $(f, S)$, where
$f$ is a tree edge and $S$ is the small set of links
currently assigned to cover~$f$.
When a pair $(f, S)$ enters the queue, all ancestors of $f$
(including $f$ itself) have already been assigned their
covering sets, and it remains only to assign covering sets
to the descendant edges of $f$.

To process the tree uniformly in a top-down manner, we
introduce a dummy \emph{parent edge} $e_r$.
The algorithm begins by inserting $(e_r, \emptyset)$ into
$X$.

\medskip
While the queue is non-empty, we remove a pair $(f, S)$ and
process it as follows.
Since the solution is always-small, every event in
$\events(E^*(f))$ has the form $\scrE = (E', L')$ with
$E' = E^*(f)$ and $L'$ the small set of links covering these
edges.
We sample such an event
$\scrE \in \Ext((f, S), E^*(f))$ with probability
\[
    \frac{y(\scrE)}{y(f, S)},
\]
and we mark~$\scrE$.
(The consistency constraints
\Cref{constr:extension_consistency} ensure that this defines
a valid probability distribution.)

\medskip
\textbf{Processing correlated children.}
For each child edge $g$ of $f$ within $E^*(f)$, we mark the
event $(\{g\}, L(\scrE) \cap L_g)$ and add
$(g, L(\scrE) \cap L_g)$ to the queue.

\medskip
\textbf{Processing uncorrelated children.}
For each child edge $g$ of $f$ in $E(f) \setminus E^*(f)$,
we sample an event
$\scrE_g \in \Ext(\scrE, E^*(f) \cup \{g\})$ with probability
\[
    \frac{y(\scrE_g)}{y(\scrE)},
\]
mark~$\scrE_g$, and then add
$(g, L(\scrE_g) \cap L_g)$ to the queue.

\medskip
\textbf{Processing events with two uncorrelated children.}
For any pair of uncorrelated children $g, g'$ of $f$, we
similarly sample an event
$\scrE_{g, g'} \in
 \Ext(\scrE, E^*(f) \cup \{g, g'\})$
with probability
\[
    \frac{y(\scrE_{g, g'})}{y(\scrE)},
\]
and we mark~$\scrE_{g, g'}$. In the case that $\{f\} \neq E^*(f)$ we also mark $(f, S)$.

For a visualization of the warm-up algorithm consult \Cref{fig:warm_up_algorithm}.
We note here that the events $\scrE_g$ and $\scrE_{g, g'}$ will not necessarily be consistent. This will not pose any problem. In particular, we will use the expected output of the algorithm to define a fractional solution, whose values will be consistent.

Next let $(x', y', z')$ be the Structured Fractional Solution such that for each event $\scrE$ of the Structured LP we let $y'(\scrE)$ be the probability that $\scrE$ is marked by the algorithm. We then set both $x'(\ell)$ and $z'(\ell)$ equal to $x(\ell)$ for each like $\ell \in L$.

\begin{figure}[htp]
    \centering
     \begin{minipage}{0.3\textwidth}
        \centering
        \begin{forest}
        for tree={
        circle, draw, minimum size=0.3cm, inner sep=0pt, outer sep=0pt,
        l sep=30pt, s sep=20pt,
        }
        [{}, name=a, fill=white
            [{}, name=r, fill=white, edge label={node[midway,left]{$f$}}, edge={ultra thick, draw=black}
                 [{}, fill = white, name=u,
                    edge label={node[midway,left, yshift=-3pt]{$e_1$}},
                    edge={ultra thick, draw=black}
                 ]
                 [{}, fill=white, name=k,
                    edge label={node[midway,right, yshift=-3pt]{$e_2$}},
                    edge={ultra thick, draw=black}
                 ]
                 [{}, fill=white, name=l,
                    edge label={node[midway,right, yshift=-3pt]{$e_3$}},
                    edge={ultra thick, draw=red}
                 ]
                 [{}, fill=white, name=m,
                    edge label={node[midway,right, yshift=-3pt]{$e_4$}},
                    edge={ultra thick, draw=red}
                 ]
            ]
        ]
        \draw[dashed, thick, bend right] (a) to node[midway, right] {$\ell_1$} (-3, -3);
        \draw[dashed, thick, bend right] (k) to node[midway, below] {$\ell_3$} (2, -4);
        \draw[dashed, thick, bend left=50] (-2, -4) to node[midway, left] {$\ell_2$} (r);
\end{forest}

    \end{minipage}%
    \hfill
     \begin{minipage}{0.3\textwidth}
        \centering
        \begin{forest}
        for tree={
        circle, draw, minimum size=0.3cm, inner sep=0pt, outer sep=0pt,
        l sep=30pt, s sep=20pt,
        }
        [{}, name=a, fill=white
            [{}, name=r, fill=white, edge label={node[midway,left]{$f$}}, edge={ultra thick, draw=black}
                 [{}, fill = white, name=u,
                    edge label={node[midway,left, yshift=-3pt]{$e_1$}},
                    edge={ultra thick, draw=black}
                 ]
                 [{}, fill=white, name=k,
                    edge label={node[midway,right, yshift=-3pt]{$e_2$}},
                    edge={ultra thick, draw=black}
                 ]
                 [{}, fill=white, name=l,
                    edge label={node[midway,right, yshift=-3pt]{$e_3$}},
                    edge={ultra thick, draw=black}
                 ]
                 [{}, fill=white, name=m,
                    edge label={node[midway,right, yshift=-3pt]{$e_4$}},
                    edge={ultra thick, draw=red}
                 ]
            ]
        ]
        \draw[dashed, thick, bend right] (a) to node[midway, right] {$\ell_1$} (-3, -3);
        \draw[dashed, thick, bend right] (k) to node[midway, below] {$\ell_3$} (2, -4);
        \draw[dashed, thick, bend left=50] (-2, -4) to node[midway, left] {$\ell_2$} (r);
\end{forest}

    \end{minipage}%
    \hfill
     \begin{minipage}{0.3\textwidth}
        \centering
        \begin{forest}
        for tree={
        circle, draw, minimum size=0.3cm, inner sep=0pt, outer sep=0pt,
        l sep=30pt, s sep=20pt,
        }
        [{}, name=a, fill=white
            [{}, name=r, fill=white, edge label={node[midway,left]{$f$}}, edge={ultra thick, draw=black}
                 [{}, fill = white, name=u,
                    edge label={node[midway,left, yshift=-3pt]{$e_1$}},
                    edge={ultra thick, draw=black}
                 ]
                 [{}, fill=white, name=k,
                    edge label={node[midway,right, yshift=-3pt]{$e_2$}},
                    edge={ultra thick, draw=black}
                 ]
                 [{}, fill=white, name=l,
                    edge label={node[midway,right, yshift=-3pt]{$e_3$}},
                    edge={ultra thick, draw=black}
                 ]
                 [{}, fill=white, name=m,
                    edge label={node[midway,right, yshift=-3pt]{$e_4$}},
                    edge={ultra thick, draw=black}
                 ]
            ]
        ]
        \draw[dashed, thick, bend right] (a) to node[midway, right] {$\ell_1$} (-3, -3);
        \draw[dashed, thick, bend right] (k) to node[midway, below] {$\ell_3$} (2, -4);
        \draw[dashed, thick, bend left=50] (-2, -4) to node[midway, left] {$\ell_2$} (r);
        \draw[dashed, thick, bend right=50] (l) to node[midway, left, xshift=8pt, yshift=-4pt] {$\ell_4$} (m);
\end{forest}

    \end{minipage}
    \caption{Warm-Up Visualization: Assume the other endpoint of $\ell_1$ is below $f$, the other endpoint of $\ell_2$ is below $e_1$, and the other endpoint of $\ell_3$ is below $e_3$.  
    Let $E^*(f) = \{f, e_1, e_2\}$ and $E(f) \setminus E^*(f) = \{e_3, e_4\}$. 
    In each figure, the black edges form the edge set of the event being considered.\\
    \textbf{Left: }
    The pair $(f, \{\ell_1\})$ is removed from $X$ and the algorithm samples the event 
    $\scrE = (E^*(f), \{\ell_1, \ell_2, \ell_3\})$.\\
    \textbf{Middle: }
    The algorithm samples 
    $\scrE_{e_3} = (E^*(f) \cup \{e_3\}, \{\ell_1, \ell_2, \ell_3\}) 
    \in \Ext(\scrE, E^*(f) \cup \{e_3\})$.\\
    \textbf{Right: }
    The algorithm samples 
    $\scrE_{e_3, e_4} = (E^*(f) \cup \{e_3, e_4\}, \{\ell_1, \ell_2, \ell_3, \ell_4\}) 
    \in \Ext(\scrE, E^*(f) \cup \{e_3, e_4\})$.  Note that here the events $\scrE_{e_3}$ and $\scrE_{e_3, e_4}$ are not consistent.}
\label{fig:warm_up_algorithm}
\end{figure}
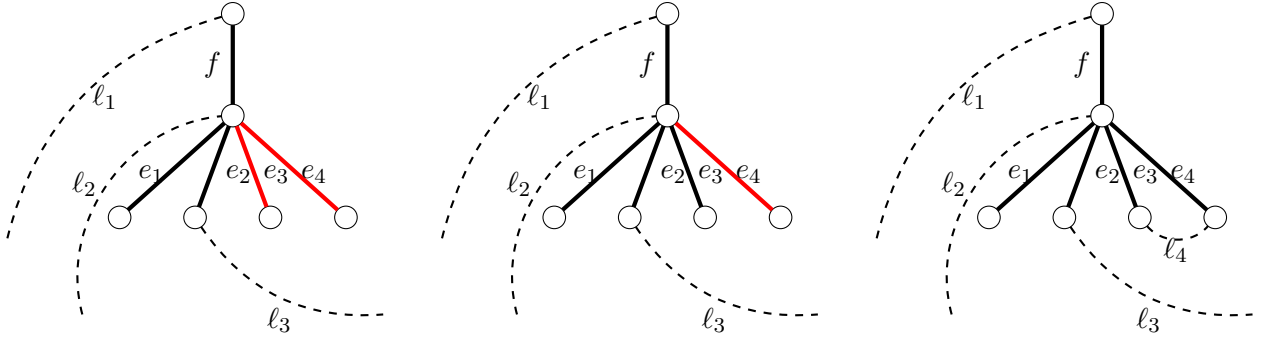

\medskip
\textbf{Correctness.}
Let $\scrE$ be any event of the Structured LP, and let $f$
denote the edge of $E(\scrE)$ incident to the highest vertex
of the event, i.e., the vertex closest to the root.
A straightforward induction on the distance of $f$ from
the dummy root edge $e_r$ shows that the algorithm marks
$\scrE$ with probability exactly $y(\scrE)$.

Hence, in the always-small case, the algorithm reproduces
the input marginals without modification, and no changes
to $x$ are required.

We remark that this solution is indeed a Structured
Fractional Solution, which follows directly from the
feasibility of $(x, y, \{y_\ell\}_{\ell \in L})$ for the
Strong LP.

With this intuition in hand, we now turn to the general algorithm.

\subsubsection{Description of the Marking Algorithm}
\label{sec:desc_subtree_alg}
Before presenting the full Structured-Sampling Algorithm, we introduce 
additional notation that will be used throughout this section.

\paragraph{Notation (subtree $T_f$).} 
For a tree edge $f \in E$, we define $T_f$ to be the subtree of $T$ consisting of the edge $f$ and all of its descendant edges in $T$ that are pairwise correlated with $f$ (with respect to $x$).
We define $T_{e_r}$ to be the trivial subtree consisting only of $e_r$ (as there are no descendant edges that are correlated with $e_r$).

The following is key to our sampling procedure as we will need to ensure that $T_f$ is indeed a member of $\subtrees$ so that we can sample from $\events(T_f)$.
\begin{observation}\label{obs:T_f_constant_leaves}
    For each $f \in E$ the tree $T_f$ has at most $\beta+1$ leaves.
\end{observation}
\begin{proof}
Since each such edge $f \in E$ is $\epsilon$-almost-always-small that $$x(L_f) \leq \epsilon \rho.$$
Since there are at most $\frac{x(L_f)}{\delta}=\beta$ leaves in $T_f$ distinct from $f$, we know that $T_f$ has at most $\beta+1$ leaves and is therefore contained in $\subtrees$.
\end{proof}

Throughout this section, we use $\scov$, $\events$, $\subtrees$, and $\Ext$
to denote the corresponding sets in the Strong LP, and
$\scov'$, $\events'$, and $\stars'$ to denote those in the
Structured Fractional Solution.
When $z'$ is not explicitly defined, we will also refer to this
solution as the Structured LP.
In addition, recall that for a link $\ell=\{u,v\}$ we denote by $\ell_u$ and 
$\ell_v$ the two up-links satisfying $\ell_u =\{u,\apex(\ell)\}$ and $\ell_v=\{v,\apex(\ell)\}.$
For any leading edge $e$ of $\ell$, we write $\ell_e$ to denote the unique
up-link in $\{\ell_u,\ell_v\}$ that covers $e$.

The algorithm proceeds by randomly selecting events from the Strong LP
solution and associating them with events of the Structured LP.
We show that assigning to each Structured-LP event a value equal to the
probability with which it is marked by the algorithm yields a solution
$(x', y')$ that approximately preserves the original marginals.
From this, we construct a Structured Fractional Solution by first setting
$z' = x$ and then augmenting $z'$ in accordance with $x'$ to satisfy the
required properties, while maintaining the cost bound.

We remark that not all Structured-LP events appear explicitly in the
algorithm. In particular, for each $v \in V$, we exclude events defined
over the edge set $E^*(v)$ together with any two uncorrelated edges in
$E(v)$; we refer to these as \emph{large} events.
 All other events are marked by the algorithm and are called \emph{small} events. 
These cases will be handled
in \Cref{subsubsec:feas_struc_frac}, where we define them explicitly so that consistency constraints (\Cref{eq:constr:small_prob_1:consistency}) hold. This differs from the exposition of the warm-up, where we directly sample and mark large events. The difference arises from a technical difficulty caused by the ``conditional resets'' performed during the algorithm. In the warm-up algorithm, when removing $(f, S)$ from the queue, we determine precisely which small events involving $f$ will be marked. In the general case, however, this is no longer true, and additional care is required to ensure consistency throughout the algorithm.

We begin with an intuitive description of the Marking Algorithm,
followed by its formal pseudocode.

\paragraph{Intuitive Description.}
The Marking Algorithm proceeds by first calling the Structured-Sampling Algorithm and then proceeding to mark small events based upon the outcome.

As in the warm-up, we introduce a dummy ``parent edge''~$e_r$ so that the algorithm can process the tree~$T$ uniformly in a top-down manner. 
Initially, the algorithm inserts the pair $(e_r, \emptyset)$ into~$X$, where each element $(f, S)$ encodes an edge~$f$ together with its current small covering set $S \in \scov(f)$.  
At each iteration, it removes some $(f, S)$ from~$X$ and samples an \emph{event} $\scrF = (R, R_{\sm}, L_{\sm})$ from $\Ext((f,S), T_f)$ (Step~\ref{step:select_event}), using the fractional LP solution~$y$ as a guide.  
The sampling distribution in~\eqref{eq:choosing_correl_event_sm_high} is derived directly from the Extension Consistency constraint~\eqref{constr:extension_consistency}. Here:
\begin{itemize}
  \item $R$ is a sub-tree of edges correlated with~$f$; in particular, $R \setminus \{f\}$ contains only \emph{correlated} edges.
  \item $R_{\sm} \subseteq R$ consists of edges covered by at most~$\rho$ links (the ``small'' edges).
  \item $L_{\sm} \subseteq L$ is the set of links that integrally cover each edge in~$R_{\sm}$.
\end{itemize}

The distribution in~\eqref{eq:choosing_ext_correl_event_sm_high} analogously defines how the algorithm extends an event~$\scrF$ to a larger one~$\scrF_g$, corresponding to covering the edge $g$ in addition to those of~$T_f$.  These events are analgous to the events $\scrE$ and $\scrE_g$ chosen by the always-small algorithm.

Unlike the always-small algorithm, the marking algorithm handles the``bad'' cases—where an edge is covered by many links— using \emph{conditional resets} and using the large conditional LP value to find a  suitable partial up-link solution for these bad edges.  
In the case of a conditional reset, the Structured-Sampling Algorithm resamples using the distributions~$\calD_e$ guaranteed by~\Cref{lem:event_usage_bounds}.  
The process continues until~$X$ becomes empty. Once complete, the Marking Algorithm replaces certain links by their shadows in 
order to enforce the consistency constraints, and then marks the 
corresponding small events.

Overall, the Marking Algorithm can be interpreted as selecting integral configurations for \emph{stars} of~$T$ while maintaining local consistency across overlapping regions.  
To this end, it maintains for each edge~$e \in E$ a small covering set $F_e \subseteq L_e$, representing the subset of links that will integrally cover~$e$.

A crucial consistency requirement is that whenever two edges appear in a common marked event in Step~\ref{step:mark} (i.e. are both contained in a common small event), their local cover sets agree. Formally:
\begin{quote}
If a pair of edges~$e$ and~$e'$ are contained in a common event marked by the algorithm in Step~\ref{step:mark} (i.e. a small event), then they are \emph{consistent}, meaning that
\[
F_e \cap L_{e'} = F_{e'} \cap L_e.
\]
\end{quote}

This consistency is enforced when we replace sampled links by a smallest possible set of shadows in Step~\ref{step:shadow_replacement_step}.  
We will later show that the expected number of included shadows for each link $\ell \in L$ is approximately equal to $x(\ell)$.  
It is also not difficult to see that if at any point two such adjacent edges are assigned a common link (i.e., $\ell \in F_e \cap F_{e'}$), then the same shadow $\ell'$ of $\ell$ will be assigned to both $F_e$ and $F_{e'}$ after Steps \ref{step:shadow_replacement_step} and \ref{step:split_uncorrelated_links} of the Marking Algorithm.
Step~\ref{step:split_uncorrelated_links} is included solely to enable the 
definition of LP values for the large events introduced in 
\Cref{subsubsec:feas_struc_frac} to satisfy the desired properties of a Structured Fractional Solution. Consequently, the reader may safely 
ignore this step until reaching that section.

Next, we present the pseudocode description of the Structured-Sampling Algorithm.  
The reader is encouraged to consult \Cref{fig:structured_sampling_alg} for an illustration of a single iteration of the \texttt{while}-loop, highlighting the selection of~$H$ and the marking of events.

\paragraph{Structured-Sampling Algorithm Pseudocode}

\begin{enumerate}

    \item \textbf{Initialization:}  
    Initialize the tuple set \( X \gets \{(e_r, \emptyset)\} \). 

    \item \textbf{Compute Resampling Distributions:}  
    For each \( e \in E \), compute the resampling distribution 
    \( \mathcal{D}_e \) over subsets \( S \subseteq L_e \) of links that cover \( e \)
    satisfying \( 1 \le |S| \le \rho \) as defined in \Cref{lem:event_usage_bounds}. These distributions are defined with respect to the expected output of the algorithm 
    above \( e \), and are therefore well-defined despite the top-down construction.

    \item
    \textbf{Compute Up-Link Covering Distributions:}  
    For each event \( \scrF \), compute the distribution \( \mathcal{H}_{\scrF} \)
    over up-links covering \( R(\scrF) \setminus R_\sm(\scrF) \)
    not covered by \( L_\sm(\scrF) \), guaranteed by \Cref{lem:covering_huge_edges}.

    \item \textbf{Main Sampling Loop:}  
    \textbf{While} \( X \neq \emptyset \):
    \begin{enumerate}

        \item \textbf{Remove Current Pair:}  
        Remove an element \( (f, S) \) from \( X \). Let \( f = \{u,v\} \) 
        and assume w.l.o.g. that \( u \) is the parent of \( v \).
        \label{step:remove_from_X_sm_high}

        \item \textbf{Select an Event:}  
        Select an event \( \scrF = (R, R_\sm, L_\sm) \)
        from \( \text{Ext}((f, S), T_f) \) with probability
        \begin{equation} \label{eq:choosing_correl_event_sm_high}
            \frac{y(\scrF)}{y((f, S))},
        \end{equation}
        and sample \( H \) from \( \mathcal{H}_{\scrF} \).
        \label{step:select_event}

        \item \textbf{Update Coverage Records:}  
        For each \( e \in (R \setminus R_\sm) \cup \{f\} \), for which a link in $L_\sm$ covers $e$
        set \( F_e \gets L_\sm \cap L_e \), and for each remaining $e (R \setminus R_\sm) \cup \{f\}$ set \( F_e \gets H \cap L_e\).
        \label{step:modify_F_e}

        \item \textbf{Process Each Uncorrelated Child of \( f \):}  
        For each child \( g \) of \( f \) not in \( R \):
        \begin{enumerate}
            \item \textbf{Sample Extension Event:}  
            Sample an event \( \scrF_g = (R', R'_\sm, L'_\sm) \)
            from \( \text{Ext}(\scrF, T_f \cup \{g\}) \)
            with probability
            \begin{equation} \label{eq:choosing_ext_correl_event_sm_high}
                \frac{y(\scrF_g)}{y(\scrF)}\,,
            \end{equation}

            and \textbf{flag $g$.}
            \label{step:select_stars_plus_one_edge}

            \item \textbf{If \( g \in R_\sm' \) (edge covered by few links):}
                 Add \( (g, L_\sm' \cap L_{g}) \) to \( X \).
                \label{step:select_prob_trim_small}
            
            \item \textbf{If \( g \notin R_\sm' \) (edge covered by many links):}
             Perform a \textbf{conditional reset at \( g \)} by  
            sampling \( S_g \) from \( \mathcal{D}_g \), and
            adding \( (g, S_g) \) to \( X \).
            \label{step:reset_uncorrelated_children}
                 
             \label{step:uncorr_children_huge}
        \end{enumerate}
        \label{step:process_children_f}

        \item \textbf{Process Edges Covered by $\scrF$:}  
        For each \( g \in R \setminus \{f\} \):
        \begin{enumerate}
            \item \textbf{If \( g \in R_\sm \) (edge covered by few links)}:
            Add \( (g, L_\sm \cap L_{g}) \) to \( X \).
            \item \textbf{If \( g \not \in R_\sm \) (edge covered by many links)}:
            Perform a \textbf{conditional reset on uncovered child edges}:
            for each child edge \( h \) of \( g \) with \( h \notin R \),
            sample a set \( S_h \) from \( \mathcal{D}_{h} \), add \( (h, S_h) \) to \( X \).\label{step:cond_reset_lower_child}
        \end{enumerate}
        \label{step:huge_edges_in_tree}

    \end{enumerate}
\end{enumerate}

Next we discuss the Marking Algorithm. 
At a high level, one might hope to run the Structured-Sampling Algorithm
(as in Step~\ref{step:run_structured_sampling}) and then simply mark events
as in Step~\ref{step:mark}.  
However, this is not sufficient for two main reasons.

First, due to the possibility of conditional resets in Steps \ref{step:reset_uncorrelated_children} and \ref{step:cond_reset_lower_child} of the Structured-Sampling Algorithm, there may exists two distinct edges $e, e' \in E(e)$ contained in a marked small event and a link $\ell \in L_e \cap L_{e'}$
yet $\ell$ is assigned to $F_e$ but not to $F_{e'}$.
Such asymmetric assignments cannot be represented within the notation of
the Structured LP. 

Rather than complicating the Structured LP with additional bookkeeping to track
such asymmetries, we resolve this issue by introducing
Step~\ref{step:shadow_replacement_step}, which restores representability while
preserving the simplicity of the LP formulation.

If there are no large events in the Structured LP, the introduction of 
Step~\ref{step:shadow_replacement_step} alone would be sufficient, and 
Step~\ref{step:split_uncorrelated_links} would not be needed. However, in order to handle the setting in which there are large events we cannot simply rely on this adjustment. This brings us to our second problem. To understand this issue we analyze some important properties of the Structured-Sampling Algorithm.

Recall that for a vertex $v \in V$, we denote by $T_v$ the subtree rooted at $v$,
and by $T_v^+$ the tree with edge set $E(T_v)\cup\{e_v\}$.

Consider a link $\ell$ in the support of the Strong LP. For each edge $e$ in the tree the algorithm indicates via $F_e$ whether or not $\ell$ will be assigned to cover $e$ (i.e. $\ell \in F_e$ is $\ell$ is assigned to cover $e$ and $\ell \notin F_e$ otherwise). For $F=\{F_e\}_{e \in E}$ and each leading edge $e$ of $\ell$ covered by $\ell$ we let $Q_\ell(e, F)$ be the empty path if $\ell \notin F_e$, otherwise we let $Q_\ell(e, F)$ be the maximal path containing $e$ such that each edge $e' \in Q_\ell(e, F)$ satisfies $\ell \in F_{e'}$. If $F$ is understood we may simplify the notation to $Q_\ell(e)$.
For a leading edge $e$ of $\ell$ we let $P_\ell(e, F)$ be the portion of $Q_\ell(e, F)$ restricted to descendants of $e$ (including $e$ itself). Similarly, if $F$ is understood we may simplify the notation to $P_\ell(e)$

Suppose $\ell$ is a link in the support of the Strong LP with two leading edges
$e$ and $e'$.
If $\ell$ is correlated or is not flagged, then the Structured-Sampling Algorithm
defines both $P_\ell(e)$ and $P_\ell(e')$ based on a single sampled event—either
in Step~\ref{step:select_event} or in
Step~\ref{step:select_stars_plus_one_edge} (together with the recursive events
below).
In this case, Step~\ref{step:shadow_replacement_step} is sufficient to ensure
consistency between the two assignments.

In contrast, if $\ell$ is uncorrelated and flagged, then the events
$\scrF_e$ and $\scrF_{e'}$ are sampled independently in the respective executions
of Step~\ref{eq:choosing_ext_correl_event_sm_high} when $g=e$ and when $g=e'$.
As a result, the assignments $P_\ell(e)$ and $P_\ell(e')$ are no longer correlated.
To address this issue, we introduce
Step~\ref{step:split_uncorrelated_links}, which modifies
Step~\ref{step:shadow_replacement_step} to handle flagged uncorrelated links.

For each flagged link $\ell$ and each of its leading edges $e$,
Step~\ref{step:split_uncorrelated_links} examines, in distribution,
how the Structured-Sampling Algorithm assigns the edges in $P_\ell(e')$
during the iteration in which $g=e$ and $g'=e'$.
More precisely, this step reproduces the conditional distribution of these
assignments given the jointly sampled extension event.

In this case, whenever $\ell \in F_e$, instead of deterministically assigning
a shadow of $\ell$ as in Step~\ref{step:shadow_replacement_step}, the Marking
Algorithm selects the shadow $\ell'$ in
Step~\ref{step:shadow_assignment_flagged} according to the random choices made
in Steps~\ref{step:extension_two_edges} and
\ref{step:sample_desc_assignments}.
This modification ensures that the resulting shadow assignments are
distributionally consistent with a joint sampling of the two leading edges,
and we will see that it suffices to guarantee that the consistency constraints
also hold for large events.

Next we describe more formally how this is done.
Let $F = \{F_e\}_{e \in E}$ denote the random link-set assignment produced
by the Structured-Sampling Algorithm, and let $\mathcal{D}$ denote its distribution.
For any deterministic assignment $F' = \{F'_e\}_{e \in E}$,
\[
\Pr_{\mathcal{D}}[F = F']
\]
equals the probability that the Structured-Sampling Algorithm assigns
$F_e = F'_e$ simultaneously for all $e \in E$.
We will use $\Pr_{\mathcal{D}}[\,F = F' \mid \scrF\,]$ to denote probability
conditioned on the event that the Structured-Sampling Algorithm samples the
event $\scrF$.
When conditioning on the event that the Structured-Sampling Algorithm
samples $\scrF$, we write
\[
F \sim (\mathcal{D} \mid \scrF)
\]
to denote that $F$ is sampled from the conditional distribution of
$\mathcal{D}$ given $\scrF$.

\begin{definition}[Restriction (Strong LP)]
An event $\scrF \in \events$ is a \emph{restriction} of an event
$\scrF' \in \events$ if there exists a set
$Q \supseteq R(\scrF)$ such that $\scrF' \in \Ext(\scrF, Q)$.
Equivalently, $\scrF$ is obtained from $\scrF'$ by restricting
to the subtree induced by $R(\scrF)$.
\end{definition}

That is, $\scrF$ agrees with $\scrF'$ on all edges in $R(\scrF)$
and ignores all assignments outside this subtree.

For an event $\scrF$ and a subset $R' \subseteq R(\scrF)$,
we write $\Restr(\scrF, R')$ to denote the restriction of
$\scrF$ to the edge set $R'$.

Finally we are in a position to describe in more depth the intuition behind Steps \ref{step:extension_two_edges} and \ref{step:sample_desc_assignments} of the Marking Algorithm. Here we exploit the consistency constraints (\Cref{constr:extension_consistency}) of the Strong LP. Recall that the Strong LP has events for sub-trees with up to $\beta + 3$ leaves. Furthermore by \Cref{obs:T_f_constant_leaves}, the subtree $T_e$ has at most $\beta+1$ leaves for each $e \in E$. Consequently, $T_e$ can be safely extended by up to two additional edges
while remaining within the support of the Strong LP. Consider any two uncorrelated children $e', e''$ of some edge $e \in E$.
In Step~\ref{step:sample_desc_assignments} when $g=e', g'=e''$, the idea is to extend the event $\scrF_{e'}$ sampled in Step~\ref{step:select_stars_plus_one_edge} of the Structured-Sampling Algorithm when $g=e'$ to an event $\scrF_{e', e''}$ with edge set $T_e \cup \{e', e''\}$. In Step~\ref{step:sample_desc_assignments}, we then sample an outcome $F'$
of the Structured-Sampling Algorithm
conditioned on the event $\scrF_{e',e''}$, restricted to the subtree
$R(\scrF_{e''})$. 
Then for each uncorrelated link $\ell$ with leading edge $e$ and $e'$ we use the outcome $P_\ell(e'', F')$ to determine the precise shadow, if any, of $\ell$ used to cover $e'$.
The analogous procedure is performed when $(g,g')=(e'',e')$.
After completing Step \ref{step:split_uncorrelated_links} the algorithm is able to correctly mark the events in Step \ref{step:mark}. For each small event $\scrE$ we can then simply assign $y'(\scrE)$ to be the probability that $\scrE$ is marked in \ref{step:split_uncorrelated_links}. In order to define $y'(\scrE)$ for large events \Cref{subsubsec:feas_struc_frac} we use Step \ref{step:split_uncorrelated_links} to create a distribution over large events. The following is the pseudocode of the Marking Algorithm.  \Cref{fig:shadow_replacement} provides a visualization of the Marking Algorithm.

\paragraph{Marking Algorithm.}
\begin{enumerate}
    \item \textbf{Run Structured-Sampling Algorithm.}
    Run the Structured-Sampling Algorithm to obtain the sets $\{F_e\}_{e \in E}$
    and the set of flagged edges.  Flag each uncorrelated link $\ell$ if both of its leading edges are flagged.
    \label{step:run_structured_sampling}

    \item \textbf{Replace Links by Shadows (Non-flagged Links).}
    For each non-flagged link $\ell \in L$, define $\trim(\ell)$ as the smallest
    subset of $\shadows(\ell)$ satisfying:
    \begin{itemize}
        \item Each shadow covers only edges that use $\ell$.
        \item Each edge covered by $\ell$ is assigned to a unique shadow.
    \end{itemize}
    For each edge $e$ with $\ell \in F_e$, replace $\ell$ in $F_e$ by the unique
    shadow in $\trim(\ell)$ covering $e$.
    \label{step:shadow_replacement_step}

    \item \textbf{Split Flagged Links via Event Extensions.}
    For each flagged edge $g=\{v,w\}$ and each uncorrelated edge
    $g'=vw' \notin \{g,e_v\}$ incident to $v$, perform the following steps:
    \begin{enumerate}
        \item \textbf{Sample an Extension Event.}
        Sample an event
        \[
        \scrF_{g,g'} \in \Ext\!\bigl(\scrF_g,\; R(\scrF_g)\cup\{g'\}\bigr)
        \]
        with probability
        \begin{equation}\label{eq:choosing_ext_correl_event_pairs}
            \frac{y(\scrF_{g,g'})}{y(\scrF_g)} \, .
        \end{equation}
        \label{step:extension_two_edges}
        \item \textbf{Sample Descendant Assignments.}
        Let $\scrF_{g,g'}'$ be the restriction of $\scrF_{g,g'}$ to
        $T_f \cup \{g'\}$, and sample
        \[
            F'=\{F'_e\}_{e \in E}
        \]
        from the distribution $(\mathcal{D} \mid \scrF_{g,g'}')$.
        \label{step:sample_desc_assignments}

        \item \textbf{Shadow Assignment for Flagged Links.}
        For flagged link $\ell$ whose leading edges are $g$ and $gl$, let $\trim(\ell)$ be the smallest subset of
        $\shadows(\ell)$ satisfying:
        \begin{itemize}
            \item Each shadow $\ell' \in \trim(\ell)$ covers only descendant
            edges $e$ of $g$ (including $g$ itself) for which $\ell \in F_e$.
            \item Each descendant edge of $g$ covered by $\ell$ is assigned to
            a unique shadow in $\trim(\ell)$.
        \end{itemize}
        Equivalently, $\trim(\ell)$ partitions the set of edges currently using $\ell$
        into maximal connected segments, and assigns one shadow per segment.
        
        If $\ell \in F_g$, let $\ell'$ be the unique shadow of $\ell$ in $\trim(\ell)$
        that covers $g$.  Replace $\ell'$ by the unique shadow of $\ell$ whose path
        contains the union of the edges covered by $\ell'$ and the edges in $P_\ell(g')$.

        For every descendant edge $d$ of $g$, whenever $\ell \in F_d$, replace
        $\ell$ by the corresponding shadow in $\trim(\ell)$.
        \label{step:shadow_assignment_flagged}
    \end{enumerate}
    \label{step:split_uncorrelated_links}

    \item \textbf{Mark Events.}
    Mark all events corresponding to stars with root $e$ and at most one
    uncorrelated edge, namely
    \[
        (e, F_e), \quad
        (E^*(e), \bigcup_{e' \in E^*(e)} F_{e'}), \quad (E^*(e) \cup \{g\}, \bigcup_{e' \in E^*(e) \cup \{g\}} F_{e'})
    \]
    for each $g \in E(e) \setminus E^*(e)$.
    \label{step:mark}
\end{enumerate}

\begin{figure}[t]
    \centering
    \begin{minipage}{0.45\textwidth}
        \centering
        \begin{forest}
        for tree={
        circle, draw, minimum size=0.3cm, inner sep=0pt, outer sep=0pt,
        l sep=30pt, s sep=20pt,
        }
        [{}, name=a, fill=white
            [{}, name=r, fill=white, , edge label={node[midway,left]{$f$}}, edge={ultra thick, draw=black}
                 [{}, fill = white, name=u,  edge label={node[midway,left, yshift=-3pt]{$e_1$}}, edge={ draw=black}
                        [{}, fill=white, name=v, edge={ultra thick, draw=black}, edge label={node[midway,left, yshift=+2pt]{$e_4$}}]
                        [{}, fill=white, name=h, edge={draw=red}, edge label={node[midway, right]{$e_5$}}]
                        ]
                [{}, fill=white, name=k, edge label={node[midway,right, yshift=-3pt]{$e_2$}}, edge={ultra thick, draw=black}]
                [{}, fill=white , name=l, edge label={node[midway,right, yshift=-3pt]{$e_3$}}, edge={ultra thick, draw=red}]
                ]
        ]
        \draw[dashed, thick, bend right] (a) to node[midway, right] {$\ell_1$} (-3, -3);
        \draw[dashed, thick, bend right] (k) to node[midway, below] {$\ell_3$} (2, -4);
        \draw[dashed, thick, bend left=50] (v) to node[midway, left] {$\ell_2$} (r);
    \end{forest}
    \end{minipage}

    \caption{Illustration of one iteration of the while loop of the Structured-Sampling Algorithm. Throughout the reader should note that the other endpoint of $\ell_1$ is below $e_4$ in the tree and the other endpoint of $\ell_3$ is below $e_3$ in the tree.
    The pair $(f, \{\ell_1\})$ is first removed from $X$ in Step \ref{step:remove_from_X_sm_high}, after which the event $\scrF$, defined on the subtree $R(\scrF)$ consisting of the black edges $(f, e_1, e_2, e_4)$, is sampled in Step \ref{step:select_event}. 
    Here, $R_\sm(\scrF)$ corresponds to the thick edges $(f, e_2, e_4)$, and $L_\sm(\scrF)$ to the dashed links $(\ell_1, \ell_2, \ell_3)$. Since the only edge $e_1$ assigned by $\scrF$ to be covered by many links is covered by $L_\sm(\scrF)$, $H$ is empty.
    The edge $e_3 \notin E^*(f)$ is then considered as $g$ in Step~\ref{step:process_children_f}, where the event $\scrF_{e_3}$ is selected with $e_3 \in R_\sm'(\scrF_{e_3})$ and $L_\sm(\scrF_{e_3}) = L_\sm(\scrF)$.
    Next, $e_1$ is treated as $g$ in Step~\ref{step:huge_edges_in_tree}; since $e_5$ is a child of $e_1$ not contained in $R(\scrF)$, a conditional reset is performed in Step~\ref{step:cond_reset_lower_child}.
    Some examples of events that will subsequently marked (up to shadow replacement in Steps \ref{step:shadow_replacement_step} and \ref{step:split_uncorrelated_links} of the Marking Algorithm)
    in Step \ref{step:mark} include $\{(f, e_1, e_2\}, \{\ell_1, \ell_2, \ell_3\})$, $(\{(e_1\}, \{\ell_1, \ell_2\})$, and $(\{(e_1, e_4\}, \{\ell_1, \ell_2\})$.
     The pairs $(e_4, \{\ell_1, \ell_2\}), (e_5, S_5), (e_2, \{\ell_3\})$ and $(e_3, \ell_3)$ are added to $X$.}
\label{fig:structured_sampling_alg}
\end{figure}
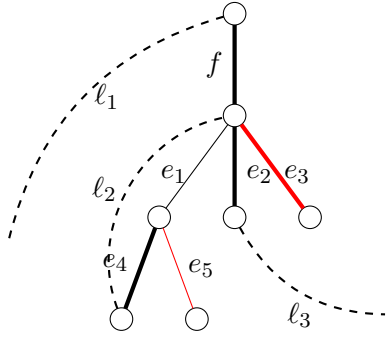

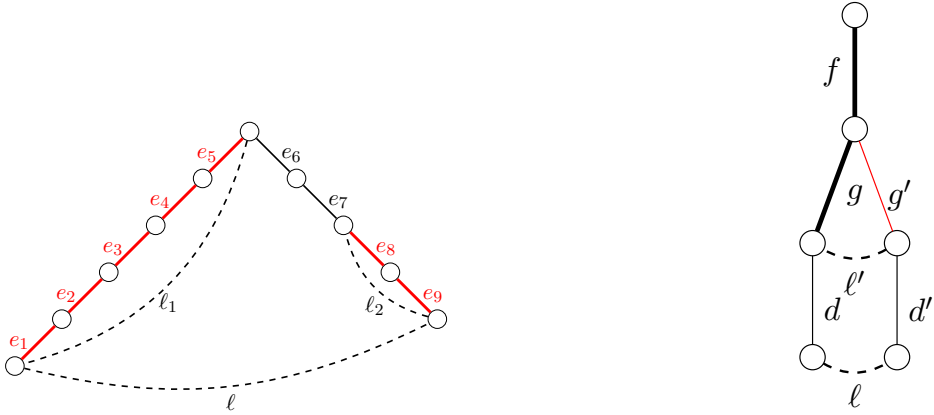
\begin{figure}[t]
\centering
\begin{minipage}[t]{0.48\linewidth}
\centering
\scalebox{0.78}{%
\begin{tikzpicture}[
    node/.style={circle, draw, fill=white, minimum size=9pt, inner sep=2pt},
    every label/.style={font=\small}
]
\def\dx{0.8}
\def\dy{0.8}
\node[node] (r) at (0,0) {};
\foreach \i in {1,...,5}{
    \pgfmathsetmacro{\x}{-\dx*\i}
    \pgfmathsetmacro{\y}{-\dy*\i}
    \node[node] (L\i) at (\x,\y) {};
}
\foreach \i in {1,...,4}{
    \pgfmathsetmacro{\x}{\dx*\i}
    \pgfmathsetmacro{\y}{-\dy*\i}
    \node[node] (R\i) at (\x,\y) {};
}
\draw[line width=1.4pt, red] (L5) -- node[midway, left] {$e_1$} (L4);
\draw[line width=1.4pt, red] (L4) -- node[midway, left] {$e_2$} (L3);
\draw[line width=1.4pt, red] (L3) -- node[midway, left] {$e_3$} (L2);
\draw[line width=1.4pt, red] (L2) -- node[midway, left] {$e_4$} (L1);
\draw[line width=1.4pt, red] (L1) -- node[midway, left] {$e_5$} (r);
\draw[thick] (r) -- node[midway, right] {$e_6$} (R1);
\draw[thick] (R1) -- node[midway, right] {$e_7$} (R2);
\draw[line width=1.4pt, red] (R2) -- node[midway, right] {$e_8$} (R3);
\draw[line width=1.4pt, red] (R3) -- node[midway, right] {$e_9$} (R4);
\draw[dashed, thick, bend right=20] (L5) to node[below] {$\ell$} (R4);
\draw[dashed, thick, bend left=30]  (r) to node[below] {$\ell_1$} (L5);
\draw[dashed, thick, bend left=30] (R4) to node[below] {$\ell_2$} (R2);
\end{tikzpicture}%
}
\end{minipage}
\hfill
\begin{minipage}[t]{0.48\linewidth}
\centering
\scalebox{1.12}{%
\begin{forest}
    for tree={
        circle, draw, minimum size=0.3cm, inner sep=0pt, outer sep=0pt,
        l sep=30pt, s sep=20pt
    }
    [{}, name=a
        [{}, name=r, edge label={node[midway,left]{$f$}}, edge={ultra thick, draw=black}
            [{}, name=j, edge label={node[midway,right, yshift=-3pt]{$g$}}, edge={ultra thick, draw=black}
                [{}, name=d, edge label={node[midway,right, yshift=-3pt]{$d$}}, edge={draw=black}]
            ]
            [{}, name=i, edge label={node[midway,right, yshift=-3pt]{$g'$}}, edge={draw=red}
            [{}, name=dprime, edge label={node[midway,right, yshift=-3pt]{$d'$}}, edge={draw=black}]
            ]
        ]
    ]
    \draw[dashed, thick, bend right] (d) to node[midway, below] {$\ell$} (dprime);
    \draw[dashed, thick, bend left]  (i) to node[midway, below] {$\ell'$} (j);
\end{forest}%
}
\end{minipage}

\caption{
Illustrations of the Marking Algorithm:\\
\textbf{Left: } (An Illustration of the Shadow Replacement Step (Step \ref{step:shadow_replacement_step}).): If $\ell$ is contained in $F_e$ for precisely $e \in \{e_1, e_2, e_3, e_4, e_5, e_8, e_9\}$, then we will replace $\ell$ by $\ell_1$ in $F_{e_i}$ for $i=1, 2, 3, 4, 5$ and by $\ell_2$ in $F_{e_8}$ and $F_{e_9}$.\\
\textbf{Right: } (An Illustration of Up-Link Splitting Step (Step \ref{step:split_uncorrelated_links}).):
Suppose that $f$ is the edge in the pair removed from $X$ during a given iteration of the while loop.
Suppose that $g$ and $g'$ are uncorrelated edges that are children of $f$ and hence are flagged in Step \ref{step:process_children_f}. The link $\ell$ depicts an uncorrelated link with leading edges $g$ and $g'$, $\ell$ is contained in $F_g$ and $F_{g'}'$ but not in $F_d$ or $ F_{d'}'$. In this case, the shadow $\ell'$ replaces $\ell$ in $F_g$.}
\label{fig:shadow_replacement}
\end{figure}

\paragraph{Defining our Structured Fractional Solution.}\label{subsubsec:construct_struct_frac_soln}

Recall that the Structured-Sampling Algorithm samples events from the support
of the Strong LP, while the Marking Algorithm replaces the links selected by
these events with shadows.
In the following, we will require a well-defined mapping from each shadow
link $\ell' \in L$ to the unique original link $\ell$ (from the support of $x$)
that was sampled in the corresponding event.

To make this mapping well defined, we duplicate links so that the shadow sets become pairwise disjoint: for every original link $\ell$, the set $\shadows(\ell)$ consists of all (and only) the shadow copies associated to $\ell$, and these sets satisfy
\[
\shadows(\ell) \cap \shadows(\ell') = \emptyset \qquad\text{for all } \ell \neq \ell'.
\]

Throughout the remainder of this section, when we refer to the set $L$,
we implicitly mean the set containing these duplicated (shadow) links.
We also remark that the same mapping applies to links sampled in $H$,
since they are likewise replaced by shadows associated with a unique
original link, as guaranteed by \Cref{lem:covering_huge_edges}.

All marked events contain
\emph{shadow links} rather than original links.  Because each shadow 
$\ell'$ belongs to exactly one set $\shadows(\ell)$, it has a unique 
associated original link~$\ell$.  Thus, whenever a shadow appears in the 
link set $L(\scrE)$ of a marked event, it is simply treated as a distinct 
link variable within the Structured LP, with the understanding that it 
originated from~$\ell$.

We now formally define our Structured Fractional Solution $(x', y', z')$  guaranteed by \Cref{lem:esp_alm_sm_struc_frac}. Recall that for $F \subseteq E$, the set $\scov(F)$ denotes the collection of link sets that are \emph{small} on each edge $e \in F$, that is, every edge $e \in F$ is covered by at most $\rho$ links.

As previously mentioned, for the Structured Fractional Solution, we increase the smallness bound so that each edge may be covered by up to
\[
   \rho' \;=\; \rho \cdot (\beta+1)
\]
links. Formally, for $F \subseteq E$, we let $\scov'(F)$ denote the collection of link sets $L \subseteq L(F)$ such that every edge $e \in F$ is covered by at most $\rho'$ links in $L$.

For each event $\scrE$ that is small enough to be marked by the 
Structured-Sampling Algorithm, we define $y'(\scrE)$ as the probability 
that the algorithm marks $\scrE$ in  
Step~\ref{step:mark}. 
Formally, we say that an event is \emph{small} if 
$\scrE \in \events'(E')$, where 
$E' \subseteq E^*(e) \cup \{g\}$ for some $g \in E(e)$ and $e \in E$. 
Small events are precisely those that can be marked by our algorithm.

The remaining \emph{large} events~$\scrE$, which satisfy 
$\scrE \in \events'(E')$ for 
$E' = E^*(e) \cup \{g, g'\}$ with 
$g, g' \in E(e) \setminus E^*(e)$ for some $e \in E$, 
will be defined in \Cref{subsubsec:feas_struc_frac} so that the 
extension consistency constraints 
(\Cref{eq:constr:small_prob_1:consistency}) are preserved.

We set $x'$ so that
the marginal-preserving constraints in \Cref{eq:constr:small_prob_1:marginals} are satisfied. 
In particular, for each link $\ell \in L$ we select any edge $e$ covered by $\ell$ and define 
\begin{equation}\label{eq:defining_x_prime}
    x'(\ell) = \sum_{\scrE \in \primeevents(e): \ell \in L(\scrE)} y'(\scrE).
\end{equation}
Although this choice may seem ambiguous due to the arbitrary choice of $e$, we remark that so long as the consistency constraints hold, this is sufficient. To see this, suppose that $e'$ is some other edge covered by $\ell$. For simplicity suppose that $e'$ is a child of $e$. Then if the consistency constraints hold the right hand side of \Cref{eq:defining_x_prime} equals $$\sum_{\scrE \in \primeevents(E^*(e) \cup \{e'\}): \ell \in L(\scrE)} y'(\scrE)$$ which again by consistency equals $$\sum_{\scrE \in \primeevents(e'): \ell \in L(\scrE)} y'(\scrE).$$ Using this argument we can walk along $P_\ell$ to show that every choice of $e$ leads to the same value for $x'(\ell)$.

When defining the solution $z'$, we
assign values only to those links that lie in the support of the original LP
solution~$x$.  Shadows contribute to $z'$ only indirectly, through their
associated original link.

To ensure that the odd cut constraints are satisfied by $z'$, we define $z'$
so that it dominates the original solution~$x$.
Intuitively, if the shadows of a link $\ell$ are selected frequently by the
algorithm, then the total mass assigned to these shadows in $x'$ should be
reflected in the value of $z'(\ell)$.

Accordingly, for each link $\ell$ in the support of~$x$, we define
\[
z'(\ell)
\;=\;
\max\!\bigl( x'(\shadows(\ell)),\; x(\ell) \bigr),
\]
where $x'(\shadows(\ell))$ denotes the total $x'$-mass assigned to all shadow
copies of~$\ell$.

This definition ensures that:
\begin{itemize}
\item the value assigned to each link $\ell$ under $z'$ is at least its
original value $x(\ell)$, implying that $z'$ is feasible for the Odd Cut LP; and
\item $c\!\left(z'(L)\right) \ge c\!\left(x'(L)\right)$.
\end{itemize}
Consequently, $z'$ satisfies the required properties of a feasible
fractional solution.

\subsubsection{Feasibility}\label{subsubsec:feas_struc_frac}
We will see in \Cref{lem:probability_including_edge} that, due to our choice of
the input solution $(x, y, \{y_\ell\}_{\ell \in L})$, it remains true that for
each vertex $v \in V$ and each edge $e \in E(v) \setminus E^*(v)$, the supports
$x'(L_{e_v})$ and $x'(L_e)$ are disjoint, and likewise
$z'(L_{e_v})$ and $z'(L_e)$ are disjoint.
This follows from the bound established in the lemma, which guarantees that any
original link absent from the support of~$x$ has all of its associated shadows absent from
the supports of both $x'$ and~$z'$.

Moreover,
\[
    c(z'(L)) \ge c(x'(L)).
\]

Since $z'$ dominates $x$ and $x$ is feasible for the odd cut LP, it follows that $z'$ is feasible for the odd cut LP.

In addition, we have constructed $x'$ to satisfy the marginal-preserving constraints stated in \Cref{eq:constr:small_prob_1:marginals} (provided that the consistency constraints hold).

Therefore, in order to prove the feasibility of the solution $(x', y', z')$, it suffices to verify the following conditions:

\begin{itemize}
    \item The \emph{coverage constraints} (\Cref{eq:constr:small_prob_1:coverage}) are satisfied.
    \item The \emph{consistency constraints} (\Cref{eq:constr:small_prob_1:consistency}) are satisfied.
    \item The \emph{uncorrelated edge bound} (\Cref{eq:uncorr_cost}) holds.
    \item The vector $y'$ is well-defined, meaning that for every event $\scrE$ to which we assign positive LP value, we have $\scrE \in \events'$.
\end{itemize}

We begin by proving the following:
\begin{lemma}\label{lem:stuc_frac_soln_well_defined_edge_events} For each $e \in E$, the Marking Algorithm always marks an event in $\events'(e)$.
\end{lemma}
\begin{proof}
    We claim that the statement already holds immediately after
    Step~\ref{step:run_structured_sampling}, i.e., after the Structured-Sampling Algorithm
    has produced the sets $\{F_e\}_{e\in E}$.  Since all subsequent modifications to the
    sets $F_e$ only replace links by their shadows, it suffices to show that after
    Step~\ref{step:run_structured_sampling} we have $F_e \in \scov'(e)$ for every $e\in E$.
    By the definition of $\events'(e)$, this implies that $(e,F_e) \in \events'(e)$ and hence
    an event in $\events'(e)$ is marked in Step~\ref{step:mark}.

    Consider an execution of the Structured-Sampling Algorithm  where the pair $(f, S)$ is removed from $X$ in Step \ref{step:remove_from_X_sm_high} and $F_e$ is fixed in Step \ref{step:modify_F_e}. Observe that this happens precisely one time per execution of the algorithm.
    
    Let $\scrF$ be the event sampled in Step~\ref{step:select_event}.
    We prove by induction that every pair $(f,S)$ that appears in $X$
    satisfies $S \in \scov(f)$ (i.e. $|S| \leq \rho$).

    Initially $X=\{(e_r,\emptyset)\}$.  (We may interpret $\emptyset \in \scov(e_r)$.) Assume $(f,S)\in X$ satisfies $S\in\scov(f)$ when it is removed.
    We show that every new pair $(f',S')$ added to $X$ during this iteration also satisfies
    $S'\in\scov(f')$.
    If a conditional reset occurs at $f'$, then $S'$ is sampled from $\mathcal{D}_{f'}$ whose
    support is contained in $\scov(f')$.
    Otherwise, $S'$ is obtained from an extension event (either $\scrF$ chosen in
    Step~\ref{step:select_event} or $\scrF_g$ chosen in Step~\ref{step:select_stars_plus_one_edge}),
    and by construction that event assigns $f'$ a small covering set, i.e.\ an element of $\scov(f')$.
    Thus every $(f',S')$ added to $X$ satisfies $S'\in\scov(f')$.

Now consider the assignments to $F_e$ in Step~\ref{step:modify_F_e}.  Let $\scrF$ be the
event selected in Step~\ref{step:select_event}.
We only assign $F_e$ for edges $e \in (R(\scrF)\setminus R_\sm(\scrF)) \cup \{f\}$.

If $e=f$ and $F_e$ equals $S$ which is in $\scov(e)$ and hence in $\scov'(e)$ by our inductive claim.

Otherwise, there are two cases to consider:
\begin{itemize}
    \item If no link in~$L_\sm(\scrF)$ covers~$e$, 
    $\ell$ is added to $F_e$ where ~$\ell$ is the unique link in~$H$ covering~$e$.
    \item If instead $L_\sm(\scrF)$ contains a link covering~$e$, 
    then $L_\sm(\scrF) \cap L_e$ is added to $F_e$. 
    Observe that $|L_\sm(\scrF) \cap L_e| \le \rho \cdot  (\beta + 1)$, 
    since there are at most $\beta + 1$ leaves in~$R(\scrF)$.
\end{itemize}
In either case $F_e \in \scov'(e)$.
This shows that after Step~\ref{step:run_structured_sampling}, we have
$F_e \in \scov'(e)$ for every $e \in E$, completing the proof.
\end{proof}

Since each edge is marked exactly once in the final step of the algorithm, as an immediate corollary to \Cref{lem:stuc_frac_soln_well_defined_edge_events}, we conclude that the coverage constraints are satisfied.
We now verify that the consistency constraints (\Cref{eq:constr:small_prob_1:consistency}) also hold. 

Consider any pair of stars $E_1, E_2 \in \stars'$ with $E_1 \subseteq E_2$,
where $E_2$ contains at most one uncorrelated edge (aside from the edge incident
to the root of $E_2$).
The corresponding consistency constraint is satisfied because whenever the
algorithm marks an event $\scrE_1 \in \events'(E_1)$, it also marks an event
$\scrE_2 \in \events'(E_2)\!\mid_{\scrE_1}$.

Indeed, let $e$ and $e'$ be two distinct edges in $E_2$ such that
$\ell \in F_e \cap F_{e'}$ prior to
Steps~\ref{step:shadow_replacement_step} and~\ref{step:split_uncorrelated_links}
of the Marking Algorithm.
We show that, after these steps, the shadow of $\ell$ in $F_e$
coincides with the shadow of $\ell$ in $F_{e'}$.
For convenience, we refer to this as the \emph{identical assignment property}
(of $\ell$).

There are two cases to consider.
First, suppose that $e$ and $e'$ are the two leading edges of $\ell$.
By our restriction on $E_2$, at least one of $e$ and $e'$ must be correlated,
and hence $\ell$ is a correlated link.
In this case, $\ell$ is never flagged, and
Step~\ref{step:split_uncorrelated_links} does not modify its assignments.
It therefore suffices to argue that the identical assignment property is
preserved by Step~\ref{step:shadow_replacement_step}.

Otherwise, $e$ and $e'$ have an ancestor--descendant relationship.
As in the previous case, it again suffices to verify that the identical
assignment property is preserved by
Step~\ref{step:shadow_replacement_step}, since
Step~\ref{step:split_uncorrelated_links} may modify assignments,
but cannot break the identical assignment property for edges that lie on the
same side of the apex of $\ell$.

We now argue that the identical assignment property is indeed preserved by
Step~\ref{step:shadow_replacement_step}.
By construction, $\trim(\ell)$ is chosen as the smallest subset of
$\shadows(\ell)$ satisfying the required coverage properties.
Since $e$ and $e'$ are adjacent edges, they are covered by the same connected
segment of $\ell$, and are therefore assigned the same shadow.

Conversely, if $\ell \in F_e \setminus F_{e'}$ prior to
Steps~\ref{step:shadow_replacement_step} and \ref{step:split_uncorrelated_links}, then the shadow of $\ell$ assigned to
$F_e$ after these steps does not cover $e'$. Hence no shadow of $\ell$ appears in
both $F_e$ and $F_{e'}$ unless $\ell \in F_e \cap F_{e'}$.

Consequently, when the algorithm marks the events
\[
\scrE_1 = \bigl(E_1,\ \bigcup_{e \in E_1} F_e\bigr)
\quad\text{and}\quad
\scrE_2 = \bigl(E_2,\ \bigcup_{e \in E_2} F_e\bigr),
\]
their associated link sets agree on $E_1$, and hence the extension consistency
constraint for $E_1 \subseteq E_2$ is satisfied.

It remains to define the variables~$y'$ for the remaining events so that the
consistency constraints continue to hold once we drop the restriction on~$E_2$.
In particular, we must address the case where~$E_2$ contains two uncorrelated
edges in addition to its root.  Recall that we refer to such stars as \emph{large stars},
and to all others as \emph{small stars}.

For small stars~$\scrE$, we already defined $y'(\scrE)$ as the probability that
$\scrE$ is marked by the Structured-Sampling Algorithm.  We now extend this
definition to large stars in a way that (i) preserves all consistency
constraints and (ii) allows us to interpret $y'(\scrE)$ as the probability
mass assigned to $\scrE$ by an explicit sampling rule.

To achieve this, for each edge~$e \in E$ and every pair of
uncorrelated children $e', e'' \in E(e) \setminus E^*(e)$,
we introduce a \emph{pairing distribution}
\[
  \phi : \events(E') \times \events(E'') \to \R_{\ge 0},
\]
where, for brevity, we write
$E' = E^*(e) \cup \{e'\}$ and $E'' = E^*(e) \cup \{e''\}$.
Intuitively, $\phi$ describes how events involving $E'$ and $E''$
are combined to form the joint events on
$E^*(e) \cup \{e', e''\}$,
while preserving both probability mass and consistency.

Formally, this pairing distribution must satisfy two key properties:

\begin{enumerate}

  \item \textbf{Marginal Consistency:}
  For every $\scrE' \in \events(E')$ and every $\scrE'' \in \events(E'')$,
  \[
    \sum_{\scrE'' \in \events(E'')} \phi(\scrE', \scrE'') = y'(\scrE'),
    \qquad
    \sum_{\scrE' \in \events(E')} \phi(\scrE', \scrE'') = y'(\scrE'').
  \]
  Hence, the marginals of $\phi$ match the previously defined values of~$y'$.

  \item \textbf{Structural Consistency:}
  Whenever $\phi(\scrE', \scrE'') > 0$, the paired events
  agree on their shared structure:
  \[
    L(\scrE') \setminus L(\scrE'') \text{ covers no edge in } E'',
    \quad
    L(\scrE'') \setminus L(\scrE') \text{ covers no edge in } E'.
  \]
  In other words, their link sets must be compatible; any link present in one set but not the other must not cover an edge in the other event's star. Consult \Cref{fig:structural_consistency} for an illustration of the structural consistency property.
\end{enumerate}

Observe that the Marginal Consistency Property implies that following:

\textbf{Pairwise Mass Preservation:}
  \[
    \sum_{\scrE' \in \events(E')} \sum_{\scrE'' \in \events(E'')}
    \phi(\scrE', \scrE'') = 1.
  \]
That is, $\phi$ redistributes total probability mass. This is because $y'(\scrE')$ is simply the probability that the algorithm marks $\scrE'$.

Given such a pairing~$\phi$, we define for every
$\scrE \in \events(E^*(e) \cup \{e', e''\})$:
\[
  y'(\scrE) =
  \sum_{\substack{\scrE' \in \events(E'),\, \scrE'' \in \events(E'') \\[1pt]
                  L(\scrE') \cup L(\scrE'') = L(\scrE)}}
  \phi(\scrE', \scrE'').
\]
Thus, $y'(\scrE)$ represents the total probability of obtaining~$\scrE$
by combining compatible smaller events on~$E'$ and~$E''$.

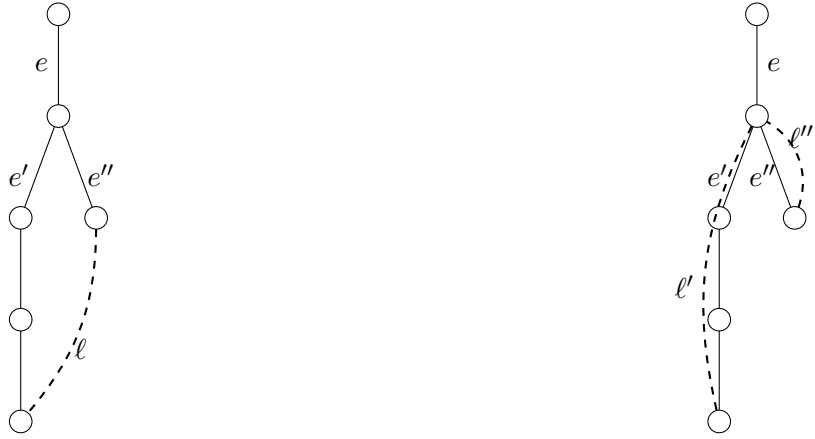
\begin{figure}[t]
\centering
\begin{minipage}{0.45\textwidth}
    \centering
    \vspace{4pt}
    \begin{forest}
        for tree={
            circle, draw, minimum size=0.3cm, inner sep=0pt, outer sep=0pt,
            l sep=30pt, s sep=20pt
        }
        [{}, name=a, fill=white
            [{}, name=r, fill=white, edge label={node[midway,left]{$e$}}, edge={draw=black}
                [{}, fill=white, name=u, edge label={node[midway,left, yshift=-3pt]{$e'$}}, edge={draw=black}
                    [{}, fill=white, name=v, edge={draw=black}, edge label={node[midway,left, yshift=+2pt]{}}
                        [{}, fill=white, name=b, edge={draw=black}]
                    ]
                ]
                [{}, fill=white, name=k, edge label={node[midway,right, yshift=-3pt]{$e''$}}, edge={draw=black}]
            ]
        ]
        \draw[dashed, thick, bend left=20] (k) to node[midway, below] {$\ell$} (b);
    \end{forest}
\end{minipage}%
\hfill
\begin{minipage}{0.45\textwidth}
    \centering
    \vspace{4pt}
      \begin{forest}
        for tree={
            circle, draw, minimum size=0.3cm, inner sep=0pt, outer sep=0pt,
            l sep=30pt, s sep=20pt
        }
        [{}, name=a, fill=white
            [{}, name=r, fill=white, edge label={node[midway,right]{$e$}}, edge={draw=black}
                [{}, fill=white, name=u, edge label={node[midway,left, yshift=-3pt]{$e'$}}, edge={draw=black}
                    [{}, fill=white, name=v, edge={draw=black}, edge label={node[midway,left]{}}
                        [{}, fill=white, name=b, edge={draw=black}]
                    ]
                ]
                [{}, fill=white, name=k, edge label={node[midway,left, xshift=4pt, yshift=-3pt]{$e''$}}, edge={draw=black}]
            ]
        ]
        Consistent non-crossing dashed links
        \draw[dashed, thick, bend right=40] (k) to node[midway, above, xshift=1pt] {$\ell''$} (r);
        \draw[dashed, thick, bend right=20] (r) to node[midway, below, xshift=-8] {$\ell'$} (b);
    \end{forest}
\end{minipage}
\caption{Illustration of structural consistency in the pairing distribution~$\phi$. 
Left: If $\phi(\scrE', \scrE'')>0$, then if $\ell$ is present in either event it must be present in both.\\
Right: If $\phi(\scrE', \scrE'')>0$, then $\scrE''$ is agnostic to $\ell'$'s presence in $\scrE'$. Likewise $\scrE'$ is agnostic to $\ell''$'s presence in $\scrE''$.}
\label{fig:structural_consistency}
\end{figure}

Before proceeding with out consistency proof we generalize the notion of restriction to the Structured LP.

\begin{definition}[Restriction (Structured LP)]
An event $\scrE \in \events'$ is said to be a \emph{restriction} of
an event $\scrE' \in \events'$ if
$\scrE' \in \Ext'(\scrE, Q)$ for some $Q \supseteq E(\scrE)$.
In this case, $\scrE$ is the restriction of~$\scrE'$ onto its edge set~$E(\scrE)$.
\end{definition}

\newcommand{\primeRestr}{\mathsf{Restr}'}
We will use the notation $\Restr'$ analogous to $\Restr$ for the Strong LP.

By the structural consistency property of~$\phi$, for $\scrE \in \events(E^*(e) \cup \{e', e''\})$ if $y'(\scrE) > 0$,
then it must equal $\phi(\scrE', \scrE'')$ for the unique
events~$\scrE'$ and~$\scrE''$ that are the restrictions of~$\scrE$
onto~$E'$ and~$E''$, respectively.

Since the consistency constraint already holds for all pairs of small events, it 
suffices to verify it in the remaining case where 
$E_2 = E^*(e) \cup \{e', e''\}$ and 
$E_1 \in \{E', E''\}$.
If $E_1$ is itself a subset of~$E^*(e)$, we leverage the fact that $E'$ and $E''$ are small stars. The consistency constraint between $(E_1, E')$ or $(E_1, E'')$ is already satisfied (as shown previously), which allows us to reduce this scenario to the desired primary case.

We now focus on the case $E_1 = E'$. An identical argument is made for $E_1=E''$.
Consider any event $\scrE_1 \in \events'(E_1)$.
By the \emph{marginal consistency} property of~$\phi$,
\[
    y'(\scrE_1)
    = \sum_{\scrE'' \in \events(E'')} \phi(\scrE_1, \scrE'').
\]

By the \emph{structural consistency} of~$\phi$, 
each event~$\scrE''$ with $\phi(\scrE_1, \scrE'') > 0$ satisfies that 
$L(\scrE'') \setminus L(\scrE_1)$ covers no edge of~$E_1$.  
Hence, every such pair $(\scrE_1, \scrE'')$ corresponds uniquely to an event 
$\scrE \in \events'(E^*(e) \cup \{e', e''\})$ whose restriction onto~$E_1$
is precisely~$\scrE_1$.

It follows, by the definition of $y'(\scrE)$ as the total probability of forming $\scrE$ from compatible pairs $(\scrE_1, \scrE'')$, that
\[
    y'(\scrE_1)
    = \sum_{\scrE \in \events'(E^*(e) \cup \{e', e''\})_{\mid \scrE_1}}
      y'(\scrE),
\]
where the subscript ${\mid \scrE_1}$ denotes events restricting to~$\scrE_1$ on~$E_1$.
Equivalently,
\[
    y'(\scrE_1)
    = \sum_{\scrE_2 \in \events(E_2)_{\mid \scrE_1}} y'(\scrE_2),
\]
which is exactly the requirement of 
\Cref{eq:constr:small_prob_1:consistency}.
Therefore, the consistency constraint holds for all~$E_1 \subseteq E_2$,
and the proof is complete.

\paragraph{Constructing $\phi$.}
Intuitively, we would like to define $\phi(\scrE', \scrE'')$ to be the probability
that both events $\scrE'$ and $\scrE''$ are marked in the following
\emph{modified} version of the Structured-Sampling Algorithm. First remove Step~\ref{step:split_uncorrelated_links} of the Marking Algorithm.
Instead of sampling $\scrF_{e'}$ and $\scrF_{e''}$ independently, we sample a
single joint event $\scrF_{e',e''}$ on the edge set
$T_e\cup\{e',e''\}$, and then define
\[
  \scrF_{e'} = \Restr(\scrF_{e',e''}, E')
  \qquad\text{and}\qquad
  \scrF_{e''} = \Restr(\scrF_{e',e''}, E'').
\]
This modified behavior is precisely what is mimicked by
Step~\ref{step:split_uncorrelated_links} of the Marking Algorithm.
The formal construction of $\phi$ is presented in
\Cref{appendix:construct_phi}.

\paragraph{Final steps.}
Having established that our definition of~$\phi$ guarantees the
consistency constraints (\Cref{eq:constr:small_prob_1:consistency}),
it remains to verify that~$y'$ is well defined and that the uncorrelated edge bound
(\Cref{eq:uncorr_cost}) holds.

The well-definedness of~$y'$ follows directly from
\Cref{lem:stuc_frac_soln_well_defined_edge_events} and the extension consistency constraints:
for every event~$\scrE$ to which we assign positive LP value, the restriction and extension relationships are consistent and the events assign a small subset to each edge, so $\scrE \in \events'$.

Finally, to verify \Cref{eq:uncorr_cost},
note that for every correlated link in the support of~$x'$,
the only way the algorithm can generate an uncorrelated shadow is by creating an up-link.
Therefore, every link $\ell \in L_U \setminus L_{UP}$ with $x'(\ell)>0$ 
maps back to a link in the support of~$x$ with the same leading edges.
Recalling that $z'(\ell) = \max(x'(\shadows(\ell)), x(\ell)) \ge x'(\shadows(\ell))$,
we obtain
\[
    c(x'(L_U \setminus L_{UP}))
    \leq \sum_{\ell \in L_U : x(\ell)>0} c(x'(\shadows(\ell)))
    \leq c(z'(L_U)).
\]
Combining the above, we conclude that $(x', y', z')$
forms a valid Structured Fractional Solution.

\subsubsection{Proving the Marginal Preserving Cost}

\label{sec:resampling_distributions}

In this section we argue that $x'$ approximately preserves the marginals of~$x$. 
Formally, we show that 
\[
  x'(\shadows(\ell)) \le (1 + o_\epsilon(1))\, x(\ell),
\]
which directly implies that the cost of $z'$ is close to that of $x$.
To prove this result, we must bound the probability that any given event is sampled by our algorithm.

During the iteration of the algorithm in which a pair $(f,S)$ is removed from
the queue, the sets $F_e$ are initialized for $e=f$ and for every edge
$e \in R \setminus R_{\sm}$.  Within the same iteration, however, new pairs
$(f',S')$ may be added to the queue, and thesis completely specifies the sets $F_{f'}$
we select in subsequent iterations. 

For the analysis in this section, it is convenient to adopt an alternative viewpoint.
We regard the set \(F_f\) as having been fixed in an earlier iteration, and view the
current iteration as determining the sets \(F_e\) for all edges \(e\) that either lie in
\(R \setminus \{f\}\), are children of an edge in \(R \setminus R_{\sm}\), or are children of \(f\).

This perspective clarifies which edges receive their covering sets as a direct
consequence of the event sampled in the current iteration.

With this viewpoint in place, we next enumerate all situations in which a given
small set \(S\) may become the covering set of an edge \(e\). These cases will form
the basis of our probability analysis.

Throughout this section, we restrict attention to events $\scrF$ in the support of the Strong LP (i.e., with $y(\scrF) > 0$), ensuring that all expressions are well-defined.

\paragraph{Algorithmic Event Definition.}
We focus on events $(e,S)$ such that $e\in E$ and $S\in\scov(e)$.
For such $e$ and~$S$, define $\mathsf{ALG}(e,S)$ to occur if any of the following 
happens during the algorithm:

\begin{enumerate}
    \item An event $\scrF=(R,R_\sm,L_\sm)$ is selected in 
    Step~\ref{step:select_event} such that 
    $e\in R_\sm\setminus\{f\}$ and $S=L_\sm\cap L_e$.
    \label{cond:alg_select_1}
    \item Edge $e$ is selected (as $g$) in 
    Step~\ref{step:select_stars_plus_one_edge},
    and the sampled event $\scrF_g=(R',R'_\sm,L'_\sm)$ satisfies 
    $e\in R'_\sm$ and $L'_\sm\cap L_e=S$.
    \label{cond:alg_select_2}
    \item Edge $e$ is conditionally reset in either 
    Step~\ref{step:reset_uncorrelated_children} or 
    Step~\ref{step:cond_reset_lower_child},
    and the sampled set of links equals~$S$.
    \label{cond:alg_select_3}
\end{enumerate}

If instead $S=\bot$, $\AlgEvent(e,\bot)$ is true whenever
an event $\scrF=(R,R_\sm,L_\sm)$ is selected in 
Step~\ref{step:select_event} with $e\in R\setminus R_\sm$.
Note that $\AlgEvent(e,\bot)$ never occurs when $e$ is a child of~$f$ 
not contained in~$R$.

Intuitively, the probability $\Pr[\AlgEvent(e,S)]$ equals the expected 
contribution of $S$ to $x'(L_e)$; thus, bounding it by $y(e,S)$
implies marginal preservation.

\paragraph{Supporting Lemmas.}
Our proof relies on the following two results, whose proofs are deferred to 
\Cref{subsubsec:resampling_dist,subsubsec:up_link_soln_huge_edges}.

\begin{restatable}[Resampling Distribution Properties]{lemma}{eventusagebounds}\label{lem:event_usage_bounds}
We can compute the resampling distributions $\{\calD_e\}_{e\in E}$ such that:
\begin{enumerate}
    \item For every $e\in T$ and $S\in\scov^+(e)$,
    \[
      \Pr[\AlgEvent(e,S)] \le (1 + 5\sqrt{\epsilon})\,y(e,S).
    \]
    \label{marginals_preserved}
    \item For every $e\in T$ and $S\in\scov(e)$, the probability of $S$ in $\calD_e$ 
    is at most $y(e,S)/\sqrt{\epsilon}$. 
    \label{marginals_resampling_not_too_big}
    \item For every edge $e \in T$ and $S \in \scov(e)$, if the probability of sampling $S$ from $\calD_e$ is nonzero, then no link $\ell \in S$ covers any ancestor of $e$ that is not pairwise correlated with~$e$ (with respect to $x$). \label{non_correlation_resampling}
\end{enumerate}
\end{restatable}

\begin{lemma}\label{lem:covering_huge_edges}
For each event $\scrF\in\events$, there exists a distribution 
$\calH_\scrF$ over up-link solutions covering all edges in 
$R(\scrF)\setminus R_\sm(\scrF)$ not covered by a link in $L_\sm(\scrF)$ such that: 
\begin{itemize}
    \item For each link $\ell\in L$, let $X_\scrF(\ell)$ denote the number of 
    shadows in $\shadows(\ell)$ included in $H\sim\calH_\scrF$. Then
    \[
      \E[X_\scrF(\ell)] \le \frac{2y_\ell(\scrF)}{\rho\,y(\scrF)},
    \]
    for all links relevant to $\scrF$ but not in $L_\sm(\scrF)$; otherwise 
    $X_\scrF(\ell)=0$ with probability~$1$.
    \item For every $H \sim \calH_\scrF$, each edge guaranteed to be covered is covered by exactly one link in $H$; otherwise, it is not covered at all.
\end{itemize}
\end{lemma}

\paragraph{Bounding Event Sampling.}
For any remaining event $\scrF\in\events$, let $\AlgEvent(\scrF)$ denote the 
event that $\scrF$ is selected in Step~\ref{step:select_event} or 
Step~\ref{step:select_stars_plus_one_edge}.  
The following corollary to \Cref{lem:event_usage_bounds} bounds this probability.

\begin{corollary}\label{cor:event_sampling_bound}
For each $e\in E$ and 
$\scrF\in\events(T_e)\cup\bigcup_{e'\in E(e)}\events(T_e\cup\{e'\})$,
\begin{equation}\label{eq:algevent_bound}
    \Pr[\AlgEvent(\scrF)] \le (1+5\sqrt{\epsilon})\,y(\scrF).
\end{equation}
\end{corollary}

For each $\scrF$ in \Cref{cor:event_sampling_bound}, the probability that 
$\scrF$ is selected by the algorithm is exactly
\[
  \Pr[\AlgEvent((f,S))]\cdot\frac{y(\scrF)}{y((f,S))},
\]
where $(f,S)$ is the restriction of $\scrF$ onto~$f$.
Applying Property~\ref{marginals_preserved} of 
\Cref{lem:event_usage_bounds} yields \Cref{cor:event_sampling_bound}.

\paragraph{Significant Events for Edges.}
The significant events for $(e,S)$ are those that ensure 
$\AlgEvent(e,S)$ occurs with probability one and include the following scenarios:
\begin{enumerate}
    \item An event $\scrF=(R,R_\sm,L_\sm)$ is selected in Step~\ref{step:select_event} 
    with $e\in R_\sm\setminus\{f\}$ and $S=L_\sm\cap L_e$.
    \label{cond:alg_select_1_again}
    \item Edge $e$ is selected (as $g$) in Step~\ref{step:select_stars_plus_one_edge} and the sampled event $\scrF' = (R', R'_\sm, L'_\sm)$ is such that $e\in R'_\sm$ and $L'_\sm \cap L_e = S$.
\end{enumerate}
We write $\AlgEvent'(e,S)$ for the event that a significant event occurs. Here we note that $\AlgEvent(e, S)$ occurs when $\AlgEvent'(e,S)$ occurs or when a conditional reset at $e$ occurs and selects $S$.
The following is an explicit characterization of the significant events for $(e, S)$.
\begin{definition}[Significant Events for an Edge]
Consider an event $(e, S)$, where $e \in E$ and $S \in \scov^+(e)$. We say that an event $\scrF \in \events$ is \emph{significant} for $(e, S)$ if:
\begin{itemize}
    \item If $S\in\scov(e)$, then $e\in R_\sm(\scrF)$ and 
          $L_\sm(\scrF)\cap L_e=S$; otherwise, $e\in R(\scrF)\setminus R_\sm(\scrF)$.
    \item Furthermore, either
          \begin{itemize}
              \item $\scrF\in\Ext((e',S'),T_{e'})$ 
                    for some ancestor $e'$ of~$e$ with $e\in R(\scrF)$, or
              \item $\scrF\in\Ext((e',S'),T_{e'}\cup\{e\})$, 
                    where $e$ is a child of $e'$ in $E(e')\setminus E^*(e')$,
          \end{itemize}
          and $S' \in \scov(e')$.
\end{itemize}
\end{definition}

Let $\signif(e,S)$ denote the set of events significant for $(e,S)$.
For example, in \Cref{fig:structured_sampling_alg}, 
the event~$\scrF$ is significant for $(e_4,\{\ell_1,\ell_2\})$, 
and $\scrF'$ is significant for $(e_3,\{\ell_3\})$.

\begin{observation}\label{obs:significant_events_relation_extension}
For every $e\in E$ and $S\in\scov^+(e)$,
\[
  y((e,S)) \ge \sum_{\scrF\in\signif(e,S)} y(\scrF).
\]
\end{observation}
\begin{proof}
If $e$ is correlated, let $e'$ be the highest ancestor of~$e$
pairwise correlated with it; otherwise, let $e'$ be its parent edge.
Let $P$ be the path in~$T$ from $e'$ to~$e$.
By the extension consistency property,
\begin{equation}\label{eq:extenstion_paths}
   y((e,S)) = \sum_{\scrF\in\Ext((e,S),P)} y(\scrF).
\end{equation}
For each $\scrF\in\Ext((e,S),P)$, let $e''$ be the edge incident 
to the root of $R(\scrF)$.
Then
\[
  y(\scrF) = \sum_{\scrF'\in\Ext(\scrF, T_{e''}\cup\{e\})} y(\scrF').
\]
Substituting this into~\eqref{eq:extenstion_paths} yields the desired result, 
since the sum now ranges over a superset of $\signif(e,S)$, 
excluding precisely those $\scrF'$ where $e''=e'$ is assigned to be 
covered by many links.
\end{proof}

\paragraph{Significant Events for Links.}
We are now ready to develop the remaining tools needed to bound the cost of our
Structured Fractional Solution~$(x', y', z')$ in comparison to the original
solution~$(x, y, \{y_{\ell}\}_{\ell \in L})$.
Our next goal is to analyze each link~$\ell \in L$ individually and bound the
expected number of distinct shadows of~$\ell$ that appear in marked events during
a single execution of the Marking Algorithm.

For links that are correlated, this quantity can be computed directly for a given
execution by analyzing the events in which a shadow of~$\ell$ is marked.
For uncorrelated links, the bound is instead obtained indirectly using the pairing
distribution~$\phi$.

Recalling that~$z'(\ell)$ is defined as the maximum of this expectation and
$x(\ell)$, this analysis will yield \Cref{lem:esp_alm_sm_struc_frac}.
The remainder of this section formalizes these arguments.

Next, we analyze the events in the Marking Algorithm that can cause shadows of a
link~$\ell$ to be included.

\paragraph{Case 1: All leading edges of $\ell$ are correlated.}

Fix a link~$\ell \in L$ and first consider the case in which both leading
edges of~$\ell$ are correlated.  In this situation, there are exactly three
ways in which~$\ell$ can be sampled from the support of the Strong LP and
added to some set~$F_e$, thereby causing a shadow of~$\ell$ to appear in a
marked event:

\begin{enumerate}
    \item \textbf{Event-Based Inclusion.}  
    An event~$\scrF$ is sampled in 
    Step~\ref{step:select_event} such that 
    $\ell \in L_\sm(\scrF) \setminus S$.

    \item \textbf{Auxiliary Inclusion.}  
    An event~$\scrF$ is sampled in Step~\ref{step:select_event}, 
    and a shadow of~$\ell$ appears in the set~$H$ sampled from~$\calH_\scrF$.

    \item \textbf{Reset Inclusion.}  
    An event~$\scrF$ is sampled in Step~\ref{step:select_event}, and a conditional reset occurs at an edge $e$ considered as $h$ in Step~\ref{step:cond_reset_lower_child} in which $\ell$ is contained in the sampled set.
\end{enumerate}

Observe that, in the case when $\ell$ has only correlated leading edges, it can never appear in the sets considered during 
a conditional reset in Step~\ref{step:reset_uncorrelated_children}. Indeed, every link covering an uncorrelated edge necessarily has at least one 
uncorrelated leading edge. 
This follows by our choice of $\{E^*(v)\}_{v \in V}$ and from the fact that our input solution satisfies for every $v\in V$ and each 
$e \in E(v)\setminus E^*(v)$,
\[
    x(L_{e_v}) \cap x(L_e) = \emptyset.
\]

We will argue later that there is exactly one iteration of the while loop 
in which any of the three situations above can occur for a given link~$\ell$.  
However, within that single iteration multiple inclusions of~$\ell$ may take place, 
and so we must carefully bound the expected number of inclusions of each type.  
This will be sufficient, since the number of shadows ~$\ell$ 
is bounded by the number of distinct inclusions arising from these three 
mechanisms.


\paragraph{Case 2: $\ell$ has at least one uncorrelated leading edge.} Next we consider the case when $\ell$ has at least one uncorrelated leading edge.
As before, there are exactly three situations in which a shadow of~$\ell$ can appear in a marked event. In this case we say that an \emph{event-based inclusion} occurs if an event~$\scrF_g$ is sampled in Step~\ref{step:select_stars_plus_one_edge} such that $\ell \in L_\sm(\scrF_g) \setminus S$ when~$g$ is an uncorrelated leading edge of~$\ell$.
An \emph{auxiliary} inclusion is defined in the same manner as the previous case. We note, however, that $\ell$ must be correlated for a secondary shadow to appear. Finally, a \emph{reset inclusion} occurs whenever $\ell$ is sampled during a conditional reset, either in Step~\ref{step:cond_reset_lower_child} or in Step~\ref{step:reset_uncorrelated_children} when~$g$ is an uncorrelated leading edge of~$\ell$.

\paragraph{Formal Characterization.}
Having identified the various types of inclusion events, 
we now give a precise description of the events~$\scrF$ which, if sampled in 
Step~\ref{step:select_stars_plus_one_edge}, may result in an inclusion of~$\ell$ 
during that same iteration of the while loop.  
We call this collection the \emph{significant events} for~$\ell$ and denote it 
by~$\signif(\ell)$.  
We now characterize this set explicitly.
As a reminder, we also use $\signif(e,S)$ for the significant events associated with the edge–set pair $(e,S)$; the notation is unified but the arguments specify which object is under consideration.

\begin{lemma}[Significant Events for a Link]
\label{defn:sig_events_link}
Consider a link $\ell \in L$ with $\apex(\ell)=u$.
Then each event $\scrF$ in $\signif(\ell)$ satisfies one of the following:
\begin{itemize}
    \item $\scrF \in \signif(e_u, \bot)$, or
    \item $\ell$ has only correlated leading edges and
    $\scrF \in \Ext((e_u, S), T_{e_u})$ for some
    $S \in \scov(e_u)$, or
    \item $\ell$ has an uncorrelated leading edge and
    $\scrF \in \Ext((e_u, S), T_{e_u} \cup \{e\})$ for some
    $S \in \scov(e_u)$ and $e$ is an uncorrelated leading edge of $\ell$.
\end{itemize}
\end{lemma}

\begin{proof}
Let $\calA$ denote the subset of events described in the lemma statement.
We show that $\signif(\ell) \subseteq \calA$.

Suppose that no event in $\calA$ has yet been sampled by the algorithm.
Since the algorithm proceeds in a top--down manner,~$\ell$ cannot be included in any set~$F_e$ for $e \in E$, because $\ell$ is never relevant for the events sampled in Steps~\ref{step:select_event} or~\ref{step:select_stars_plus_one_edge}, 
nor for the edges where conditional resets are performed.
Observe that there is exactly one iteration of the while loop in which an
event from~$\mathcal{A}$ is sampled.  
Note that if $\ell$ has two uncorrelated leading edges, then two events from~$\mathcal{A}$
may be sampled in that same iteration.)
Let $X'$ denote the set of pairs added to $X$ during this iteration.

We remark that any other significant event for $\ell$ can only be sampled in a later iteration in
which the algorithm processes a pair $(f,S)$ whose edge $f$ is a
\emph{descendant} of one of the edges appearing in a pair in $X'$ and $\ell \in L_f$.  
This is because if $f$ were not a descendant of these edges, then $\ell$
would not be relevant for the sampled events. Furthermore, if $\ell$ does not cover $f$ but covers an ancestor of $f$, then it covers no descendant of $f$.

Now take $(f, S)\in X'$ and suppose $\ell \in L_f$. 
Then, in the iteration of the while loop corresponding to $(f, S)$, 
$\ell$ can never be sampled in~$H$.
Indeed, by \Cref{lem:covering_huge_edges}, $X_\scrF(\ell)=0$ 
for all such events $\scrF$, since either $\ell \in L_\sm(\scrF)$ 
or $y_\ell(\scrF)=0$. Hence no auxiliary inclusion occurs.

By the same reasoning, $\ell$ cannot have a event-based inclusion
at this iteration, since by definition $\ell$ cannot be in~$S$ for this to occur.

Moreover, for any descendant edge~$e$ of~$f$ but not contained in~$T_f$, 
$e$ is pairwise uncorrelated with~$f$, and thus by Property~\ref{non_correlation_resampling}
of \Cref{lem:event_usage_bounds}, $\ell$ can never be sampled in a 
conditional reset in any later iteration of the while loop corresponding 
to a proper descendant of~$f$.

The same reasoning applies to all pairs $(f',S')$ added to $X$ whose edge
$f'$ is a descendant of~$f$: for none of them can an inclusion event involving
$\ell$ occur.

As a result, all inclusion events must occur when an element of $\calA$ is sampled, concluding the proof.
\end{proof}

For the remainder of the analysis we will, for simplicity, assume that 
$\signif(\ell) = \calA$.  
Strictly speaking, this need not hold: it may happen that 
$y_\ell(\scrF)=0$ while $R(\scrF)$ contains all edges covered by~$\ell$, 
in which case no inclusion event for~$\ell$ can occur.  
This technical exception, however, has no effect on the arguments that follow, 
and we therefore suppress it for clarity of presentation.

The following is the main result we use to prove \Cref{{lem:esp_alm_sm_struc_frac}}:
\begin{lemma}\label{lem:probability_including_edge} For each $\ell \in L$ we have,
\begin{equation} \label{eq:prob_link__sm_high_corr_link}
x'(\shadows(\ell)) \leq (1+9\sqrt{\epsilon})x(\ell)
\end{equation}
\end{lemma}

\begin{proof}

Let $\ell \in L$ with $\apex(\ell) = u$. 
Each significant event~$\scrF$ for~$\ell$ determines a set of 
edges on the path~$P_\ell$ where a conditional reset is performed; 
we denote this set by~$\restart(\scrF)$, which
can be constructed as follows.

If $\scrF \in \signif(e_u, \bot)$, then
for every leading edge~$e$ of~$\ell$ 
we add to~$\restart(\scrF)$ at most one edge on the path~$P$ 
between the apex of~$\ell$ and the endpoint of~$\ell$ lying below~$e$ 
in the tree, according to the following rule:
\begin{itemize}
    \item If $P$ contains an element of~$R_\sm(\scrF)$, add no edge.
    \item Otherwise, add the first edge of~$P$ that is not contained in~$R(\scrF)$, 
    if such an edge exists.
\end{itemize}

Otherwise we do the above procedure for each leading edge $e$ of $\ell$ contained in $R(\scrF)$. Observe that if $\ell$ is correlated, $R(\scrF)$ contains both leading edges, otherwise it contains exactly one. This modification ensures that we do not double count the cost of performing conditional resets in our charging schemes. The readers is encouraged to consult \Cref{fig:restart} for an illustration. 
 
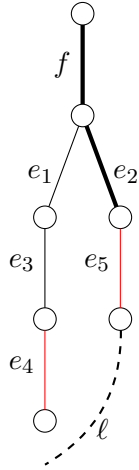
\begin{figure}[t]
    \centering
    \begin{minipage}{0.45\textwidth}
        \centering
        \begin{forest}
for tree={
circle, draw, minimum size=0.3cm, inner sep=0pt, outer sep=0pt,
l sep=30pt, s sep=20pt,
}
[{}, name=a, fill=white
    [ {}, name=r, fill=white, edge label={node[midway,left]{$f$}}, edge={ultra thick, draw=black}
        [ {}, fill=white, name=u, edge label={node[midway,left, yshift=-3pt]{$e_1$}}, edge={draw=black}
            [ {}, fill=white, name=v, edge={draw=black}, edge label={node[midway,left, yshift=+2pt]{$e_3$}}
                [ {}, fill=white, name=b, edge={draw=red}, edge label={node[midway,left, yshift=+2pt]{$e_4$}} ]
            ]
        ]
        [ {}, fill=white, name=k, edge label={node[midway,right, yshift=-3pt]{$e_2$}}, edge={ultra thick, draw=black} 
        [ {}, fill=white, name=m, edge={draw=red}, edge label={node[midway,left, yshift=+2pt]{$e_5$}} ]
        ]
    ]
]
\draw[dashed, thick, bend left] (m) to node[midway, below] {$\ell$} (-0.5, -6);
\end{forest}

    \end{minipage}
    \caption{Illustration of $\restart(\scrF)$, for a correlated link $\ell$: Note that the second endpoint of $\ell$ is below $e_4$ in the tree. We are considering an event $\scrF \in \signif(\ell)$, where the edges in $P_\ell$ contained in $R(\scrF)$ are $e_1, e_2$ and $e_3$. Among them only $e_2$ is assigned to be small. Therefore $\restart(\scrF)$ contains no edge below $e_2$. However $\restart(\scrF)$ contains $e_4$ since neither $e_1$ nor $e_2$ are in $R(\scrF)$.}
    \label{fig:restart}
    \end{figure}

In the following, we use $\Pr_{\calD_e}[S]$ to denote the probability of 
sampling the link set~$S$ from the distribution~$\calD_e$.

For simplicity, we define $y_\ell(\scrF) = 0$ whenever the link~$\ell$ 
is not relevant to the event~$\scrF$.
\newcommand{\corr}{\mathsf{corr}}
\newcommand{\uncorr}{\mathsf{uncorr}}
Let $\Lead_{\uncorr}(\ell)$ denote the subset of $\Lead(\ell)$ consisting of 
uncorrelated edges.

\paragraph{Correlated Links: }
Suppose that $\ell$ is a correlated link. 
We will establish the following bound on $x'(\shadows(\ell))$, 
where for each $\scrF \in \signif(\ell)$ we write $\scrF'$ to denote 
the restriction of~$\scrF$ to the subtree 
$R(\scrF) \setminus \Lead_{\uncorr}(\ell)$.

\begin{equation}\label{eq:correlated_link_bound}
\begin{aligned}
    x'(\shadows(\ell)) &\leq
    \sum_{\scrF \in \signif(\ell)} \Pr[\AlgEvent(\scrF)] \Big( \mathbbm{1}[\ell \in L_\sm(\scrF)] + \mathbbm{1}[\ell \notin L_\sm(\scrF'), \ell \in L_C]\frac{2 y_\ell(\scrF')}{\rho y(\scrF')}\\
    &+ \sum_{e \in \restart(\scrF)} \sum_{S \in \scov(e): \ell \in S} \Pr_{\calD_{e}}[S] 
    \Big)\\
\end{aligned}
\end{equation}

The first term on the right-hand side of \Cref{eq:correlated_link_bound} bounds the probability of an event-based inclusion.

The second term on the right-hand side of \Cref{eq:correlated_link_bound} bounds the probability of an auxiliary inclusion. To see this observe that $\ell$ can only be sampled in~$H$ if 
$\ell \notin L_\sm(\scrF')$.  
Moreover, by \Cref{lem:covering_huge_edges}, 
for any given~$\scrF'$, this probability is bounded by
\[
    \frac{2\, y_\ell(\scrF')}{\rho\, y(\scrF')}.
\]
The third term on the right-hand side accounts for the cost of reset inclusions.

To bound the first term on the right hand side of the inequality we apply \Cref{cor:event_sampling_bound} to obtain that $$\sum_{\scrF \in \signif(\ell)} \Pr[\AlgEvent(\scrF)] \mathbbm{1}[\ell \in L_\sm(\scrF)] \leq \sum_{\scrF \in \signif(\ell)} (1+5\sqrt{\epsilon}) y(\scrF).$$
By \Cref{obs:significant_events_relation_extension} and the marginal preserving constraints (\Cref{constr:marginal_pres_const_child}) this is at most $(1+5\sqrt{\epsilon})x(\ell)$. 

To bound the second term observe that in the worst case whenever $\ell \in L_\sm(\scrF)$ we have $\ell \notin L_\sm(\scrF')$. Using the same logic in our bound of the first term, extension consistency
(\Cref{constr:extension_consistency}) and our choice of $\rho$ we obtain that

 $$\sum_{\scrF \in \signif(\ell)} \Pr[\AlgEvent(\scrF)] \mathbbm{1}[\ell \notin L(\scrF'), \ell \in L_C]\frac{2 y_\ell(\scrF')}{\rho y(\scrF')} \leq (1+5 \sqrt{\epsilon}) \epsilon^2 x(\ell).$$

Since 
\( (x, y, \{y_\ell\}_{\ell \in L}) \) is an 
\( \epsilon \)-almost-always-small solution, by applying \Cref{obs:significant_events_relation_extension} to $(e, \bot)$ for any $e \in \Lead(\ell) \cup \{e_u\}$ we have:
\[
    \sum_{\scrF \in \signif(\ell) : e \notin R_\sm(\scrF)} y(\scrF) \leq \epsilon.
\]    

Applying this fact and Property \ref{marginals_resampling_not_too_big} of \Cref{lem:event_usage_bounds}, we conclude that

$$ \sum_{\scrF \in \signif(\ell)} \Pr[\AlgEvent(\scrF)] \sum_{e \in \restart(\scrF)} \sum_{S \in \scov(e): \ell \in S} \Pr_{\calD_{e}}[S] \leq 3 \sqrt{\epsilon} x(\ell)
    $$

Putting this all together we obtain: 
\begin{equation*}
    x'(\shadows(\ell)) \leq (1+5 \sqrt{\epsilon}) ((1+\epsilon^2)x(\ell) + 3 \sqrt{\epsilon} x(\ell) \leq (1+5 \sqrt{\epsilon} + 4\sqrt{\epsilon}) x(\ell) \leq (1+9 \sqrt{\epsilon}) x(\ell).
\end{equation*}

\paragraph{Uncorrelated Links: }
Now consider the case when $\ell$ is an uncorrelated link. Fix an edge $e^* \in \Lead(\ell)$ and let $\signif^*(\ell)$ be the subset of $\signif(\ell)$ such that $\scrF \in \signif(e_u, \bot)$ or $\scrF \in \Ext((e_u, S), T_{e_u} \cup \{e^*\})$ for some $S \in \scov(e_u)$.

Then we obtain the the following bound on $x'(\shadows(\ell))$, 

\begin{equation}\label{eq:uncorrelated_link_bound}
\begin{aligned}
    x'(\shadows(\ell)) &\leq
    \sum_{\scrF' \in \signif^*(\ell)} \sum_{\scrF'' \in \Ext(\scrF', R(\scrF')\cup \Lead(\ell))}\Pr[\AlgEvent(\scrF')] \cdot \frac{y(\scrF'')}{y(\scrF')} \mathbbm{1}[\ell \in L_\sm(\scrF'')] \\
    &+ \sum_{\scrF \in \signif(\ell)} \sum_{e \in \restart(\scrF)} \sum_{S \in \scov(e): \ell \in S}
    \Pr[\AlgEvent(\scrF)]
    \Pr_{\calD_{e}}[S] 
\end{aligned}
\end{equation}

We interpret the first term on the right-hand side of
\Cref{eq:uncorrelated_link_bound} as the total probability of an
event-based inclusion.

To bound this term, we invoke the pairing $\phi$ from
\Cref{subsubsec:feas_struc_frac}, which provides a consistent way to combine
the two one-extra-edge stars $E^*(e_u)\cup\{e'\}$ and $E^*(e_u)\cup\{e''\}$
into a joint event on $E^*(e_u)\cup\{e',e''\}$, where $\{e', e''\}=\Lead(\ell)$.
In particular, the definition of $y'$ via $\phi$ implies that, for any fixed
choice of $e^*\in \Lead(\ell)$, summing over extensions
$\scrF'' \in \Ext(\scrF', R(\scrF')\cup \Lead(\ell))$ yields the same total
mass assigned to events that include $\ell$.
Consequently, the value of the first term in \Cref{eq:uncorrelated_link_bound}
is independent of the choice of $e^*$.

Using this coupling view together with extension consistency
(\Cref{constr:extension_consistency}) and the event-sampling bound
\Cref{cor:event_sampling_bound}, we obtain that the first term is at most
$(1+5\sqrt{\epsilon})\,x(\ell)$.
The second term (reset inclusions) is bounded exactly as in the correlated-link
case, yielding the claimed bound.
Finally, we use our earlier remark that an uncorrelated link is never sampled
in $H$.

\end{proof}

\subsubsection{Proof of Resampling Distributions (\Cref{lem:event_usage_bounds})}\label{subsubsec:resampling_dist}
This section is dedicated to proving:
\eventusagebounds*

The procedure for constructing the resampling distributions are based on the observation that, other than the resampling events, $\Pr[\AlgEvent(e,S)]$ only depends on the ancestors of $e$ (and the LP solution).
We therefore construct the resampling distributions in increasing order of distances to the root.

Consider now the task of computing $\calD_e$ for a tree-edge $e = \{u,v\} \in T$ where $u$ is the parent of $v$. By the order in which we define the distributions, we have defined $\calD_f$ for all edges $f$ on the path $P$ from $u$ to the root ending in the dummy edge $e_r$. 
We now proceed to define the distribution $\calD_e$ and prove that it satisfies the properties of the lemma.

 We will first bound the probability of the event, and then define our distribution $\calD_e$ by adjusting the probabilities based upon how close they are to this bound.

\begin{claim} Suppose that $e=uv \in T$. If \Cref{lem:event_usage_bounds} holds for every edge $f$ on the path from $u$ to the root ending in $e_r$, then for each $S$ that either equals $\bot$ or is a small subset of links $S\subseteq L_e$ we have,   
    
    \[\Pr[\AlgEvent'(e, S)] \leq (1 + 5 \sqrt{\epsilon})\cdot  y(e,S).\]
        
\end{claim}
\begin{proof}
By our previous discussion:
\begin{equation}\label{eq:prob_event_occurs_significant}
    \Pr[\AlgEvent'(e,S)] = \sum_{\scrF \in \signif(e, S)} \Pr[\AlgEvent(\scrF)]
\end{equation}

Consider each event $\scrF \in \signif(e, S)$ and suppose that $e'$ is the edge incident to the root of $R(\scrF)$ and that $S' = L_\sm(\scrF) \cap L_{e'}$. Then we know that 

\begin{equation}
   \Pr[\AlgEvent(\scrF)] = \Pr[\AlgEvent(e', S')] \frac{y(\scrF)}{y(e',S')}. 
\end{equation}

The right-hand side is at most $(1+ 5 \sqrt{\epsilon}) y(\scrF)$ by assumption.

Hence we obtain that
\begin{equation*}
  \Pr[\AlgEvent'(e,S)] \leq (1+ 5 \sqrt{\epsilon}) \sum_{\scrF \in \signif(e, S)} y(\scrF).  
\end{equation*}

By \Cref{obs:significant_events_relation_extension} this is at most $\leq (1+ 5 \sqrt{\epsilon}) y(e, S)$ as required.

\end{proof}

The probability $\AlgEvent(e,S)$ is the same as that of $\AlgEvent'(e,S)$ except that we have not counted the resampling that can occur in  Steps~\ref{step:reset_uncorrelated_children} and ~\ref{step:cond_reset_lower_child}.  We now  define $\calD_e$ so as to ensure the properties of the lemma: 
\begin{quote}
    Let $f$ be the first edge encountered by walking from $e$ towards the root (including $e_r$) such that $x(L_f \cap L_e) < \delta$, i.e., $e$ is not correlated with $f$. Further,  let $\calL(e)$ be the family containing the small link subsets $S\in \scov(e)$ such that  $L_f \cap S = \emptyset$ and $\Pr[\AlgEvent(e,S)] \leq (1+2\sqrt{\epsilon}) y(e,S)$.  The distribution $\calD_e$ is defined by sampling a subset $S \in \calL(e)$ proportional to $y(e, S)$.
\end{quote}
Note that this choice of $\calD_e$ ensures the third property of the lemma holds. We proceed to verify the remaining properties of the lemma.
\begin{claim}
    We can calculate $\calD_e$. Furthermore, $\Pr[\AlgEvent(e,S)]\leq (1+ 5\sqrt{\epsilon})\cdot y(e,S)$.
\end{claim} 
\begin{proof}[Proof of claim]
That we can calculate $\calD_e$ follows from that we can calculate $\Pr[\AlgEvent'(e,S)]$.  This is true because to calculate $\Pr[\AlgEvent'(e,S)]$ we simply need to evaluate the right-hand-side of~\eqref{eq:prob_event_occurs_significant}, using that the probabilities of $\AlgEvent(e',S')$ are already defined for each ancestor edge $e'$ of $e$ and each $S' \in \scov^+(e)$. 


To bound $\Pr[\AlgEvent(e,S)]$ it suffices to consider only $S \in \calL(e)$. We have 
$\Pr[\AlgEvent(e,S)] = \Pr[\AlgEvent'(e,S)]  + p \cdot \Pr_{\calD_e}[S]$, where we recall that $\Pr_{\calD_e}[S]$ denotes the probability of $S$ in $\calD_e$ and  $p$ is the probability that $e$ is resampled in Step~\ref{step:reset_uncorrelated_children} or Step~\ref{step:cond_reset_lower_child}.
To prove the upper bound on $\Pr[\AlgEvent(e,S)]$ in this case we have $S \in \scov(e)$, note that $$p \leq \Pr[\AlgEvent(e, \bot)] + \Pr[\AlgEvent(e_u, \bot)] \leq (1+5\sqrt{\epsilon})\cdot 2 \epsilon,$$
where the last inequality comes from applying the fact that $(x, y, \{y_\ell\}_{\ell \in L})$ is $\epsilon$-almost-always-small, and our inductive hypothesis on events higher in the tree.
Thus,
    \begin{align*}
        \Pr[\AlgEvent(e,S)] = \Pr[\AlgEvent'(e,S)] + p \Pr_{\calD_e}[S] &\leq (1+2\sqrt{\epsilon}) y(e,S) + (1+5\sqrt{\epsilon}) 2 \epsilon y(e,S)/\sqrt{\epsilon} \\
        & = (1+4\sqrt{\epsilon} + 10 \epsilon) y(e,S) \\
        & \leq (1+5 \sqrt{\epsilon}) y(e,S)
    \end{align*}
    as required.
\end{proof}
Where the first inequality comes from the assumption that $S \in \calL(e)$, the above bound on $p$ and the subsequent claim( \Cref{claim:not_too_big}).

We now proceed to upper bound the probability of any small subset of links $S \subseteq L_e$ in $\calD_e$ (i.e., the second property of the lemma).

\begin{claim}
    For any small subset of links $S \subseteq L_e$, the probability of $S$ in $\calD_e$ is at most $y(e,S)/\sqrt{\epsilon}$.
    \label{claim:not_too_big}
\end{claim}
\begin{proof}[Proof of Claim]
This follows by proving that $\sum_{S\in \calL(e)}y(e,S) \geq \sqrt{\epsilon}$. Indeed, suppose otherwise toward contradiction. Then  (as before the sums are over small subsets of links $S$)
    \begin{align*}
        \sum_{S \in \scov(e) : S\cap L_f = \emptyset,  S \not \in \calL(e)}& \Pr[\AlgEvent'(e,S) ] 
        \\
        & > (1+2\sqrt{\epsilon})\sum_{S \in \scov(e): S \cap L_f = \emptyset, S \not\in \calL(e)} y(e,S) \\
        & \geq (1+2\sqrt{\epsilon}) (1- \sqrt{\epsilon} - \epsilon - \delta) > 1\,,
    \end{align*}
    The inequalities use  the following facts:
    \begin{itemize}
        \item  The instance is $\epsilon$-almost-always-small so $y(e, \bot) \leq \epsilon$.
        \item By selection of $f$, we have $x(L_e \cap L_f) < \delta$ which implies that $ \sum_{S \subseteq L_e \setminus L_f} y(e, S) > 1- \delta$.
        \item The selection of the parameters implying that (for sufficiently small $\epsilon$) $\epsilon + \delta \leq \sqrt{\epsilon}/2$ and so  $(1+2 \sqrt{\epsilon}) (1-\sqrt{\epsilon} - \epsilon - \delta) > 1$.
    \end{itemize} 
    This is a contradiction as the events $\AlgEvent'(e,S)$ are disjoint and their probabilities can therefore not sum up to more than one.
\end{proof}

\subsubsection{Up-Link Solutions for Huge Edges}\label{subsubsec:up_link_soln_huge_edges} 

\begin{proof}[Proof of \Cref{lem:covering_huge_edges}]
Fix an event $\scrF = (R, R_{small}, L_{small}) \in \mathrm{events}$. If $y(\scrF) = 0$, then $y_\ell(\scrF)=0$ for all $\ell$. We can define $\mathcal{H}_{\scrF}$ to always output the empty set, satisfying the lemma trivially. Assume $y(\scrF) > 0$.

Let us first define the subproblem. 
Let $E'_{huge}$ be the set of edges that must be covered by the solution $H$:
\[
E'_{huge} = \{e \in R(\scrF) \setminus R_{small}(\scrF) \mid L_e \cap L_{small}(\scrF) = \emptyset\}.
\]
These are the edges designated as ``huge'' in the event $\scrF$ that are not already covered by the links in $L_{small}(\scrF)$.

We then construct a feasible fractional solution $z$ for the Cut LP corresponding to the requirement of covering $E'_{huge}$. We restrict the available links to $L' = L(R(\scrF)) \setminus L_{small}(\scrF)$, i.e., the links relevant to the event but not included in $L_{small}$.

For each link $\ell \in L$, we define $z(\ell)$ as follows:
\[
z(\ell) =
\begin{cases}
\frac{y_l(\scrF)}{\rho \cdot y(\scrF)} & \text{if } l \in L' \\
0 & \text{otherwise}
\end{cases}
\]
We verify that $z$ is a feasible fractional cover for $E'_{huge}$. Consider an edge $e \in E'_{huge}$. By definition, $e \in R(\scrF) \setminus R_{small}(\scrF)$. According to Constraint (13) of the Strong LP:
\[
\sum_{\ell \in L_e} y_\ell(\scrF) \geq \rho \cdot y(\scrF).
\]
Furthermore, by the definition of $E'_{huge}$, $L_e \cap L_{small}(\scrF) = \emptyset$. Thus, the summation is entirely over links in $L'$, i.e., $L_e \subseteq L'$.
The fractional coverage of $e$ by $z$ is:
\[
\sum_{\ell \in L_e} z(\ell) = \sum_{\ell \in L_e} \frac{y_\ell(\scrF)}{\rho \cdot y(\scrF)} = \frac{1}{\rho \cdot y(\scrF)} \sum_{\ell \in L_e} y_\ell(\scrF) \geq \frac{\rho \cdot y(\scrF)}{\rho \cdot y(\scrF)} = 1.
\]
Thus, $z$ is a feasible solution to the Cut LP for covering $E'_{huge}$.

Now, observe that the solution $z$ may utilize in-links or cross-links. We transform $z$ into a solution $z'$ that only uses up-links using the standard splitting procedure (Section 2.1). For every link $\ell$ in the support of $z$, we let $S(\ell)$ denote the set containing each up-link shadow $\ell_e$ for each leading edge $e$ of $\ell$. Note that $|S(\ell)| \leq 2$. We define the new fractional solution $z'$ on these shadows. For each $\ell' \in S(\ell)$, we set $z'(\ell') = z(\ell)$.

Since the union of edges covered by $S(\ell)$ is the same as the edges covered by $\ell$, $z'$ remains a feasible fractional cover for $E'_{huge}$. The total fractional value assigned to the shadows of $\ell$ is bounded:
\begin{equation}\label{eq:bound_doubling}
   \sum_{\ell' \in S(\ell)} z'(\ell') \leq 2 z(\ell). 
\end{equation}

The solution $z'$ is a feasible solution to the Cut LP on an instance consisting solely of up-links. As noted in Section 2.2, the Cut LP is integral for instances containing only up-links. Therefore, $z'$ can be expressed as a convex combination of integral solutions:
\[
z' = \sum_i \lambda_i H_i,
\]
where each $H_i$ is the characteristic vector of an integral up-link solution covering $E'_{huge}$, $\lambda_i \geq 0$, and $\sum_i \lambda_i = 1$.

We define the distribution $\mathcal{H}_{\scrF}$ by sampling the integral solution $H_i$ with probability $\lambda_i$.

We now argue that  $\mathcal{H}_{\scrF}$ satisfies the properties of the lemma. 
\begin{itemize}
    \item \textbf{Exact Coverage:} Since the input instance is assumed to be shadow-complete, we can ensure that the decomposition of $z'$ uses minimal integral solutions $H_i$. A minimal up-link cover ensures that every edge in $E'_{huge}$ is covered exactly once. This satisfies the exact coverage requirement of the lemma.

    \item \textbf{Bounding Expected Usage:} We analyze the expected value of $X_{\scrF}(\ell)$, the number of shadows of $\ell$ included in $H \sim \mathcal{H}_{\scrF}$. By the definition of the rounding procedure, this expectation is exactly the total fractional value of the shadows of $\ell$ in $z'$:
    \begin{equation}\label{eq:expectation_shadows}
        E[X_{\scrF}(\ell)] = \sum_{\ell' \in S(l)} z'(\ell').
    \end{equation}

    If $\ell \in L_{small}(\scrF)$ (or if $\ell$ is not relevant to $\scrF$), then $\ell \notin L'$, so $z(\ell)=0$. Consequently, $z'(\ell')=0$ for all $\ell' \in S(\ell)$, and $E[X_{\scrF}(\ell)] = 0$.

    If $\ell \in L'$ (relevant but not in $L_{small}(\scrF)$), we use the bound derived in \Cref{eq:bound_doubling}:
    \[
    E[X_{\scrF}(\ell)] \leq 2 z(l).
    \]
    Substituting the definition of $z(\ell)$ from \Cref{eq:expectation_shadows}:
    \[
    E[X_{\scrF}(\ell)] \leq 2 \cdot \frac{y_\ell(\scrF)}{\rho \cdot y(\scrF)}.
    \]
    This confirms the required bound and completes the proof.
\end{itemize}
    
\end{proof}

\bibliographystyle{alpha}
\bibliography{ref}
\appendix
\section{Formalizing \Cref{subsec:preprocessing_removing_correlations}}\label{apx:formal_preprocessing}
Here we formalize the pre-processing algorithm from \Cref{subsec:preprocessing_removing_correlations} and prove that the resulting solution is feasible for the Strong LP.

\paragraph{Correlation Correction Algorithm.}
\begin{enumerate}
    \item \textbf{Initialization.}  
    Set \(x' \gets x\), \(y' \gets y\), and \(y'_\ell \gets y_\ell\) for each \(\ell \in L\).

    \item \textbf{While a dangerous link remains in the support of \(x'\):}
    \begin{enumerate}
        \item \textbf{Select edge.}  
        Choose the lowest uncorrelated edge \(f = \{u,v\}\) whose parent edge \(e_u\)
        satisfies \(x'(L_f \cap L_{e_u}) > 0\).
        \label{step:selecting_lowest_edge}

        \item \textbf{Select link.}  
        Choose a link \(\ell = \{w, y\} \in L_f \cap L_{e_v}\) with \(x'(\ell) > 0\),
        where \(w\) is the descendant of \(v\) in \(\ell\).
        \label{step:link_selection}

        \item \textbf{Split the link.}  
        Create two shadows of \(\ell\):
        \begin{itemize}
            \item \(\ell_1\): obtained by moving the endpoint \(w\) to \(u\),
            \item \(\ell_2\): obtained by moving the endpoint \(y\) to \(u\), covering
            the remaining portion of \(P_\ell\).
        \end{itemize}
        \label{step:link_splits}

        \item \textbf{Compute up-link cover.}  
        Let \(H\) be the cheapest set of up-links covering each edge of \(P_{\ell_2}\)
        exactly once and no others.

        \item \textbf{Redistribute link weight.}  
        For each \(\ell' \in H \cup \{\ell_1\}\),
        \[
            x'(\ell') \gets x'(\ell') + x'(\ell).
        \]
        \label{step:increase_x_prime}

        \item \textbf{Update relevant event variables.}  
        For each event \(\scrF = (R, R_\sm, L_\sm)\) with \(y'_\ell(\scrF) > 0\):
        \begin{enumerate}
            \item Set \(\eta = y'(\scrF)\) and \(\eta_{\ell'} = y'_{\ell'}(\scrF)\)
            for all relevant links~\(\ell'\).
            \item Define \(\scrF' = (R', R'_\sm, L'_\sm)\) initially as a copy of~\(\scrF\).
            If \(\ell \in L_\sm\), remove \(\ell\) from \(L'_\sm\) and add each
            \(\ell' \in H \cup \{\ell_1\}\) whose path intersects \(R_\sm\).
            \label{step:select_replacement_event}
            \item Increase event weights:
            \[
                y'(\scrF') \gets y'(\scrF') + \eta,
                \quad
                y'_{\ell'}(\scrF') \gets y'_{\ell'}(\scrF') + \eta_{\ell'}
                \quad \forall \ell' \ne \ell.
            \]
            \label{step:transfer_val_to_replacement}
            \item For each \(\ell' \in H \cup \{\ell_1\}\) relevant to~\(\scrF\),
            increase \(y'_{\ell'}(\scrF')\) by~\(\eta_\ell\).
            \label{step:increase_ell_replacements}
            \item Decrease the original event’s values:
            \[
                y'(\scrF) \gets y'(\scrF) - \eta,
                \quad
                y'_{\ell'}(\scrF) \gets y'_{\ell'}(\scrF) - \eta_{\ell'} \text{ for all }\ell'.
            \]
            \label{step:decrease_weights_event}
        \end{enumerate}
        \label{step:alter_event_values_relevant_events}

        \item \textbf{Drop the original link.}  
        Set \(x'(\ell) \gets 0.\)
        \label{step:drop_x_prime_ell}
    \end{enumerate}

    \item \textbf{Output.}  
    Return \((x', y', \{y'_\ell\}_{\ell \in L})\).
\end{enumerate}

\paragraph{Proof of Feasibility}

First we see that $(x', y', \{y_\ell\}_{\ell \in L})$ is well-defined. Formally the event $\scrF'$ selected Step \ref{step:select_replacement_event} is always in $\events(R(\scrF))$. This follows because the coverage of $\ell$ is identical in multiplicity to that of $H  \cup \{\ell_1\}$:
To conclude the proof of \Cref{lem:removing_overlap_uncorrelated}
we argue that $(x', y', \{y'_\ell\}_{\ell \in L})$ satisfies each constraint of the Strong LP (i.e ~\Cref{constr:odd_cut_consraint}--\Cref{constr:extension_consistency_with_ell}). For clarity we will use $\scrF, \scrF'$ to denote events specifically chosen by the algorithm and $\scrE, \scrE'$ etc to denote general events.

\begin{fact}
    $x'$ satisfies the Odd Cut LP constraints (\Cref{constr:odd_cut_consraint}).
\end{fact}
\begin{proof}
    Select $S \subseteq V$ such that $|\delta_E(S)|$ is odd, and 
    consider the corresponding odd cut constraint:
    \begin{align*}
        x'(\delta_L(S)) + \sum_{e \in \delta_E(S)} x'(L_e) \ge |\delta_E(S)| + 1
    \end{align*}
    We argue that if the constraint is feasible before a given execution of the while loop it remains feasible afterward.

    Note that \(x'(L_e)\) does not change for any \(e \in \delta_E(S)\).  
    Moreover, if a link \(\ell \in \delta_L(S)\) is modified by the algorithm, 
    then at least one of its replacements—either \(\ell_1\) or one of the up-links in \(H\)—also crosses the cut \(\delta_L(S)\).  
    Hence, \(x'(\delta_L(S))\) cannot decrease, and the inequality continues to hold.

\end{proof}

\begin{fact}\label{fact:coverage_holds_after_corr_correct}
The solution $(x', y', \{y'_\ell\}_{\ell \in L})$ satisfies the Coverage Constraints (\Cref{constr:cov_consistency}).
\end{fact}
\begin{proof}
    Let $e \in E$, and consider the coverage constraint corresponding to $e$:
        $$
        \sum_{\scrE \in \events(e)}y'(\scrE)=1.
        $$
    We argue that if the constraint is feasible before any given execution of Step \ref{step:alter_event_values_relevant_events}, then it remains feasible after. Consider an execution of Step \ref{step:alter_event_values_relevant_events} in which $\scrF$ is selected.

    If $\scrF \notin \events(e)$  all variables appearing in $e$'s coverage constraint remain unchanged. Therefore we assume that $\scrF \in \events(e)$. But then $\scrF'$ is also in $\events(e)$.
    Since $y'(\scrF)$ decreases by $\eta$ and $y'(\scrF')$ increases by $\eta$ the left hand side of the constraint remains unchanged.
\end{proof}
 
\begin{fact}\label{fact:marginal_holds_after_corr_correct}
The solution $(x', y', \{y'_\ell\}_{\ell \in L})$ satisfies the Marginal-Preserving Constraints (\Cref{constr:marginal_pres_const_child}).
\end{fact}

\begin{proof}
    Let $e \in E$ and $\ell^* \in L_e$ and consider the marginal preserving constraint corresponding to $e$ and $\ell^*$:
    $$
        x'(\ell^*) = \sum_{\scrE \in \events(e)}y'_{\ell^*}(\scrE).
    $$
    We argue that if the constraint is satisfied at the beginning of an iteration of the while loop, it remains satisfied at its end.  
    Consider an iteration in which the link $\ell$ is selected in Step~\ref{step:link_selection}.
    
    There are two possible cases for $e$:  
    $e \notin P_\ell$ or $e \in P_{\ell}$. When $e \notin P_\ell$ no event in $\events(e)$ is ever selected in Step \ref{step:alter_event_values_relevant_events}, and therefore the right hand side of the constraint remains unchanged. The value $x'(\ell^*)$ remains unchanged because $x'$ only changes its value on links in $H \cup \{\ell_1, \ell\}$.

    We consider when $e \in P_\ell$. Let $\ell'$ be the unique link in $H \cup \{\ell_1\}$ covering $e$. There are three cases to consider:

    \begin{itemize}
        \item If $\ell^* = \ell$, then after the iteration of the while loop, 
    $x'(\ell^*)$ is set to zero by Step~\ref{step:drop_x_prime_ell}.  
    Note that each event $\scrE \in \events(e)$ with 
    $y'_{\ell^*}(\scrE) > 0$ at the beginning of the iteration 
    will be selected as $\scrF$ in some execution of 
    Step~\ref{step:alter_event_values_relevant_events}.  
    During that execution, $y'_{\ell^*}(\scrF)$ is set to zero in 
    Step~\ref{step:decrease_weights_event} and it is not modified elsewhere.  
    Moreover, $y'_{\ell^*}(\scrF')$ never increases in subsequent steps.  
    Therefore, upon completion of the iteration, both sides of the constraint 
    evaluate to zero.

        \item If $\ell^* = \ell'$, then $x'(\ell^*)$ increases by the value of $x'(\ell)$ at the beginning of the iteration, which we denote by $\zeta$, due to Step \ref{step:increase_x_prime}. We know at the beginning of the iteration the marginal preserving constraint for $\ell$ and $e$ holds, and hence:
        $$\zeta = \sum_{\scrE \in \events(e)}y'_\ell(\scrE).$$  We argue that the right hand side of $\ell^*$'s constraint increases by the right hand side of the above equality.
        To see this consider any execution of Step \ref{step:alter_event_values_relevant_events} where we select an event $\scrF \in \events(e)$. 
        Then if $\Delta(\scrF')$ is the change in the value of $y_{\ell^*}(\scrF')$ during this execution, we know that $\Delta(\scrF')=y_\ell'(\scrF)$ due to Steps \ref{step:transfer_val_to_replacement}, and {step:increase_ell_replacements}.

        \item If $\ell^* \in L_e \setminus \{\ell, \ell'\}$ then there is no change in the value of $x'(\ell^*)$. Furthermore, in any execution of Step \ref{step:alter_event_values_relevant_events} there is no change to the right hand side of the constraint due to Step \ref{step:increase_ell_replacements}.
    \end{itemize}

\end{proof}

We omit the proof of the following fact which can be easily derived:
\begin{fact} The solution $(x', y', \{y'_\ell\}_{\ell \in L})$ the Non-Negativity and Link Consistency Constraints (\Cref{constr:non_neg,eq:non_inclusion_links_on_huge_edges,constr:inclision_on_small,eq:non_inclusion_links_on_huge_edges,constr:huge_edges_const_child}).
\end{fact}

The following fact will aid us in proving the remaining constraints of the Strong LP hold in \Cref{fact:consistency_holds_link_replacement}.
\begin{fact}\label{obs:event_types}
Consider any iteration of the while loop in the correlation correction algorithm, where $\ell$ is selected in Step~\ref{step:link_selection}.  
Then the following enumerates all possibilities for each event $\scrE \in \events$,  
where $\Delta(\scrE)$ denotes the change in $y'(\scrE)$ between the beginning and end of the iteration.
\begin{enumerate}[label=(T\arabic*)]
    \item \label{cond:contains_ell}
    $\ell \in L(\scrE)$ and $P_{\ell} \cap R_\sm(\scrE) \neq \emptyset$.  
    In this case, $y'(\scrE) = 0$ upon completion of the iteration.
    
    \item \label{cond:no_overlap_P_ell}
    $P_{\ell} \cap R_\sm(\scrE) = \emptyset$.  
    Then $\Delta(\scrE) = 0$.

    \item \label{cond:overlap_P_ell_no_containment}
    $P_{\ell} \cap R_\sm(\scrE) \neq \emptyset$ and $\ell \notin L(\scrE)$.  
    Then either $\Delta(\scrE) = 0$, or there exists a unique 
    \ref{cond:contains_ell}-type partner event $\scrE'$ such that 
    $\Delta(\scrE) = y(\scrE')$.  
    Furthermore, $\scrE$ is obtained from $\scrE'$ by removing $\ell$ and 
    adding each link in $H \cup \{\ell_1\}$ that covers an edge of $R_\sm(\scrE)$.
\end{enumerate}
\end{fact}

\begin{proof}
If $\scrE$ satisfies Condition~\ref{cond:contains_ell}, then $y'(\scrE)$ becomes zero after the iteration of the while loop.  
Indeed, there exists an execution of Step~\ref{step:alter_event_values_relevant_events} in which $\scrF = \scrE$, 
and in that execution the event $\scrF'$ is distinct from~$\scrE$, causing $y'(\scrE)$ to be set to zero in 
Step~\ref{step:decrease_weights_event}.

If $\scrE$ satisfies Condition~\ref{cond:no_overlap_P_ell}, then $\Delta(\scrE) = 0$.  
This follows because either:
\begin{itemize}
    \item $P_\ell \cap R(\scrE) = \emptyset$, in which case $\scrE$ is never selected as $\scrF$ or $\scrF'$ in any execution of 
    Step~\ref{step:alter_event_values_relevant_events}, or
    \item $P_\ell \cap R(\scrE)$ is a non-empty subset of $R(\scrE) \setminus R_\sm(\scrE)$, in which case 
    whenever $\scrE \in \{\scrF, \scrF'\}$ we have $\scrE = \scrF = \scrF'$, and thus no change occurs.
\end{itemize}

Finally, suppose $\scrE$ satisfies Condition~\ref{cond:overlap_P_ell_no_containment}.  
Then one of two cases holds:
\begin{itemize}
    \item There exists a unique Condition~\ref{cond:contains_ell}-type event $\scrE'$ such that 
    $\scrE$ is obtained from $\scrE'$ by removing $\ell$ and adding each link in 
    $H \cup \{\ell_1\}$ that covers an edge of $R_\sm(\scrE)$.  
    In the corresponding iteration where $\scrF = \scrE'$, we have $\scrF' = \scrE$, 
    and hence $\Delta(\scrE) = y'(\scrE')$.
    \item Otherwise, $\scrE$ is never considered as $\scrF$ or $\scrF'$ during the algorithm, and therefore $\Delta(\scrE) = 0$.
\end{itemize}
\end{proof}

\begin{fact}\label{fact:consistency_holds_link_replacement} The solution $(x', y', \{y'_\ell\}_{\ell \in L})$ satisfies the Consistency Constraints Among Events Constraints (\Cref{constr:extension_consistency,constr:extension_consistency_with_ell}).
\end{fact}
\begin{proof}
    Consider an event $\scrE \in \events$ and some $Q \in \subtrees$ for which $R(\scrE) \subseteq Q$, and its corresponding \Cref{constr:extension_consistency}:
    $$
        y(\scrE) = \sum_{\scrE' \in \Ext(\scrE, Q)} y(\scrE').
    $$
    Using the same approach as in \Cref{fact:marginal_holds_after_corr_correct} we argue that if the constraint is feasible before any given execution of the while loop it remains feasible after.

    We consider the following scenarios:
    \begin{itemize}
        \item If $\scrE$ is a \ref{cond:contains_ell} event, then all elements of $\Ext(\scrE, Q)$ satisfy the same property, meaning both sides of the equation are zero after the iteration of the while loop.
        \item If $\scrE$ is a \ref{cond:no_overlap_P_ell} event, then if $\scrE' \in \Ext(\scrE, Q)$ is of type \ref{cond:contains_ell} or \ref{cond:overlap_P_ell_no_containment}, one can observe that its partner is also in $\Ext(\scrE, Q)$.
        \item Suppose $\scrE$ is a Condition~\ref{cond:overlap_P_ell_no_containment} event.  
Observe that every element of $\Ext(\scrE, Q)$ is also a 
Condition~\ref{cond:overlap_P_ell_no_containment} event.  
Moreover, if $\scrF$ is selected in an execution of 
Step~\ref{step:alter_event_values_relevant_events}, then the same holds for 
all events in $\Ext(\scrF, Q)$ whenever this set is well-defined.  

Therefore, if $\scrE$ has a partner $\scrG$ of 
type~\ref{cond:contains_ell}, then all elements of $\Ext(\scrG, Q)$ are 
likewise selected as $\scrF$ at some point during the execution of the algorithm.  
Since
\[
    y'(\scrG) = \sum_{\scrG' \in \Ext(\scrG, Q)} y'(\scrG')
\]
at the beginning of the iteration, the total increase on the right-hand side 
of the constraint corresponding to $(\scrE, Q)$ is exactly $y'(\scrG)$.  
This follows by pairing each event in $\Ext(\scrE, Q)$ with its unique partner 
in $\Ext(\scrG, Q)$, ensuring that the aggregated contribution of all such 
pairs preserves the equality required by the extension consistency constraint.

    \end{itemize}
     
An analogous argument can be made for \Cref{constr:extension_consistency_with_ell}.   
    
\end{proof}

Combining the above facts we obtain that $(x', y', \{y_\ell\}_{\ell \in L})$ is indeed feasible for the Strong LP.

\section{Formal Construction of the Pairing Function}\label{appendix:construct_phi}

Next we describe how to formally construct the pairing distribution $\phi$ we outlined in \Cref{subsubsec:feas_struc_frac}.
Consider an event~$\scrF$ that could be sampled in
Step~\ref{step:select_event} for which either $e=f$, or $e \in R(\scrF) \setminus R_\sm(\scrF)$.
We first describe how~$\phi$ is defined \emph{conditionally on a fixed}~$\scrF$;
the overall pairing distribution is then obtained by taking the expectation of this
conditional definition over the random choice of~$\scrF$ in the Structured-Sampling Algorithm. In the following, for simplicity we will use $\phi(\scrE', \scrE'')$ to denote $\phi(\scrE', \scrE'' | \scrF)$.

\newcommand{\agree}{\mathsf{agree}}
\newcommand{\pagree}{\mathsf{pagree}}

\paragraph{Case 1: $e = f$.}

.

For each event $\scrE' \in \events(E')$ and each event $\scrE'' \in \events(E'')$ we define $\phi(\scrE', \scrE'' | \scrF)$ as follows:

If assigning $\phi(\scrE', \scrE'')>0$ would violate the Structural Consistency requirement, we set $$\phi(\scrE', \scrE'')=0.$$ Otherwise, Structural Consistency implies the restriction of $\scrE'$ onto $E^*(e)$ equals the restriction of $\scrE''$ onto $E^*(e)$; we denote
this common restriction by $\hat{\scrE}$.

\newcommand{\corr}{\mathsf{corr}}
\newcommand{\uncorr}{\mathsf{uncorr}}

For a link $\ell$, we partition its set of leading edges as
\[
\Lead(\ell)
\;=\;
\Lead_{\corr}(\ell)
\;\cup\;
\Lead_{\uncorr}(\ell),
\]
where $\Lead_{\corr}(\ell)$ denotes the set of correlated leading edges of $\ell$,
and $\Lead_{\uncorr}(\ell)$ denotes the set of uncorrelated leading edges.

Next we define some events that will aid us in setting the value for $\phi(\scrE', \scrE'')$ in the remaining case.
In the following we will consider an event $\scrE \in \primeevents$, a link-set assignment
$F=\{F_e\}_{e \in E}$, an edge $e \in E(\scrE)$, and a link $\ell \in L_e$.

We write $\agree(\scrE, F, \ell, e)$ to denote the event that the path
$Q_\ell(e,F)$ coincides exactly with the shadow of $\ell$ contained in $L(\scrE)$.
If $Q_\ell(e,F)$ is the empty path, we interpret this as requiring that no shadow
of $\ell$ is contained in $L(\scrE)$.

Similarly, when $e$ is a leading edge of $\ell$ we write $\pagree(\scrE, F, \ell, e)$ to denote the event that the path
$P_\ell(e,F)$ coincides exactly with the shadow of $\ell$ contained in $L(\scrE)$.
As before, if $P_\ell(e,F)$ is the empty path, we interpret this as requiring that no shadow
of $\ell$ is contained in $L(\scrE)$.  Consult \Cref{fig:consist_pconsist} for a visualization.

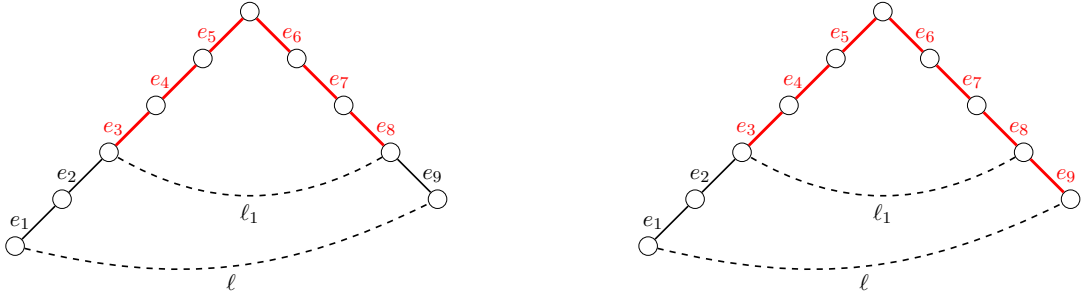
\begin{figure}[t]
\centering

\begin{minipage}[t]{0.49\linewidth}
\centering
\scalebox{0.78}{%
\begin{tikzpicture}[
    node/.style={circle, draw, fill=white, minimum size=9pt, inner sep=2pt},
    every label/.style={font=\small}
]
\def\dx{0.8}
\def\dy{0.8}
\node[node] (r) at (0,0) {};
\foreach \i in {1,...,5}{
    \pgfmathsetmacro{\x}{-\dx*\i}
    \pgfmathsetmacro{\y}{-\dy*\i}
    \node[node] (L\i) at (\x,\y) {};
}
\foreach \i in {1,...,4}{
    \pgfmathsetmacro{\x}{\dx*\i}
    \pgfmathsetmacro{\y}{-\dy*\i}
    \node[node] (R\i) at (\x,\y) {};
}
\draw[thick] (L5) -- node[midway, left] {$e_1$} (L4);
\draw[thick] (L4) -- node[midway, left] {$e_2$} (L3);
\draw[line width=1.4pt, red] (L3) -- node[midway, left] {$e_3$} (L2);
\draw[line width=1.4pt, red] (L2) -- node[midway, left] {$e_4$} (L1);
\draw[line width=1.4pt, red] (L1) -- node[midway, left] {$e_5$} (r);
\draw[line width=1.4pt, red] (r) -- node[midway, right] {$e_6$} (R1);
\draw[line width=1.4pt, red] (R1) -- node[midway, right] {$e_7$} (R2);
\draw[line width=1.4pt, red] (R2) -- node[midway, right] {$e_8$} (R3);
\draw[thick] (R3) -- node[midway, right] {$e_9$} (R4);
\draw[dashed, thick, bend right=20] (L5) to node[below] {$\ell$} (R4);
\draw[dashed, thick, bend left=30]  (R3) to node[below] {$\ell_1$} (L3);
\end{tikzpicture}%
}
\end{minipage}%
\hfill%
\begin{minipage}[t]{0.49\linewidth}
\centering
\scalebox{0.78}{%
\begin{tikzpicture}[
    node/.style={circle, draw, fill=white, minimum size=9pt, inner sep=2pt},
    every label/.style={font=\small}
]
\def\dx{0.8}
\def\dy{0.8}
\node[node] (r2) at (0,0) {};
\foreach \i in {1,...,5}{
    \pgfmathsetmacro{\x}{-\dx*\i}
    \pgfmathsetmacro{\y}{-\dy*\i}
    \node[node] (L\i2) at (\x,\y) {};
}
\foreach \i in {1,...,4}{
    \pgfmathsetmacro{\x}{\dx*\i}
    \pgfmathsetmacro{\y}{-\dy*\i}
    \node[node] (R\i2) at (\x,\y) {};
}
\draw[thick] (L52) -- node[midway, left] {$e_1$} (L42);
\draw[thick] (L42) -- node[midway, left] {$e_2$} (L32);
\draw[line width=1.4pt, red] (L32) -- node[midway, left] {$e_3$} (L22);
\draw[line width=1.4pt, red] (L22) -- node[midway, left] {$e_4$} (L12);
\draw[line width=1.4pt, red] (L12) -- node[midway, left] {$e_5$} (r2);
\draw[line width=1.4pt, red] (r2) -- node[midway, right] {$e_6$} (R12);
\draw[line width=1.4pt, red] (R12) -- node[midway, right] {$e_7$} (R22);
\draw[line width=1.4pt, red] (R22) -- node[midway, right] {$e_8$} (R32);
\draw[line width=1.4pt, red] (R32) -- node[midway, right] {$e_9$} (R42);
\draw[dashed, thick, bend right=20] (L52) to node[below] {$\ell$} (R42);
\draw[dashed, thick, bend left=30]  (R32) to node[below] {$\ell_1$} (L32);
\end{tikzpicture}%
}
\end{minipage}

\caption{
Illustration of \(\agree(\scrE, F, \ell, e)\) and \(\pagree(\scrE, F, \ell, e)\).
The link \(\ell_1\) is the unique shadow of \(\ell\) contained in \(\scrE\). In both figures the red edges represent the set of edges $e$ such that $\ell$ is assigned to $F_e$.
Consider $e \in \{e_3,\dots, e_8\}$. In the left image $Q_\ell(e, F)$ contains the edges $e_3, \dots, e_8$ and
\(\agree(\scrE, F, \ell, e)\) is true for each such edge. Similarly for $e \in \{e_1, e_2, e_9\}$ $Q_\ell(e, F)$ contains no edges and
\(\agree(\scrE, F, \ell, e)\) is true for each such edge.
However in the right image $Q_\ell(e, F)$ contains the edges $e_3, \dots, e_9$ and \(\agree(\scrE, F, \ell, e)\) is false.
In the right image $P_\ell(e_5, F)$ contains $e_3, e_4$ and $e_5$ and 
\(\pagree(\scrE, F, \ell, e_5)\) is true. However, $P_\ell(e_6, F)$ contains $e_6, \dots, e_9$ and \(\pagree(\scrE, F, \ell, e_6)\) is false.
}
\label{fig:consist_pconsist}
\end{figure}

Whenever $\scrE \in \events(E^*(e))$, we say that $F$ \emph{partially agrees} with $\scrE$ on the pair $(\ell,g)$ if the corresponding $\pagree(\scrE,F,\ell,g)$
event holds and \emph{agrees} with $\scrE$ on 
the pair $(\ell,g)$ if $\agree(\scrE,F,\ell,g)$
holds.

Next we build up further our notation. Let
\[
\hat{E} = E^*(e)\cup\{e',e''\},
\qquad
\hat{L} = L(\hat{E}).
\] We define $\calG(\scrE,F)$ to be the event that
$F$ agrees  with all pairs of links $(\ell, g)$ where $\ell \in \hat{L} \cap L_e$ and $g=e$ and partially agrees with $\scrE$ on all pairs of links $(\ell, g)$ with $\ell \in \hat{L} \setminus L_e$ and $g \in \Lead_\corr(\ell)$.
Formally,
\[
\calG(\scrE, F)
\;=\;
\Bigl(\bigwedge_{\ell \in \hat{L} \cap L_e}\agree(\scrE, F, \ell, e)\Bigr)
\;\land\;
\Bigl(\bigwedge_{\ell \in \hat{L} \setminus L_e}
        \;\bigwedge_{g \in \Lead_{\corr}(\ell)}
        \pagree(\scrE, F, \ell, g)\Bigr).
\]

When $\calG(\scrE, F)$ holds we say that $F$ is in \emph{base agreement} with $\scrE$.

For any uncorrelated child $\bar e$ of $e$ and $\bar{\scrE} \in \events(E^*(e)\cup\{\bar e\})$,
we define $\calG_{\bar e}(\bar{\scrE},F)$ to be the event that $F$ partially agrees with $\bar \scrE$ on all pairs $(\ell,\bar e)$ for
links $\ell \in \hat{L} \cap L_{\bar e}$. Formally,
\[
\calG_{\bar e}(\bar{\scrE},F)
\;:=\;
\bigwedge_{\ell \in \hat{L} \cap L_{\bar e}}
\pagree(\bar{\scrE},F,\ell,\bar e).
\]
When $\calG_{\bar e}(\bar{\scrE},F)$ holds we say that
$F$ \emph{locally agrees} with $\bar{\scrE}$ at $\bar e$.

Recall that $L_e \cap L_{\bar{e}}$ contains no links in the support of the Strong LP solution and hence we can assume they do not exist.

In addition, for an uncorrelated child $\tilde e \neq \bar e$,
we let $\calG_{\bar e,\tilde e}(\bar{\scrE},F)$ be the event that $F$ partially agrees with $\bar \scrE$ on all pairs $(\ell,\tilde e)$ for
links $\ell \in \hat L \cap L_{\bar e} \cap L_{\tilde e}$. Formally,
\[
\calG_{\bar e, \tilde{e}}(\bar{\scrE},F)
\;:=\;
\bigwedge_{\ell \in \hat{L} \cap L_{\bar e} \cap L_{\tilde{e}}}
\pagree(\bar{\scrE},F,\ell,\tilde e).
\]

When $\calG_{\bar e,\tilde e}(\bar{\scrE},F)$ holds we say that
$F$ \emph{cross-agrees} with $\bar{\scrE}$ at $\tilde e$.

Before proceeding further, we note one observation about our distribution $\calD$ that will be helpful in simplifying our notation.

\begin{observation}[Conditional independence of descendant assignments]
\label{obs:independence_descendants}
Let $\scrF \in \events(T_e)$ be an event that can be sampled in
Step~\ref{step:select_event} when $f=e$ in the Structured-Sampling Algorithm.
Fix distinct edges $\bar e,\tilde e \in E(e)\setminus E^*(e)$ and an event
$\bar{\scrE} \in \primeevents(E^*(e)\cup\{\bar e\})$.
Let $\scrF_{\bar e} \in \Ext(\scrF, T_e\cup\{\bar e\})$ and
$\scrF_{\tilde e} \in \Ext(\scrF, T_e\cup\{\tilde e\})$.

Then
\[
\Pr_{F \sim (\calD \mid \scrF_{\bar e})}
\!\Bigl[\calG_{\bar e}(\bar{\scrE}, F)\ \Bigm|\ 
\calG\!\bigl(\Restr'(\bar{\scrE}, T_e), F\bigr)\Bigr]
\;=\;
\Pr_{F \sim (\calD \mid \scrF_{\bar e})}
\!\bigl[\calG_{\bar e}(\bar{\scrE}, F)\bigr],
\]
and
\[
\Pr_{F \sim (\calD \mid \scrF_{\tilde e})}
\!\Bigl[\calG_{\bar e, \tilde e}(\bar{\scrE}, F)\ \Bigm|\ 
\calG\!\bigl(\Restr'(\bar{\scrE}, T_e), F\bigr)\Bigr]
\;=\;
\Pr_{F \sim (\calD \mid \scrF_{\tilde e})}
\!\bigl[\calG_{\bar e, \tilde e}(\bar{\scrE}, F)\bigr].
\]
\end{observation}

\begin{proof}
Condition on the events $\scrF$ and $\scrF_{\bar e}$.
Since $\bar e$ is an uncorrelated child of $e$, all random choices that determine
the assignments $F_d$ for edges $d \in E(T_{\bar e}^+)$ are made independently
during the recursive processing of the subtree $T_{\bar e}^+$, and are independent
of the random choices determining the assignments involved in
$\calG(\Restr'(\bar{\scrE},T_e),F)$.
The first equality follows. The second equality follows by the same argument,
with $\scrF_{\tilde e}$ and $\calG_{\bar e,\tilde e}$ in place of
$\scrF_{\bar e}$ and $\calG_{\bar e}$.
\end{proof}

By \Cref{obs:independence_descendants}, all conditional probabilities above are
independent of $\calG(\hat{\scrE},F)$, and we henceforth omit this conditioning. To simplify notation we introduce the following:

\newcommand{\Deltae}{E(e)\setminus\bigl(E^*(e)\cup\{e',e''\}\bigr)}

For $\scrF' \in \events(T_e \cup \{e', e''\})$, $\bar e \in \{e', e''\}$,
$\tilde e \in \Deltae$, and $\bar{\scrE} \in \events(T_e \cup \{\bar e\})$, we define $\Lambda(\scrF', \bar e, \tilde e, \bar{\scrE})$ to be the probability of the following experiment:

\begin{itemize}
    \item Restrict $\scrF'$ to $T_e \cup \{\bar e\}$.
    \item Extend this restriction to an event on $T_e \cup \{\bar e,\tilde e\}$ according to the Strong LP extension probabilities.
    \item Restrict the resulting event to $T_e \cup \{\tilde e\}$.
    \item Run the Structured-Sampling Algorithm conditioned on the restriction obtained in the previous step.

\end{itemize}

We then take the probability that the resulting assignment $F$
cross-agrees with $\bar{\scrE}$ at $\tilde e$.

Then,

\begin{equation}\label{eq:sigma_definition}
\begin{aligned}
\Lambda(\scrF', \bar e, \tilde e, \bar{\scrE})
\;=\;&
\sum_{\scrF_{\bar e,\tilde e} \in
      \Ext\!\bigl(\Restr(\scrF', T_e \cup \{\bar e\}),\, T_e \cup \{\bar e,\tilde e\}\bigr)}
\frac{y(\scrF_{\bar e,\tilde e})}
     {y(\Restr(\scrF', T_e \cup \{\bar e\}))}
\\
&\quad\cdot
\Pr_{F \sim (\calD \mid \Restr(\scrF_{\bar e,\tilde e}, T_e \cup \{\tilde e\}))}
\Bigl[\calG_{\bar e,\tilde e}(\bar{\scrE}, F)\Bigr].
\end{aligned}
\end{equation}

Now we define $\Psi(\scrF', \bar e, \bar{\scrE})$ to equal the probability of locally agreeing
with $\bar{\scrE}$ at $\bar e$ and cross-agreeing with $\bar{\scrE}$ at every
$\tilde e \in \Deltae$, where the experiment defining $\Lambda(\scrF',\bar e,\tilde e,\bar{\scrE})$
is performed independently for each $\tilde e$. Then,

\begin{equation}\label{eq:psi_definition}
\begin{aligned}
\Psi(\scrF', \bar e, \bar{\scrE})
\;=\;&
\Pr_{F \sim (\calD \mid \Restr(\scrF', T_e \cup \{\bar e\}))}
\Bigl[\calG_{\bar e}(\bar{\scrE}, F)\Bigr]
\cdot
\prod_{\tilde e \in \Deltae}
\Lambda(\scrF', \bar e, \tilde e, \bar{\scrE}) .
\end{aligned}
\end{equation}

We are now in a position to define $\phi(\scrE', \scrE'')$.
Recall that
$\hat{\scrE}=\primeRestr(\scrE', T_e)=\primeRestr(\scrE'', T_e)$ and that the Structured-Sampling Algorithm samples the extension events
$\scrF_{e'}$ and $\scrF_{e''}$ independently in Step~\ref{step:select_stars_plus_one_edge}.
To define $\phi(\scrE',\scrE'')$, we instead couple these choices by first sampling an
event $\scrF' \in \Ext(\scrF, T_e\cup\{e',e''\})$, and then using its restrictions to
$T_e\cup\{e'\}$ and $T_e\cup\{e''\}$ to simulate
Step~\ref{step:shadow_assignment_flagged} for $g=e'$ and $g=e''$ (considering only
$g' \in E(e)\setminus(E^*(e)\cup\{e',e''\})$).
We define $\phi(\scrE',\scrE'')$ to be the probability that this coupled experiment
marks $\scrE'$ and $\scrE''$.
For each $\bar e \in \{e',e''\}$, we write $\bar{\scrE}$ to denote $\scrE'$ if
$\bar e = e'$ and $\scrE''$ if $\bar e = e''$.

With the help of our simplified notation we can express $\phi(\scrE', \scrE'')$ as follows:

\begin{equation}\label{eq:phi_definition_simplified}
\begin{aligned}
\phi(\scrE', \scrE'')
\;=\;&
\Pr_{F \sim (\calD \mid \scrF)}
\bigl[\calG(\hat{\scrE}, F)\bigr]\Biggl(
\sum_{\scrF' \in \Ext(\scrF, T_e \cup \{e', e''\})}
\frac{y(\scrF')}{y(\scrF)}
\\
&\quad\cdot
\Psi(\scrF', e', \scrE') \cdot \Psi(\scrF', e'', \scrE'')
\Biggr).
\end{aligned}
\end{equation}


By construction $\phi$ satisfies the Structural Consistency Property and therefore to prove that it is a pairing distribution it remains to show that the marginal consistency property holds. We show that the property holds for each event $\scrE' \in \events(E')$ as the proof for events $\scrE'' \in \events(E'')$ is identical.

\newcommand{\Struct}{\mathsf{Struct}}
We will use $\Struct(\scrE')$ to denote the set of events $\scrE'' \in \events(E'')$ that are structurally consistent with $\scrE'$.
In the following we write $\hat{\scrE} = \primeRestr(\scrE', E^*(e))$. As always if $\scrE'' \in \Struct(\scrE')$, then $\primeRestr(\scrE'', E^*(e))=\hat{\scrE}$.

\begin{equation}\label{eq:proof_marginal_consistency}
    \begin{aligned}
        \sum_{\scrE'' \in \events(E'')} \phi(\scrE', \scrE'')
        =& 
        \sum_{\scrE'' \in \Struct(\scrE')} \phi(\scrE', \scrE'')\\
        =\;&
        \sum_{\scrE'' \in \Struct(\scrE')}
        \Biggl(
        \Pr_{F \sim (\calD \mid \scrF)}
        \bigl[\calG(\hat{\scrE}, F)\bigr]\Biggl(
        \sum_{\scrF' \in \Ext(\scrF, T_e \cup \{e', e''\})}
        \frac{y(\scrF')}{y(\scrF)}
        \\
        &\quad\cdot
        \Psi(\scrF', e ', \scrE') \cdot \Psi(\scrF', e '', \scrE'')
        \Biggr) \Biggr)\\
        =\;&
        \Pr_{F \sim (\calD \mid \scrF)}
        \bigl[\calG(\hat{\scrE}, F)\bigr]\Biggl(
        \sum_{\scrF' \in \Ext(\scrF, T_e \cup \{e', e''\})}
        \frac{y(\scrF')}{y(\scrF)}
        \\
        &\quad\cdot
        \Psi(\scrF', e ', \scrE') \cdot \sum_{\scrE'' \in \Struct(\scrE')}\Psi(\scrF', e '', \scrE'')
        \Biggr)\\
        =\;&
        \Pr_{F \sim (\calD \mid \scrF)}
        \bigl[\calG(\hat{\scrE}, F)\bigr]\Biggl(
        \sum_{\scrF' \in \Ext(\scrF, T_e \cup \{e', e''\})}
        \frac{y(\scrF')}{y(\scrF)}
        \\
        &\quad\cdot
        \Psi(\scrF', e ', \scrE') \cdot \Pr_{F \sim (\calD | \Restr(\scrF', T_e \cup \{e''\}))} \left[\calG_{e'', e'}(\scrE', F) \right]
        \Biggr)\\
        =\;&
        \Pr_{F \sim (\calD \mid \scrF)}
        \bigl[\calG(\hat{\scrE}, F)\bigr]\Biggl(
        \sum_{\scrF_{e'} \in \Ext(\scrF, T_e \cup \{e'\})} \frac{y(\scrF_{e'})}{y(\scrF)}
        \Biggl( 
        \sum_{\scrF_{e', e''} \in \Ext(\scrF, T_e \cup \{e', e''\})} \frac{y(\scrF_{e', e''})}{y(\scrF_{e'})}    \\
        &\quad\cdot
        \Psi(\scrF_{e', e''}, e ', \scrE') \cdot \Pr_{F \sim (\calD | \Restr(\scrF_{e', e''}, T_e \cup \{e''\}))} \left[\calG_{e'', e'}(\scrE', F) \right]
        \Biggr) \Biggr)\\ 
        =\;&
    \Pr_{F \sim (\mathcal{D} \mid \scrF)}
      \bigl[\calG(\hat{\scrE}, F)\bigr]
    \\
    & \cdot
    \sum_{\scrF_{e'} \in \Ext(\scrF, T_e \cup \{e'\})}
         \frac{y(\scrF_{e'})}{y(\scrF)}
        \Pr_{F \sim (\calD \mid \scrF_e')} 
        \!\Bigl[\calG_{e'}(\scrE', F)\Bigr]
        \\
        &\cdot
    \prod_{\tilde{e} \in E(e) \setminus (E^*(e) \cup \{ e'\})} 
    \left(
    \sum_{\scrF_{e',\tilde{e}} \in \Ext(\scrF_{e'}, T_e \cup \{e', \tilde{e}\})}
    \frac{y(\scrF_{e',\tilde{e}})}{y(\scrF_{e'})}
     \cdot
    \Pr_{F' \sim (\mathcal{D} \mid \Restr(\scrF_{e',\tilde{e}},\, T_e \cup \{\tilde{e}\}))}[
    \calG_{e',\tilde e}(\scrE', F')] \right)\\
    &= y'(\scrE')
    \end{aligned}
\end{equation}
The first equality holds since 
$\phi(\scrE', \scrE'') = 0$ whenever 
$\scrE'' \notin \Struct(\scrE')$.
The second equality follows directly from 
\Cref{eq:phi_definition_simplified}.
The third equality is obtained by rearranging the order of summation.
The fourth equality comes from explicitly evaluating
\[
    \sum_{\scrE'' \in \Struct(\scrE')}
    \Psi(\scrF', e'', \scrE'').
\]
The fifth equality uses the extension consistency constraint 
of the Strong LP.
The penultimate equality is obtained by expanding each term
\[
    \Psi(\scrF_{e',e''},\, e',\, \scrE'').
\]
Finally, the last equality follows because this quantity equals
the probability that $\scrE'$ is marked, which by definition
is $y'(\scrE')$.

\paragraph{Case 2: $e \in R(\scrF) \setminus R_\sm(\scrF)$.}
Here, $e$ is a huge edge correlated with the root $f$. We define the conditional
pairing distribution $\phi(\scrE', \scrE'' \mid \scrF)$ to be the joint probability
that the marked events corresponding to $\scrE'$ and $\scrE''$ are both produced
by the algorithm in the final marking Step~\ref{step:mark}.

This definition naturally ensures marginal consistency.
Structural consistency is then automatically achieved due to the link localization
enforced in Step~\ref{step:shadow_replacement_step}, which guarantees the link sets
of $\scrE'$ and $\scrE''$ agree on their shared edges $E^*(e)$.

\end{document}